%% file: scaling-main.tex
\documentclass[11pt,a4paper]{article}

\usepackage{lmodern}
\usepackage{amsmath,amsfonts,amssymb,amsthm}
\usepackage{amsthm}
\usepackage{xfrac}
\usepackage{graphicx,color}
\usepackage{boxedminipage}
\usepackage[ruled,boxed,linesnumbered]{algorithm2e}
\usepackage{framed}
\usepackage{enumitem}
\usepackage{thmtools}
\usepackage{thm-restate}
\usepackage{xspace}
\usepackage{bm}
\usepackage{todonotes}
\usepackage{footnote}
\usepackage{mathtools}
\usepackage{caption}
\usepackage{float}
\usepackage{wrapfig}
\usepackage{mathrsfs}
\usepackage{vmargin}
\usepackage{array}
\newcolumntype{P}[1]{>{\centering\arraybackslash}p{#1}}
\setmarginsrb{2.3cm}{2.3cm}{2.3cm}{2.3cm}{0pt}{0pt}{0pt}{6mm}
\usepackage{todonotes}
\usepackage[most]{tcolorbox}
\usepackage{footnote}
\usepackage{slantsc}
\usepackage{multirow}
\usepackage{rotating}
 \usepackage[pdftex, plainpages = false, pdfpagelabels, 
                 bookmarks=true,
                 bookmarksopen = true,
                 bookmarksnumbered = true,
                 breaklinks = true,
                 linktocpage,
                 pagebackref,
                 colorlinks = true,  
                 linkcolor = blue,
                 urlcolor  = blue,
                 citecolor = red,
                 anchorcolor = green,
                 hyperindex = true,
                 hyperfigures
                 ]{hyperref} 
 \usepackage{xifthen}
 \usepackage{tabularx}
\usepackage{tikz} 
 \usetikzlibrary{calc}
 \usepackage{thm-autoref}
\usepackage[nameinlink]{cleveref}

\usepackage{soul}

\newtheorem{theorem}{Theorem}
\newtheorem{lemma}{Lemma}
\newtheorem{claim}{Claim}
\newtheorem{corollary}{Corollary}
\newtheorem{definition}{Definition}
\newtheorem{observation}{Observation}
\newtheorem{proposition}{Proposition}
\newtheorem{redrule}{Reduction Rule}

\theoremstyle{definition}

\newcommand{\poly}{{\sf poly}}
\renewcommand{\epsilon}{\varepsilon}

\DeclareMathOperator{\operatorClassNP}{NP}
\newcommand{\classNP}{\ensuremath{\operatorClassNP}\xspace}
\DeclareMathOperator{\operatorClassNPC}{NP-complete\xspace}
\newcommand{\classNPC}{\ensuremath{\operatorClassNPC}}
\DeclareMathOperator{\operatorClassNPH}{NP-hard\xspace}
\newcommand{\classNPH}{\ensuremath{\operatorClassNPH}}
\DeclareMathOperator{\operatorClassCoNP}{coNP}
\newcommand{\classCoNP}{\ensuremath{\operatorClassCoNP}}
\DeclareMathOperator{\operatorClassFPT}{FPT\xspace}
\newcommand{\classFPT}{\ensuremath{\operatorClassFPT}\xspace}
\DeclareMathOperator{\operatorClassW}{W}
\newcommand{\classW}[1]{\ensuremath{\operatorClassW[#1]}}

\newlength{\RoundedBoxWidth}
\newsavebox{\GrayRoundedBox}
\newenvironment{GrayBox}[1]%
   {\setlength{\RoundedBoxWidth}{.93\textwidth}
    \def\boxheading{#1}
    \begin{lrbox}{\GrayRoundedBox}
       \begin{minipage}{\RoundedBoxWidth}}%
   {   \end{minipage}
    \end{lrbox}
    \begin{center}
    \begin{tikzpicture}%
       \node(Text)[draw=black!20,fill=white,rounded corners,%
             inner sep=2ex,text width=\RoundedBoxWidth]%
             {\usebox{\GrayRoundedBox}};
        \coordinate(x) at (current bounding box.north west);
        \node [draw=white,rectangle,inner sep=3pt,anchor=north west,fill=white] 
        at ($(x)+(6pt,.75em)$) {\boxheading};
    \end{tikzpicture}
    \end{center}}

\newenvironment{defproblemx}[2][]{\noindent\ignorespaces%
                                \FrameSep=6pt%
                                \parindent=0pt%
                \vspace*{-1.5em}
                \ifthenelse{\isempty{#1}}{%
                  \begin{GrayBox}{\textsc{#2}}%
                }{%
                  \begin{GrayBox}{\textsc{#2} parameterized by~{#1}}%
                }
                \begin{tabular*}{\textwidth}{@{\hspace{.1em}} >{\itshape} p{1.8cm} p{0.8\textwidth} @{}}%
            }{
                \end{tabular*}%
                \end{GrayBox}%
                \ignorespacesafterend
            }  

\newenvironment{defproblemxb}[2][]{\noindent\ignorespaces
  \FrameSep=6pt%
  \parindent=0pt%
  \vspace*{-1.5em}
  \ifthenelse{\isempty{#1}}{%
    \begin{GrayBox}{\textsc{#2}}%
    }{%
      \begin{GrayBox}{\textsc{#2} parameterized by~{#1}}%
      }
      \begin{tabular*}{\textwidth}{@{\hspace{.1em}} >{\itshape} p{1.2cm} p{0.85\textwidth} @{}}%
      }{
      \end{tabular*}%
    \end{GrayBox}%
    \ignorespacesafterend
  }

\newif\iffvs
\fvstrue
\newcommand{\cO}{\mathcal{O}}
\newcommand{\Oh}{\mathcal{O}}

\newcommand{\lr}[1]{\left( #1\right)}
\newcommand{\LR}[1]{\left\{ #1\right\}}

\newcommand{\cost}{{\sf cost}}

\newcommand{\dstart}{d_{\sf start}}
\newcommand{\sfstart}{{\sf start}}
\newcommand{\sfend}{{\sf end}}
\newcommand{\sfmax}{{\sf max}}
\newcommand{\dend}{d_{\sf end}}
\newcommand{\dmax}{d_{\sf max}}

\newcommand{\cI}{\mathcal{I}}

\newcommand{\cS}{\mathcal{S}}

\newcommand{\cD}{\mathcal{D}}

\newcommand{\unitvec}{\bm{1}}
\newcommand{\real}{\mathbb{R}}
\newcommand{\realplus}{\mathbb{R}_{\ge 0}}

\newcommand{\pname}[1]{{\sc #1}}
\newcommand{\ProblemFormat}[1]{\pname{#1}}
\newcommand{\ProblemIndex}[1]{\index{problem!\ProblemFormat{#1}}}
\newcommand{\ProblemName}[1]{\ProblemFormat{#1}\ProblemIndex{#1}{}\xspace}

\newcommand{\probConnectivityFracUp}{\ProblemName{$k$-Expanding to Connectivity}}
\newcommand{\probConnectivityFracDown}{\ProblemName{$k$-Shrinking to Connectivity}}

\newcommand{\probAcyclicity}{\ProblemName{Min Shrinking to Acyclicity}}
\newcommand{\probEdgelessnessDown}{\ProblemName{$k$-Shrinking to Independence}}
\newcommand{\probAcyclicityDown}{\ProblemName{$k$-Shrinking to Acyclicity}}
\newcommand{\probEdgelessnessMinDown}{\ProblemName{Min $k$-Shrinking to Independence}}
\newcommand{\probAcyclicityMinDown}{\ProblemName{Min $k$-Shrinking to Acyclicity}}

\newcommand{\probConnGen}{{\sc Generalized $k$-Expanding to Connectivity}}

\begin{document}


\title{Parameterized Geometric Graph Modification with Disk Scaling\footnote{F.~V.~Fomin and P.~A.~Golovach acknowledge support from the Research Council of Norway via the project BWCA (grant no. 314528). T.~Inamdar acknowledges support from  IITJ Research Initiation Grant (grant number I/RIG/TNI/20240072) and the European Research Council (ERC) under the European Union’s Horizon 2020 research and innovation programme. S. Saurabh acknowledges support from the European Research Council (ERC) under the European Union’s Horizon 2020 research and innovation programme (grant agreement No. 819416) Swarnajayanti Fellowship (No. DST/SJF/MSA01/2017-18). M.Zehavi acknowledges support by the European Research Council (ERC) grant no. 101039913 (PARAPATH).}}

\author{
Fedor V. Fomin\thanks{
Department of Informatics, University of Bergen, Norway.}
\and
Petr A. Golovach\addtocounter{footnote}{-1}\footnotemark{}
\and
Tanmay Inamdar\thanks{Indian Institute of Technology Jodhpur, Jodhpur.}
\and
Saket Saurabh \thanks{The Institute of Mathematical Science, HBNI, Chennai, India, and University of Bergen. }
\and
Meirav Zehavi~\thanks{Ben-Gurion University of the Negev, Beer-Sheva, Israel} 
}

\date{}

\maketitle
\thispagestyle{empty}
\begin{abstract}
The parameterized analysis of graph modification problems represents the most extensively studied area within Parameterized Complexity. Given a graph $G$ and an integer $k\in\mathbb{N}$ as input, the goal is to determine whether we can perform at most $k$ operations on $G$ to transform it into a graph belonging to a specified graph class $\mathcal{F}$. Typical operations are combinatorial and include vertex deletions and edge deletions, insertions, and contractions. However, in many real-world scenarios, when the input graph is constrained to be a geometric intersection graph, the modification of the graph is influenced by changes in the geometric properties of the underlying objects themselves, rather than by combinatorial modifications. It raises the question of whether vertex deletions or adjacency modifications are necessarily the \emph{most appropriate} modification operations for studying modifications of geometric graphs.

We propose the study of the disk intersection graph modification through the scaling of disks. This operation is typical in the realm of topology control but has not yet been explored in the context of Parameterized Complexity.
We design parameterized algorithms and kernels for modifying to the most basic graph classes: edgeless, connected, and acyclic. Our technical contributions encompass a novel combination of linear programming, branching, and kernelization techniques, along with a fresh application of bidimensionality theory to analyze the area covered by disks, which may have broader applicability.

\end{abstract}

\section{Introduction}\label{sec:intro}
\setcounter{page}{1}
In a {\em graph modification problem}, the input consists of an $n$-vertex graph $G$ and an integer $k$. The objective is to determine whether $k$ {\em modification operations}---such as vertex deletions, or edge deletions, insertions or contractions, or combinations thereof---are sufficient to obtain a graph with prescribed structural properties such as being planar, bipartite, chordal, interval, acyclic or edgeless.
Graph modification problems encompass some of the fundamental challenges in graph theory and graph algorithms. 
They represent the most intensively studied area of research within Parameterized Complexity~\cite{downey2013fundamentals,DBLP:books/sp/CyganFKLMPPS15,fomin2019kernelization,DBLP:journals/talg/Thomasse10,bliznets2016lower,cao2016linear,cao2017minimum,bliznets2018subexponential,DBLP:journals/talg/AgrawalLMSZ19,jansen2018approximation,DBLP:conf/soda/AgrawalM0Z19,DBLP:journals/talg/FominLLSTZ19,DBLP:conf/icalp/BougeretJS20,DBLP:conf/soda/LiN20,DBLP:journals/corr/abs-1911-02653} (this list is illustrative rather than comprehensive; see \cite{CrespelleDFG23,downey2013fundamentals,DBLP:books/sp/CyganFKLMPPS15,fomin2019kernelization} for more information).  Here, the number of allowed modifications, $k$, is considered a {\em parameter}. With respect to $k$, we seek a {\em fixed parameter tractable }(\classFPT) algorithm, namely, an algorithm whose running time has the form $f(k)n^{\cO(1)}$ for some computable function $f$. 
It would not be an exaggeration to state that the design of \classFPT  algorithms for graph modification problems has played one of the most central roles in the development of Parameterized Complexity as a field.

In recent years, parametrized analysis of classic graph problems restricted to geometric intersection graphs---that is, where the input graph $G$ belongs to a prespecified class of geometric intersection graphs---has gained a lot of interest. Naturally, particular attention has been given to graph modification problems restricted to geometric intersection graphs. For example, we refer to~\cite{alber2004geometric,DBLP:conf/esa/Marx05,DBLP:conf/iwpec/Marx06,giannopoulos2008parameterized,jansen2010polynomial,DBLP:conf/soda/FominLS12,DBLP:conf/stoc/BergBKMZ18,DBLP:journals/dcg/FominLPSZ19,DBLP:conf/soda/Panolan0Z19,DBLP:conf/icalp/FominLP0Z19,DBLP:journals/tcs/BergKW19,DBLP:journals/algorithmica/BonnetR19,DBLP:conf/compgeom/FominLP0Z20}. While most classic \classNP-hard graph problems become polynomial-time solvable on proper interval, interval and circular arc graphs, they remain \classNP-hard\ on most other geometric intersection graph classes of interest, which gives rise to their parameterized analysis (see the aforementioned papers). 
Geometric intersection graph classes studied in this context include, for example, unit disk and disk graphs, unit square and square graphs, convex polygon  and polygon graphs, map graphs, string graphs, and the generalization of these classes to higher dimensions.
However, one central question arises:

\begin{tcolorbox}[colback=gray!5!white,colframe=gray!75!black]
	Are vertex deletion and edge deletion/insertion/contraction the {\em right} modification operations to use when dealing with graph modification problems on geometric intersection graphs?
\end{tcolorbox}

Geometric intersection graphs naturally give rise to modification operations that are {\em geometric}---specifically, we modify the graph by modifying its geometric structure! Arguably, the most fundamental operations in this regard are {\em scaling} and {\em shifting (i.e., moving)} the input geometric~objects.

\begin{tcolorbox}[colback=green!5!white,colframe=gray!75!black]
	In this paper, our primary objective is to launch a systematic investigation into parameterized graph modification problems concerning geometric intersection graphs, utilizing the operation of scaling the input geometric structure.
	We regard a significant aspect of our contribution as both conceptual and foundational, with the goal of paving the way for this emerging research direction.
\end{tcolorbox}

\subsection{Topology Control and Disk Scaling as a Modification Operation}
Topology Control is a fundamental technique used in wireless ad hoc network and sensor networks to reduce energy consumption and radio interference \cite{santi2005topology}. This technique is used to control the topology of the graph representing the communication links between network nodes.  The usual goal is to maintain or achieve some global graph property (like being connected), or reduce interference (i.e., obtain an edgeless graph) while reducing energy consumption.  The geometric operation of scaling---i.e., shrinking or expanding radii---of disks has received significant attention within research of wireless network design due to its relevance to power consumption and interference in the networks. From a theoretical perspective, there is also a large body of literature on algorithmic methods for range assignment in wireless sensor and ad-hoc networks, see \cite{lloyd2005algorithmic,ChambersFHMMVSW18,calinescu2003range,caragiannis2006energy,wan2002minimum}
for a (very incomplete) sample of various approaches, and it is also of interest for other purposes (see, e.g., \cite{acharyya2020range,eppstein2016maximizing,biniaz2018faster}). Due to this significance, we consider problems that use disk scaling (i.e., shrinking or expanding a subset of disks) in order to achieve certain topological properties in graph. Next, we formally define the problems considered in the paper.

\paragraph*{Recent work on geometric graph modification.} Recently, Fomin et al.~\cite{FominG0Z22,FominG0Z23} have studied some problems related to (re)packing of unit disks and point spreading from the perspective of parameterized complexity. Taking a broader perspective, these problems can also be seen as working in the model of geometric graph modification via \emph{geometric operations}, also studied in \cite{ChambersFHMMVSW18,AloupisDFKORW13,CarmiK07}. Indeed, these papers consider \emph{movement} of points or disks as an operation to achieve intersection graph that is edgeless. In this paper, we focus on \emph{scaling} (i.e., shrinking or expanding) as the operation, and in addition to edgelessness, we also consider acyclicity and connectivity. Another common theme with the work of \cite{FominG0Z23} is that they also combine enumeration of partial solutions from the (partial) kernel, and then reduce the problem to checking a system of polynomial inequalities. This high level scheme is similar to the one described for designing FPT algorithms for \probEdgelessnessMinDown/\textsc{Acyclicity}, where we use an LP. But this similarity is only superficial, and seems unavoidable due to the inherently \emph{continuous nature} of the problems. The specific context, and in particular, the kind of arguments used to arrive at a point where we can use LP, are quite involved, and different from that in \cite{FominG0Z23}, as we will explain this further in the following section. Nevertheless, we hope that using geometric operations for graph modification opens up a new research direction, and we return to this topic in the conclusion. 

\subsection{Problem Definitions}
We fix some notation. Let $P$ be a finite family of points in the plane, and $r: P \to \real_{\ge 0}$ be an assignment of points to radii. Then, the disk graph associated with $P$ and $R$ is defined as follows. Let $\cD(P,r)$ be a set that contains the disk centered at $p$ with radius $r(p)$ for each $p \in P$, and let $G(P,r)$ denote the disk graph corresponding to $\cD(P,r)$, i.e., a graph where we have a vertex for each point in $P$, and we have an edge between two points if and only if their corresponding disks intersect. Let $\unitvec: P \to \{1\}$ denote the unit radius assignment, and let $G(P, \unitvec)$ denote the \emph{unit disk graph} corresponding to a set of points $P$.

\subsubsection{Shrinking to Achieve Edgelessness and Acyclicity} \label{subsec:intro-subset}

The first set of problems is related to shrinking a subset of disks in order to achieve 
two fundamental graph properties, namely, that of being edgeless and acyclic. 
The general setting is as follows. We are given a unit disk graph $G(P, \unitvec)$ corresponding to a set of $n$ points $P \subset \real^2$. We want to \emph{shrink} the disks corresponding to a subset $S \subseteq P$ (i.e., change the radii of disks of each point in $S$ to a value smaller than $1$), such that the resulting intersection graph induces an independent set (i.e., is edgeless). We consider two variants of this setting. First, we have a cardinality constrained problem, called \probEdgelessnessDown (see \Cref{fig:indep}).
\begin{figure}[t]
	\centering
	\includegraphics[scale=0.8,page=2]{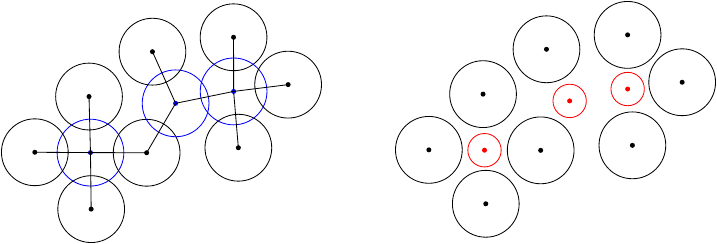}
	\captionsetup{font=small}
	\caption{Left: Original intersection graph $G = G(P, \unitvec)$ denoting an instance $(P, \sfrac{1}{2}, 3)$ of \probEdgelessnessDown. Right: Edgeless graph $G(P, r)$ that shrinks three disks to radius $\sfrac{1}{2}$ (shown in red).} 
	\label{fig:indep}
\end{figure}
\begin{tcolorbox}[colback=white!5!white,colframe=gray!75!black]
	\probEdgelessnessDown
	\begin{description}
		\item[Input:]  A finite set of points $P \subset \real^2$, a constant $0 \le \alpha \le 1$, and an integer $k \ge 0$.
		\item[Task:] Decide whether there exists a solution $(S, r)$ such that
		\begin{enumerate}[leftmargin=*]
			\item $S \subseteq P$ is a set of size \textbf{at most} $k$ corresponding to shrunken disks,
			\item \vspace{-0.15cm}corresponding radii $r: P \to \real_{\ge 0}$ such that $r(p) = \begin{cases}
				1 &\text{ if } p \not\in S
				\\\alpha &\text{ if } p \in S
			\end{cases}$
			\item \vspace{-0.4cm}$G(P, r)$ is edgeless
		\end{enumerate}
	\end{description}
\end{tcolorbox}
We also consider a cost-minimization version of the problem, called \probEdgelessnessMinDown, where the cost of a solution is the total change in the radii, or the sum of \emph{shrinking factors}.
\begin{tcolorbox}[colback=white!5!white,colframe=gray!75!black]
	\probEdgelessnessMinDown
	\begin{description}
		\item[Input:]  A finite set of points $P \subset \real^2$, a constant $0 \le \alpha \le 1$, an integer $k \ge 0$, and budget $\mu \ge 0$.
		\item[Task:] Decide whether there exists a solution $(S, r)$ such that
		\begin{enumerate}[leftmargin=*]
			\item $S \subseteq P$ is a set of size \textbf{at most} $k$ corresponding to shrunken disks,
			\item corresponding radii $r: P \to \realplus$ where $\begin{cases}
				r(p) = 1 &\text{ if } p \not\in S
				\\\alpha \le r(p) \le 1 &\text{ if } p \in S
			\end{cases}$
			\item \vspace{-0.2cm}$G(P, r)$ is edgeless, and
			\item $\displaystyle \cost(S, r) \coloneqq \sum_{p \in S} (1-r(p)) \le \mu$.
		\end{enumerate}
	\end{description}
\end{tcolorbox}

\begin{figure}
	\centering
	\includegraphics[scale=0.7,page=3]{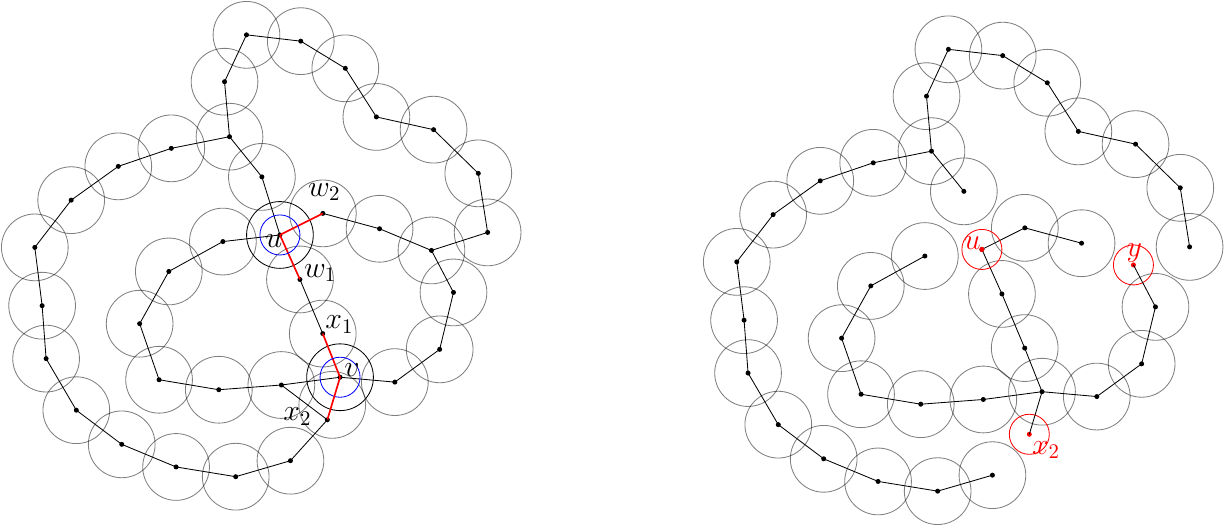}
	\captionsetup{font=small}
	\caption{Left: Original intersection graph $G = G(P, \unitvec)$ corresponding to an instance $(P, \sfrac{1}{2}, 3)$ of \probAcyclicityDown. Note that $\LR{u, v}$ is a minimum feedback vertex set for $G$. However, even after shrinking the corresponding disks to $\sfrac{1}{2}$ (shown in blue), the edges $\LR{uw_1, uw_2, vx_1, v_2}$ (shown in red) are present in the resulting intersection graph, showing that shrinking disks is not equivalent to vertex deletion. Right: A solution that shrinks three disks (shown in red) to radius $\sfrac{1}{2}$ corresponding to $u, x_2, y$, resulting in an acyclic graph $G(P, r)$.} 
	\label{fig:acyclic}
\end{figure}

Analogous to \probEdgelessnessDown, we define \probAcyclicityDown, where the goal is to achieve $G(P, r)$ that is \emph{acyclic} instead of edgeless (as in item 3 in the definition) (see \Cref{fig:acyclic} for an example). \probAcyclicityMinDown is the ``cost-minimization'' version (similar to \probEdgelessnessMinDown), where we want to determine whether we can shrink at most $k$ disks to radii in $[\alpha, 1]$ in order to achieve an acyclic subgraph, such that $\cost(S, r) \le \mu$. 

Note that the operation of shrinking a disk to radius $0$ is not necessarily equivalent to deleting the corresponding vertex. Thus, even when $\alpha = 0$, \probAcyclicityDown (resp.~\probEdgelessnessDown) is not exactly equivalent to finding a feedback vertex set (resp.~vertex cover) of size $k$ in the original unit disk graph. Nevertheless, we will make use of this connection between the two problems. Now, we turn to the second class of problems considered in this paper.

\subsubsection{Shrinking/Expanding to Connectivity} \label{subsec:intro-connectivity}
Now we turn to another graph property that is fundamental in topology control, namely connectedness. We consider two variants: shrinking \emph{at least} $k$ disks while \emph{retaining connectivity}, and expanding \emph{at most} $k$ disks to \emph{achieve connectivity}.\footnote{Observe that expansion versions do not make sense for the problems considered in the earlier section, where we want to obtain an edgeless/acyclic graph -- if the graph is not already edgeless/acyclic, one cannot achieve the desired property by expanding a disk.} Next, we define the two problems.

In the first problem, called \probConnectivityFracDown \footnote{Ideally, the problem should be named \textsc{$k$-Shrinking} \textsl{\textsc{While Retaining}} \textsc{Connectivity}. However, for the sake of brevity, we stick with \probConnectivityFracDown.}, we are given a connected unit disk graph $G(P, \unitvec)$ corresponding to a set of $n$ points $P \subset \real^2$. We want to \emph{shrink} the disks corresponding to a subset $S \subseteq P$ of size \emph{at least} $k$ (i.e., change the radii of disks corresponding to $S$ to a value smaller than $1$), such that the resulting intersection graph \emph{remains} connected. 
\begin{tcolorbox}[colback=white!5!white,colframe=gray!75!black]
	\probConnectivityFracDown
	\begin{description}
		\item[Input:]  A finite set of points $P \subset \real^2$, a constant $0 \le \alpha \le 1$, and a parameter $k \ge 0$.
		\item[Task:] Decide whether there exists a solution $(S, r)$ such that
		\begin{itemize}[leftmargin=*]
			\item $S \subseteq P$ is a set of size \textbf{at least} $k$ corresponding to shrunken disks,
			\item \vspace{-0.15cm}corresponding radii $r: P \to \realplus$ such that $r(p) = \begin{cases}
				1 &\text{ if } p \not\in S
				\\\alpha &\text{ if } p \in S
			\end{cases}$
			\item \vspace{-0.4cm}$G(P, r)$ is connected.
		\end{itemize}
	\end{description}
\end{tcolorbox}

\begin{figure}[H]
	\centering
	\includegraphics[scale=0.9,page=5]{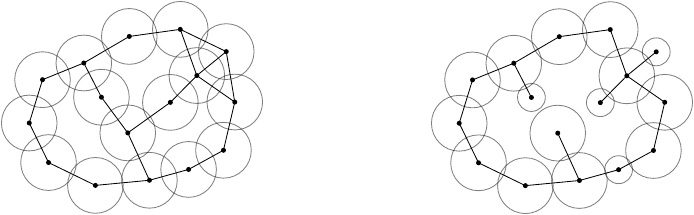}
	\captionsetup{font=small}
	\caption{Left: Original intersection graph $G = G(P, \unitvec)$ corresponding to an instance $(P, \sfrac{1}{2}, 4)$ of \probConnectivityFracDown. Right: A solution that shrinks $4$ disks to radius $\sfrac{1}{2}$ while maintaining connectivity in the resulting intersection graph $G(P, r)$. }
	\label{fig:connex}
\end{figure}

The next problem, which we call \probConnectivityFracUp, is in a sense, the ``complement version'' of \probEdgelessnessDown. Here, we are given a unit disk graph $G(P, \unitvec)$ corresponding to a set of $n$ points $P \subseteq \real^2$. We want to \emph{expand} the disks corresponding to a subset $S \subseteq P$ of size at most $k$ (i.e., change the radii of the disks corresponding to $S$ to a value larger than $1$), such that the resulting intersection graph \emph{becomes} connected. Formally, the problem is defined as follows.
\begin{tcolorbox}[colback=white!5!white,colframe=gray!75!black]
	\probConnectivityFracUp
	\begin{description}
		\item[Input:]  A finite set of points $P \subset \real^2$, a constant $\alpha \ge 1$, and a parameter $k \ge 0$.
		\item[Task:] Decide whether there exists a solution $(S, r)$ such that
		\begin{itemize}[leftmargin=*]
			\item $S \subseteq P$ is a set of size \textbf{at most} $k$ corresponding to expanded disks,
			\item \vspace{-0.15cm}corresponding radii $r: P \to \realplus$ such that $r(p) = \begin{cases}
				1 &\text{ if } p \not\in S
				\\\alpha &\text{ if } p \in S
			\end{cases}$
			\item \vspace{-0.4cm}$G(P, r)$ is connected.
		\end{itemize}
	\end{description}
\end{tcolorbox}

\section{Overview of the Technical Results} \label{sec:overview}

Our algorithmic technical contribution is, mainly, twofold:
\begin{itemize}
	\item The first contribution is a novel combination of linear programming, kernelization and branching/exhaustive search: we observe that linear programming can handle continuous domains, kernelization can reduce the search space within these domains, and, when combined with exhaustive search/branching, these can yield the desired solution. While several of our results are based on this approach, we would like to highlight, in particular, the FPT algorithm for \probAcyclicityDown (see Theorem~\ref{thm:minfvs-fpt}), which exhibits a more intricate combination. 
	\item The second contribution is, in a sense, the first ``geometric'' use of bidimensionality theory. Here, the application is tied up with the analysis of the different areas covered by disks. We employ this approach to design a subexponential FPT algorithm for \probConnectivityFracDown (see Theorem~\ref{thm:subexp-conn}). 
\end{itemize}

We believe that our two approaches are quite general in nature, and will be applicable to other parameterized problems that involve placement of geometric objects in the Euclidean space, particularly when dealing with continuous domains. We summarize our results in \Cref{table:results}, and we elaborate on each of them in the upcoming subsections. One remark is that, all of these problems are shown to be \classNPH~in \Cref{sec:lower}. These \classNP-hardness results are shown via reduction from well-known planar problems, namely, \textsc{Independent Set, Feedback Vertex Set} in planar graphs, and \textsc{Planar 3-SAT}, respectively. Now we turn to our algorithmic results.

\begin{table}
	\renewcommand{\arraystretch}{1.3}
	\centering
	\footnotesize
	\begin{tabular}{|P{4.6cm}|P{2.7cm}|P{3.8cm}|P{3.8cm}|}
		\hline
		Objective &  Type of Result & Cardinality version &  Cardinality and cost-minimization version \\
		\hline
		\multirow{3}{*}{{\sc $k$-Shrinking to} {\sc Independence}}  & Kernelization  & $\Oh(k)$ partial kernel (Thm.~\ref{thm:vc-partialpolykernels-informal}), true polynomial kernel (Thm.~\ref{thm:polykernel-vc}) & $\Oh(k)$ partial kernel (Thm.~\ref{thm:vc-partialpolykernels-informal}) \\\cline{2-4}
		& Parameterized complexity  & Subexponential FPT (Thm.~\ref{thm:vc-subexp}) & FPT (Thm.~\ref{thm:vc-sumk}) \\\cline{2-4}
		& Approximation  & EPTAS (Thm.~\ref{thm:bicriteria}) & Bicriteria EPTAS (Thm.~\ref{thm:bicriteria}) \\\cline{2-4}
		\hline
		\multirow{2}{*}{\sc $k$-Shrinking to Acyclicity}  & Kernelization \& Compression  & true polynomial kernel (Thm.~\ref{thm:fvs-polykernel}) & partial polynomial compression (Thm.~\ref{thm:fvs-partialpolykernels-informal}) \\\cline{2-4} 
		& Parameterized complexity & Subexponential FPT (Thm.~\ref{thm:fvs-subexp}) & FPT (Thm.~\ref{thm:minfvs-fpt}) \\
		\hline
		{\sc $k$-Shrinking to Connectivity} & Parameterized complexity & Subexponential FPT (Thm.~\ref{thm:subexp-conn}) & -- \\
		\hline {\tiny A generalization of}\qquad
		{\sc $k$-Shrinking to Connectivity} & Parameterized complexity & W[1]-hard (Thm.~\ref{thm:genw1hard}) & -- \\
		\hline
	\end{tabular}
	\caption{A summary of the results in the paper. Here, cost constraints refer to the cost-minimzation versions of the respective objectives, as defined before. Additionally, we also show that all of the problems to be NP-hard, which is not mentioned in the table.} \label{table:results}
\end{table}

\subsection{Subexponential FPT Algorithms} \label{subsec:overview-subexp}

In this section, we discuss the first technical contribution of our paper, namely, to design subexponential FPT algorithms for \probEdgelessnessDown, \probAcyclicityDown and \probConnectivityFracDown (specifically the last problem).

\paragraph*{Independence and Acyclicity.}
Recall that in these problems, the goal is to shrink at most $k$ disks to radius $\alpha$ such that the resulting intersection graph is edgeless (resp.~acyclic). For the sake of ease of exposition, let us suppose that $\alpha$, the radius of the shrunken disks, is a constant. We will start by eliminating \emph{large} cliques, i.e., we show using packing arguments that, if the largest clique in the original intersection graph $G = G(P, \unitvec)$ has size $\Omega_{\alpha}(1)$ \footnote{For the sake of simplicity, we use $\alpha$ in the subscript of $\Oh(\cdot)$ and $\Omega(\cdot)$ notation to suppress the dependence on $\alpha$ in the current overview.}, then no solution can shrink any subset of disks to radius $\alpha$ and remove all the edges. This bounds the maximum clique-size. Now, we use \emph{bidimensionality} arguments to show that, either the graph has $\Omega(\sqrt{k}) \times \Omega(\sqrt{k})$ grid as a minor, or the treewidth of the graph is bounded by $\Oh_{\alpha}(\sqrt{k})$. In the former case, we can conclude that it is a no-instance, whereas in the latter case one can perform dynamic programming on the tree decomposition of width $\Oh_{\alpha}(\sqrt{k})$ to solve the problem. This leads to the following results.

\begin{restatable}{theorem}{subexponentialvck} \label{thm:vc-subexp}
	There exists an algorithm that solves an instance $\cI = (P, k, \alpha)$ of \probEdgelessnessDown in time $2^{\Oh((\frac{1}{\alpha})^2\sqrt{k})}\cdot n^{\Oh(1)}$.
\end{restatable} 

\begin{restatable}{theorem}{subexponentialfvsk} \label{thm:fvs-subexp}
	There exists an algorithm that solves an instance $\cI = (P, k, \alpha)$ of \probAcyclicityDown in time $(\frac{1}{\alpha}k)^{\Oh((\frac{1}{\alpha})^2\sqrt{k})}\cdot n^{\Oh(1)}$.
\end{restatable}

\paragraph*{Connectivity}
Now we turn to the subexponential FPT algorithm for \probConnectivityFracDown, which is one of the main results of the paper. This algorithm also begins with the elimination of large cliques. However, already for this step, our arguments are substantially more involved 
\begin{wrapfigure}{r}{4.5cm}
	\centering
	\includegraphics[scale=0.8,page=1]{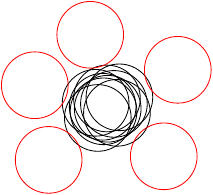}
	\captionsetup{font=footnotesize}
	\caption{There is a solution that shrinks some of the disks in the clique; however, shrinking \emph{all} the disks in the clique results in the disconnected graph (due to the red disks)}\label{fig:connclique}
\end{wrapfigure} 
compared to all previously known subexponential-time algorithms for problems on geometric intersection graphs that also eliminate large cliques 
(see, e.g., \cite{DBLP:conf/soda/FominLS12,DBLP:journals/dcg/FominLPSZ19,DBLP:conf/icalp/FominLP0Z19,DBLP:conf/compgeom/FominLP0Z20,DBLP:conf/soda/LokshtanovPSXZ22}). Specifically, for all of these problems, the existence of a large clique trivially implies that we have a yes-instance (or a no-instance) at hand. 
However,  for \probConnectivityFracDown, the proof is more complicated. Further, and perhaps more importantly, even once we have proved that every optimal solution contains ``most'' vertices of a clique, and we have ``guessed'' (by branching) which these vertices are, we can neither shrink their corresponding disks (since then, we might be left with a clique just as large as it was before) nor delete them. In particular, deletion is incorrect since it might turn a no-instance into a yes-instance, and vice versa (see \Cref{fig:connclique}). Fortunately, we are able to prove that it is possible to delete all of the guessed set except for a carefully chosen subset of a few representatives from it.

The next step of our proof handles cycles (or, more precisely, closed walks) that we term {\em non-empty}. Specifically,  we denote some disks as unshrinkable; at the beginning, these are simply the disks that cannot be shrunk without violating connectivity. Then, for a cycle, we consider  the ``area that it encloses'', and identify two sets of disks of interest: the disks intersecting this area (the {\em relevant set}), and the disks strictly in the interior of this area (the {\em interior set}). If all the disks in the relevant set are unshrinkable, then the cycle is termed unshrinkable, and if the interior set is non-empty, then the cycle is termed non-empty. Then, we prove that, given an unshrinkable cycle, it is safe to delete its entire interior set (which, in case of a non-empty cycle, yields progress). Additionally, if the interior set of a cycle is non-empty and contains a shrinkable disk, then the cycle is called {\em reducible}; we note that a shrinkable non-empty cycle might {\em not} be reducible. We prove that if we have a large enough collection of reducible cycles whose relevant sets are disjoint, then we have a yes-instance at hand.

Thirdly, we turn to deal with the case where we have a large grid as minor---here, it will also become clear why we had to analyze non-empty cycles. 
Again, let us remark that all previously known subexponential-time algorithms for problems on geometric intersection graphs, when encountering the case of a large grid as a minor, can simply say Yes or No. For us, having a large grid minor does not imply that we have a yes-instance (or a no-instance), although it may seem so at first glance---indeed, we might even be dealing with an ``extreme'' case where all of the disks that compose the grid minor are unshrinkable (see \Cref{fig:gridminor}). 
To overcome this difficulty, first of all, we do not consider just any grid minor, but a grid minor such that the closed walks corresponding 
\begin{wrapfigure}{r}{5.9cm}
	\centering
	\vspace{-0.2cm}
	\includegraphics[scale=0.58,page=6]{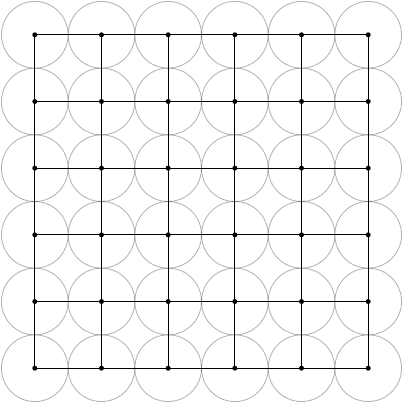}
	\captionsetup{font=footnotesize}
	\caption{Each disk corresponding to the vertices of the grid minor is ``touching'' its neighbors, and shrinking any of the disks results in a disconnected graph.}\label{fig:gridminor}
\end{wrapfigure} 
(according to the minor model) 
to its $4$-cycles have some particular embedding on the plane. 
Then, we are able to show that within such a grid minor, one can identify either an unshrinkable non-empty closed walk (and, then, apply a reduction rule) or  a large enough collection of reducible cycles whose relevant sets are disjoint (and, then, output Yes and terminate). Putting everything together, our algorithm works as follows. As long as we have a large clique, we branch to eliminate most of it. Then, based on a known result, we know that either we have a large grid minor of a particular form (and, then, proceed as described in the previous paragraph) or that the treewidth of the graph is small. In this latter case, we can simply use dynamic programming to directly solve the problem. Thus, we obtain the following result, which is discussed in \Cref{subsec:fpt-subexp-conn}.
\vspace{0.5cm}
\begin{restatable}{theorem}{subexpconnectivity} \label{thm:subexp-conn} 
	There exists an algorithm that solves an instance $\cI = (P, k, \alpha)$ of	\probConnectivityFracDown in time $(\frac{k}{\alpha})^{\Oh((\frac{1}{\alpha})^2 \cdot k^{3/4})}\cdot n^{\Oh(1)}$.
\end{restatable}

\paragraph{Expanding to achieve Connectivity.} Surprisingly, it turns out that, unlike its shrinking counterpart, a generalization of \probConnectivityFracUp is \classW{1}-hard. In this generalization, we may only expand the disks centered at a specified subset of points; whereas the disks centered at the rest of the points are fixed to be unit disks. We show this using a reduction from \textsc{Covering Points by Unit Disks}, which is known to be \classW{1}-hard \cite{Marx07a}. This result is described in \Cref{subsec:w1-hard-conn}. This result shows that not all natural problems of interest in this geometric graph modification model are fixed-parameter tractable. \footnote{We remark that our algorithmic results (FPT algorithms, kernels, etc.) for the other problems---namely, \probEdgelessnessDown/\textsc{Acyclicity}/\textsc{Connectivity}and their variants---also hold for such a generalization, where a subset of disks cannot be shrunk. In fact, our algorithm for \probEdgelessnessDown supposes the presence of unshrinkable disks, as there, specifically, it simplifies the write-up. Thus, the preceding remark remains valid. However, for the sake of simplicity of presentation of the other results, we decided to stick with the special cases where any disks can be shrunk.}

\subsection{Polynomial Kernels and Compressions and their Consequences} \label{subsec:overview-kernels}
Now we discuss our results regarding polynomial kernels and compressions for \probEdgelessnessDown/\textsc{Acyclicity} and their minimization variants. We also discuss how these results also enable us to design FPT algorithms for the minimization variants, using a combination of partial enumeration of the solution space of the kernels, and linear programming. Note that the subexponential FPT algorithms for \probEdgelessnessDown/\textsc{Acyclicity} discussed in the previous section work directly on the original instances and do not rely on the kernels.

\subsubsection{{\sc (Min) $k$-Shrinking to Independence}} \label{subsubsec:vc-kernels}

We show the following result in \Cref{subsec:vc-partialkernel}.
\begin{restatable}{theorem}{vcpartialpolykernels}{\rm [Informal]} \label{thm:vc-partialpolykernels-informal}
	There exists a partial kernel for \probEdgelessnessDown (resp.\ for \probEdgelessnessMinDown), where the number of points in the resulting instance is $\Oh(k)$.
\end{restatable}
Note that in this result, we are only able to bound the number of points in the resulting instance by $\Oh(k)$. However, the number of bits required to encode the coordinates of each point may still be unbounded. Therefore, we only obtain a \emph{partial kernel} for these two problems. 

Since \probEdgelessnessDown is a ``discrete'' problem, we can in fact extend this argument to obtain a (fully) polynomial kernel. Here, we need to additionally show that the geometric information that is relevant for solving \probEdgelessnessDown can be encoded using polynomially many bits in $k$. Thus, we obtain the following result, which we prove in \Cref{subsec:polykernel-vc}.
\begin{restatable}{theorem}{vcpolykernel} \label{thm:polykernel-vc}
	There exists a (true) polynomial kernel for \probEdgelessnessDown, paramterized by $k$.
\end{restatable}

\paragraph{FPT Algorithm for \probEdgelessnessMinDown.} Note that usually the existence of (partial) kernels is sufficient to immediately conclude that the problem is FPT by enumerating all possible solutions. However, for the case of \probEdgelessnessMinDown, even if one fixes a subset of disks of size $k$ to be shrunk (i.e., a \emph{partial solution}), the new radius of each disk is a real number in the range $[\alpha, 1]$. Thus, the number of possible solutions is infinite. Here, we enumerate all partial solutions, and use \emph{linear programming} \footnote{Since our LP involves real numbers, we cannot use standard LP solvers that run in (weakly) polynomial time. However, since we work in Real RAM model, we can use an algebraic LP solver such as \cite{ChanTALG18,MatousekSW96} to find an optimum solution, resulting in the algorithm with the claimed running time bound. More on this in the preliminaries and technical sections.} to check whether there exists a radius assignment of cost at most $\mu$, resulting in an edgeless graph. Thus, we obtain the following result in \Cref{subsec:fpt-vc-ksum}.
\begin{restatable}{theorem}{vcsumk} \label{thm:vc-sumk}
	There exists an algorithm that decides an instance $\cI = (P, k, \alpha, \mu)$ of \probEdgelessnessMinDown in time $k^{\Oh(k)}\cdot n^{\Oh(1)}$. Additionally, if $\cI$ is a yes-instance, the algorithm can return an \emph{optimal solution} for $\cI$.
\end{restatable}

\subsubsection{{\sc (Min) $k$-Shrinking to Acyclicity}} \label{subsec:overview-acyclic}

We first give an overview of our result regarding the (partial) polynomial compressibility of \probAcyclicityMinDown, which, along with the FPT algorithm for the problem (as stated subsequently) is the second main result of the paper. 

\begin{restatable}{theorem}{fvspartialpolykernel}{\rm [Informal]}  \label{thm:fvs-partialpolykernels-informal}
	There exists a partial polynomial compression for \probAcyclicityMinDown into an annotated graph problem, where the size of the resulting instance is bounded by a polynomial in $k$, assuming each real number in the instance can be stored using $\Oh(1)$ bits. 
\end{restatable}

We first state relatively straightforward reduction rules. First, we argue that, if any vertex in the initial intersection graph has a ``large'' degree (i.e., larger than some $ck$), then it implies the existence of a clique of size at least $k+3$ -- note that no solution that is allowed to shrink at most $k$ disks can remove all the edges of such a clique, which implies that we have a no-instance. Thus, either the maximum degree in the given intersection graph is bounded by $\Oh(k)$, or we can say No. Further, we can safely delete all vertices of degree at most one as well as acyclic components from the graph. So far, these arguments bear resemblance to those used to obtain a polynomial kernel for \textsc{Feedback Vertex Set}. From this point onward, however, we deviate significantly from the standard arguments, and need to rely on the inherent geometric nature of the problem.

\begin{figure}[t]
	\centering
	\includegraphics[scale=0.7,page=9]{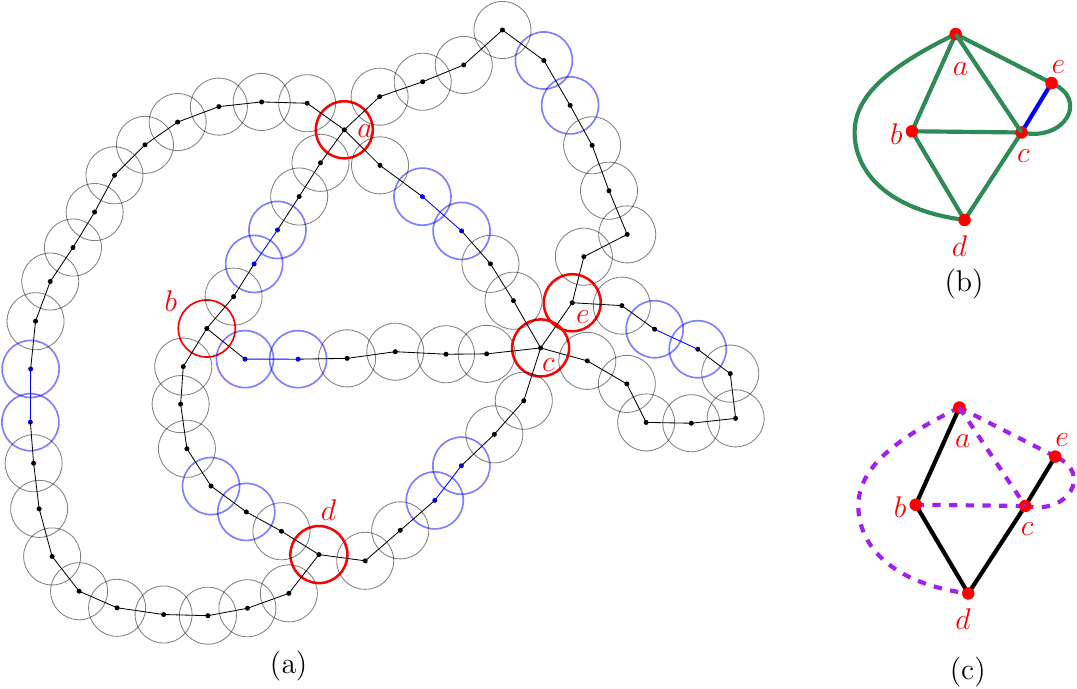}
	\captionsetup{font=small}
	\caption{(a) Given intersection graph after removing the vertices of degree at most $1$. Vertices of degree at least $3$ are shown in red. (b) Resulting multigraph $H$ on vertices of degree at least $3$, obtained after replacing degree-$2$ paths with reducible edges (shown in green) and one of the two edges $ce$ is an original edge (blue). Annotations of green edges include, for example, distance between farthest pairs of disks in the corresponding path (shown in blue in figure (a)). (c) A particular guess for the edges of $R$ (dashed purple), such that $H - R$ (graph induced on black edges) is acyclic. For each edge of $R$, we guess the number of disks that will be shrunk to remove this edge. Then, an optimal values of the shrunk radii that is compatible with this guess is found using linear programming. For example, such a solution may shrink the disk centered at $c$ (removing $bc, ac$), $e$  (removing $ae$ and together with $c$, the edge $ce$), and the two blue disks in the path corresponding to $ad$.}\label{fig:fvscompression}
\end{figure}

We now handle each connected component of the graph separately. If a connected component is a simple cycle, then observe that an optimal solution shrinks either one or two disks, depending on the distances between the consecutive vertices in the cycle. We can compute the optimal choice and get rid of all such isolated cycles. Next, we aim to handle ``long'' degree-$2$ paths in a component. Note that an optimal solution shrinks $\ell \le 2$ disks from such a path. We replace such a path in the intersection graph with a single edge, called a \emph{reducible edge}, that is annotated with an optimal choice for shrinking $0 \le \ell \le 2$ disks from the original path. The remaining edges in the graph are \emph{original} edges, and we annotate each edge with the Euclidean distance between its two endpoints. 

Thus, we obtain a graph $H$ (in fact, in the actual construction it will be a multigraph). Further, the minimum degree of $H$ is $3$, and the maximum degree of $H$ is $\Oh(k)$. At this point, we can argue that, if the given instance is a yes-instance, then $|V(H)| = \Oh(k^2)$ and $|E(H)| = \Oh(k^3)$. The task of determining whether the original instance is a yes-instance of \probAcyclicityMinDown is equivalent to determining whether $H$ is a yes-instance of an appropriately defined graph problem, say $\Pi$. Note that although the size of $H$ is bounded by $\Oh(k^3)$, the edges of $H$ are annotated with relevant distances denoting optimal choices, and thus can be real numbers in general. Thus, we can only infer a (partial) polynomial compression from \probAcyclicityMinDown into $\Pi$.

We can use the compressed instance of $\Pi$ to design an FPT algorithm for \probAcyclicityMinDown, as follows. Recall that $|V(H)| = \Oh(k^2)$ and $|E(H)| = \Oh(k^3)$. Thus, any acyclic subgraph of $H$ with vertex set $V(H)$ must contain at most $|V(H)| - 1 = \Oh(k^2)$ edges. We iterate over all such edge sets $F \subseteq E(H)$ inducing an acyclic subgraph -- note that there are at most $k^{\Oh(k^2)}$ guesses. This implies that each edge in $R = E(H) \setminus F$ must be removed by the solution; note that $|R| = \Oh(k^3)$. Note that an original edge of $R$ can be removed by shrinking either one or both of its endpoints ($3$ choices), whereas a reducible edge can be removed by shrinking $1$ or $2$ disks in the corresponding path, and these choices are encoded in the annotation. Thus, there are a constant number of choices for removal of each edge of $R$, which results in $2^{\Oh(k^3)}$ further guesses. For each set of guesses, we can determine an optimal choice for each such guess using the annotated information, and formulate a linear program (LP) that captures the best choices. By solving the LP, in time $2^{\Oh(k^3 \log k)}$, we can find a solution that shrinks at most $k$ disks, if one exists. Finally, we return the best solution found over all $2^{\Oh(k^3)}$ guesses, which leads to the following result. 

\begin{restatable}{theorem}{minfvsfpt} \label{thm:minfvs-fpt}
	There exists an algorithm that takes as an input of \probAcyclicityMinDown, and in time $2^{\Oh(k^3\log k)} \cdot n^{\Oh(1)}$ time, either correctly concludes that it is a no-instance; or finds a solution of the smallest cost that shrinks at most $k$ disks.
\end{restatable}

The aforementioned approach also yields a polynomial compression for \probAcyclicityDown into an intermediate multigraph problem. However, since \probAcyclicityDown is a ``discrete'' problem, one can, in fact, encode all the relevant geometric information using polynomially bits in $k$, resulting in a polynomial kernel for the problem.

\begin{restatable}{theorem}{fvspolykernel} \label{thm:fvs-polykernel}
	There exists a polynomial kernel for \probAcyclicityDown parameterized by the number of disks $k$ that one is allowed to shrink.
\end{restatable}

\subsection{Approximation Schemes for {\sc (Min) \probEdgelessnessDown}} Finally, we design EPTASes, i.e., efficient polynomial-time approximation schemes for (\textsc{Min}) \probEdgelessnessDown. For \probEdgelessnessDown, for a fixed $\epsilon > 0$, our algorithm either correctly concludes that the given instance is a no-instance, or returns a solution that shrinks at most $(1+\epsilon)k$ disks to radius $\alpha$. For \probEdgelessnessMinDown, our algorithm returns a \emph{bicriteria} EPTAS. That is, it either correctly concludes that the given instance is a no-instance, or returns a solution of cost at most $(1+\epsilon)\mu$ that shrinks at most $(1+\epsilon)k$ disks to a radius in $[\alpha, 1]$. Thus, we obtain the following theorem. These results are described in \Cref{sec:vc-eptas}.
\begin{restatable}{theorem}{vcbicrtieria} \label{thm:bicriteria}
	There exists a bicriteria EPTAS (resp.\ EPTAS) for \probEdgelessnessMinDown (resp.\ \probEdgelessnessDown) running in time $2^{\Oh(\frac{1}{\epsilon} \log(\frac{1}{\epsilon}))} \cdot n^{\Oh(1)}$.
\end{restatable}
These results are obtained using the well-known \emph{shifting technique} \cite{Baker94,HochbaumM85} to reduce the problem to bounded-size instances, at the expense of a small loss in the approximation factor. Each bounded-size instance is then solved optimally. We give an overview for our bicriteria EPTAS for \probEdgelessnessMinDown where the arguments are more involved; those for \probEdgelessnessDown are relatively simpler.  First, we subdivide the instance into sub-instances induced by $\Oh(1/\epsilon) \times \Oh(1/\epsilon)$ squares, such that the number of optimal disks intersecting the boundaries of the squares is at most $\epsilon k$, and the total sum of changes in the radii of such disks is at most $\epsilon \mu$. In order to solve each subproblem, we further partition the square into smaller grid cells, such that each cell induces a clique in the given UDG. If such a clique contains at most $c/\epsilon$ disks, then we guess the exact subset of disks that are shrunk by an optimal solution. Otherwise, if the clique contains more than $c/\epsilon$ disks, then we mark all disks as \emph{potentially shrinkable}. Note that all except at most one disk from each clique must be shrunk by any solution. Thus, we can argue that the total number of disks that are shrunk in the latter case is at most $\epsilon k$. Now, for each guess of the \emph{potentially shrinkable disks}, we use linear programming to find an optimal solution that is only allowed to shrink the potentially shrinkable disks, in order to result in an edgeless graph. Finally, we need to combine the solutions of subproblems in a careful manner -- since the disks near the boundaries of the cell appear in multiple subproblems. We achieve this using dynamic programming over the optimal solutions computed for the bounded-size instances.

\section{Outline} \label{sec:outline}
We start in \Cref{sec:prelim} with preliminaries including definitions and notations regarding graph theory, geometry, and parameterized algorithms. In \Cref{sec:vc}, we discuss the (partial) kernels and FPT algorithms for (\textsc{Min}) \probEdgelessnessDown. In \Cref{sec:fvs-probs}, we discuss our kernel (resp.~compression and FPT algorithm) for \probAcyclicityDown (resp.~\probAcyclicityMinDown). In \Cref{sec:subexpFPT}, we discuss subexponential FPT algorithms for \probEdgelessnessDown/\textsc{Acyclicity/Connectivity}. In \Cref{sec:vc-eptas}, we design (bicriteria) EPTASes for (\textsc{Min}) \probEdgelessnessDown. In \Cref{sec:lower} we give our lower bound results including NP-hardness for \probEdgelessnessDown/\textsc{Acyclicity/Connectivity}, and \classW{1}-hardness for \probConnectivityFracUp. Finally, in \Cref{sec:conclusion}, we give concluding remarks and future directions. We remark that the conclusion section can be read independently of the technical sections.

\section{Preliminaries}\label{sec:prelim} 
Given $n\in\mathbb{N}$, let $[n]=\{1,2,\ldots,n\}$. 

\paragraph{Graphs.} We use standard graph-theoretic terminology and refer to the textbook of Diestel~\cite{Diestel12} for missing notions. We consider only finite undirected graphs.  For  a graph $G$,  $V(G)$ and $E(G)$ are used to denote its vertex and edge sets, respectively. Given a vertex $v \in V(G)$, let $N_G(v)$ and $N_G[v]$ denote the open and closed neighborhoods of $v$ in $G$, respectively. When $G$ is clear from context, we drop it from the subscript. An {\em $a\times b$ grid}, denote by $\Gamma_{a\times b}$, is the graph on vertex set $\{v_{i,j}: i\in[a], j\in[b]\}$ and edge set $\{\{v_{i,j},v_{i',j'}\}: i,i'\in[a], j,j'\in[b], |i-i'|+|j-j'|=1\}$. The {\em natural plane embedding of $\Gamma_{a\times b}$}, denoted by $d_{a\times b}$, is the mapping of each vertex $v_{i,j}$ to the point $(i,j)$ on the plane, and of each edge to the straight line segment between its endpoints. We say that a graph $H$ is a {\em minor} of a graph $G$ if there exists a function $\varphi: V(H)\rightarrow 2^{V(G)}$, called a {\em minor model}, such that {\em (i)} for every $v\in V(H)$, $G[\varphi(v)]$ is a non-empty connected graph, and {\em (ii)} for every distinct $u,v\in V(H)$, $\varphi(u)\cap \varphi(v)=\emptyset$.
A graph $G$ is \emph{planar} if it has a plane embedding, that is, it can be drawn on the plane without crossing edges.    A \emph{rectilinear} embedding is a planar embedding of $G$ such that vertices are mapped to points with integer coordinates and each edge is mapped into a broken line consisting of an alternate sequence of horizontal and vertical line segments. 

Treewidth is a structural parameter indicating how much a graph resembles a tree. Formally:

\begin{definition}\label{def:treewidth}
	A \emph{tree decomposition} of a graph $G$ is a pair ${\cal T}=(T,\beta)$ of a tree $T$
	and $\beta:V(T) \rightarrow 2^{V(G)}$, such that
	\begin{enumerate}
		\itemsep0em 
		\item\label{item:twedge} for any edge $\{x,y\} \in E(G)$ there exists a node $v \in V(T)$ such that $x,y \in \beta(v)$, and
		\item\label{item:twconnected} for any vertex $x \in V(G)$, the subgraph of $T$ induced by the set $T_x = \{v\in V(T): x\in\beta(v)\}$ is a non-empty tree.
	\end{enumerate}
	The {\em width} of $(T,\beta)$ is $\max_{v\in V(T)}\{|\beta(v)|\}-1$. The {\em treewidth} of $G$, denoted by $\mathsf{tw}(G)$, is the minimum width over all tree decompositions of $G$.
\end{definition}

\paragraph{Points, Disks and segments.} For two points $a$ and $b$ in the plane, we use $ab$ to denote the line segment with endpoints in $a$ and $b$. The \emph{distance} between 
$a=(x_1,y_1)$ and $b=(x_2,y_2)$  or the \emph{length} of $ab$, is $|ab|=\|a-b\|_2=\sqrt{(x_1-x_2)^2+(y_1-y_2)^2}$. 
The \emph{closed disk} $D(c,r)$ with a \emph{center} $c=(c_1,c_2)$ of radius $r$ in the plane is the set of points $(x,y)$ satisfying the inequality $(x-c_1)^2+(y-c_2)^2 \le r^2$. Analogously, the corresponding \emph{open disk} $D_o(c, r)$ is the set of points $(x, y)$ that satisfy the preceding inequality in a strict manner. For a (open or closed) disk with center $c$ and radius $r$, the set of points $(x, y)$ with $(x-c_1)^2 + (y-c_2)^2 = r^2$ is called the \emph{boundary} of the disk. Note that two open disks that only share a common point on their boundaries do not intersect with each other. We measure the distance between two disks as the distance between their centers.

\paragraph{Remark.} We work with \emph{open disks} for (\textsc{Min}) \probEdgelessnessDown / \textsc{Acyclicity}; and with \emph{closed disks} in setting of \probConnectivityFracDown. By making this assumption---which is without any significant loss of generality---we avoid repeated references to infinitesimal perturbations as well as dealing with infima/suprema, which makes the presentation cleaner.

\paragraph{Geometric Definitions.}
As it is common in Computational Geometry, we assume the Real RAM computational model in our algorithmic results. Informally, this means that we assume that basic arithmetic operation may over reals can be performed in unit time. In particular, we allow the input points to have real coordinates. 

Let $\cS$ be a set of geometric objects in the plane (i.e., non-empty subsets of $\mathbb{R}^2$). Then the intersection graph $G(\cS)$ is the graph with the set of vertices  $\cS$ such that  two vertices are adjacent if and only if the corresponding two objects in $\cS$ have a non-empty intersection. 
A \emph{disk graph} is the intersection graph of a finite family of disk on the plane and a disk graph is a \emph{unit disk graphs} if all the disks are of the same radius.  

Let $P$ be a set of points on the plane and let $r: P \to \realplus$ be a family of nonnegative real numbers. Then, we use $\cD(P,r)$ to denote the set of disks $\{D(p,r(p))\colon p\in P\}$. To simplify notation, we use $G(P,r)$ to denote the disk graph $G(\cD(P,r))$. Recall that we use $\unitvec$ to denote the unit radius assignment (which assigns every point of the corresponding set of points a radius of $1$).

\paragraph{Parameterized Complexity.} We refer to the book of Cygan et al.~\cite{DBLP:books/sp/CyganFKLMPPS15} for introduction to the area and undefined notions.  A \emph{parameterized problem} is a language $L\subseteq\Sigma^*\times\mathbb{N}$, where $\Sigma^*$ is a set of strings over a finite alphabet $\Sigma$. An input of a parameterized problem is a pair $(x,k)$, where $x$ is a string over $\Sigma$ and $k\in \mathbb{N}$ is a \emph{parameter}. 

A parameterized problem is \emph{fixed-parameter tractable} (or \classFPT) if it can be solved in $f(k)\cdot |x|^{\Oh(1)}$ time for some computable function~$f$.  
The complexity class \classFPT contains  all fixed-parameter tractable problems. 
Parameterized complexity theory also provides tools  to refute the existence of an \classFPT algorithm for a paramterized problem under some complexity-theoretic assumptions.
The most used assumption is that $\classFPT\neq\classW{1}$. 
Then to show that it is unlikely that a parameterized problem is in \classFPT, one can prove that it is \classW{1}-hard by demonstrating a  \emph{parameterized reduction} from a known \classW{1}-hard problem;  we refer to \cite{DBLP:books/sp/CyganFKLMPPS15} for the formal definitions of the class \classW{1} and parameterized reductions. 

A \emph{kernelization} (or \emph{kernel}) for a paramterized problem $L$ is a polynomial time algorithm that, given an instance $(x,k)$ of $L$, outputs an instance $(x',k')$ of $L$ such that 
(i) $(x,k)\in L$ if and only if $(x',k')\in L$ and (ii) $|x'|+k'\leq f(k)$ for a computable function $f$.  The function $f$ is called the kernel \emph{size}; a kernel is \emph{polynomial} if $f$ is a polynomial. It can be shown that every decidable \classFPT problem admits a kernel. However, it is unlikely that all \classFPT problems
have polynomial kernels. In particular, there is the now standard \emph{cross-composition} technique to show that a parameterized problem does not admit a polynomial kernel unless $\classNP\subseteq \classCoNP/\poly$. We refer to~\cite{DBLP:books/sp/CyganFKLMPPS15} and the recent book on kernelization of Fomin et al.~\cite{fomin2019kernelization} for details. 

There is a slightly weaker notion of kernelization, called \emph{partial kernelization}, which was introduced by Betzler et al. \cite{BetzlerGKN11} (also \cite{BasavarajuFRS16}). Here, instead of reducing the entire instance to one bounded by the size of the parameter, we try to bound only a ``part'' of the instance by a function of the parameter. In our context, for some problems, our partial kernelization algorithm outputs instances where the number of points in the instance is bounded by a function of the parameter, but the values (of coordinates, radii, cost budget or annotations) may be real numbers. 

\paragraph{Linear Programming in Real RAM model.} In many of our algorithms, we will need to find optimal/feasible solutions for linear programs defined over real numbers (recall that we work in real RAM model). To this end, we cannot use the standard weakly-polynomial algorithms based on interior point or ellipsoid methods. Instead, we use the following result.

\begin{proposition}[\cite{ChanTALG18}] \label{prop:lpsolver}
	There exists a deterministic algorithm for finding an optimal solution to a linear program with $d$ variables and $n$ constraints that runs in time $d^{\Oh(d)} \cdot n$ time in Real RAM model.
\end{proposition}

\section{(Partial) Kernels for {\sc (Min) $k$-Shrinking to Independence} and Consequences} \label{sec:vc}

In \Cref{subsec:vc-partialkernel}, we design linear partial kernels for (\textsc{Min}) \probEdgelessnessDown. Then, in \Cref{subsec:polykernel-vc}, we use this result to obtain a true polynomial kernel for \probEdgelessnessDown. Finally, in \Cref{subsec:fpt-vc-ksum}, we use the partial linear kernel for \probEdgelessnessMinDown in order to design an FPT algorithm for the problem parameterized by $k$. For the sake of brevity, we use the shorthand $\Pi_1$ for \probEdgelessnessDown, and $\Pi_2$ for \probEdgelessnessMinDown. 

\subsection{Partial Kernels for (\textsc{Min}) \probEdgelessnessDown} \label{subsec:vc-partialkernel}

Observe that if $\cI$ is a yes-instance of $\Pi_1$ (resp. $\Pi_2$), then the set of $k$ shrunken disks must be a vertex cover of $H \coloneqq G(P, \mathbf{1})$. Thus, we compute a $2$-approximation to vertex cover of $H$, and if the size of this vertex cover is larger than $2k$, we conclude that $H$ has no vertex cover of size at most $k$, which implies that $\cI$ is a no-instance. This leads to the following observation.
\begin{observation} \label{obs:min-vc}
	There exists a polynomial-time algorithm that, given an instance $\cI$ of $\Pi_1$ (resp.\ $\Pi_2$), either correctly concludes that $\cI$ is a no-instance, or outputs $U \subseteq P$ that is a vertex cover for $H$ of size at most $2k$.
\end{observation}

Now, we prove the following theorem.

\begin{theorem} \label{thm:vc-partialpolykernels}
	There exists a polynomial-time algorithm that, given an instance $\cI = (P, k, \alpha)$ of $\Pi_1$ (resp.\ $\cI = (P, k, \alpha, \mu)$ of $\Pi_2$), outputs an equivalent instance $\cI' = (T, k,  \alpha)$ of $\Pi_1$ (resp.\ $\cI' = (T, k, \alpha, \mu)$ of $\Pi_2$) such that $P' \subseteq P$, and $|P'| = \Oh(k)$.
\end{theorem}
\begin{proof}
	Since the arguments for $\Pi_1$ and for $\Pi_2$ are almost identical, we focus on $\Pi_1$, and only highlight the differences for $\Pi_2$.
	Given an instance $\cI = (P, \alpha, k)$ of $\Pi_1$, the kernelization algorithm works as follows. First we use \Cref{obs:min-vc} to obtain a subset $U \subseteq P$ that forms a vertex cover of size at most $2k$ for the intersection graph $H$, or conclude that $\cI$ is a no-instance. Suppose that $|U| \le 2k$. Let $T \subseteq P$ denote the union of closed neighborhoods of the vertices in $U$; in particular $U \subseteq T$. Let $F \coloneqq P \setminus T$, and $I \coloneqq P \setminus U$. 
	We summarize the simple consequence of the definitions in the following observation.
	\begin{observation} \label{obs:vc-properties}\ 
		\begin{enumerate}
			\item $T = \{p \in P : |pq| \le 2 \text{ for some } q \in U \}$
			\item For any two distinct points $p, q \in I$, $|pq| > 2$.
			\item For any $p \in F$ and $q \in U$, $|pq| > 2$.
		\end{enumerate}
	\end{observation}
	Let $\cI' = (T, k, \alpha)$ be an instance of $\Pi_1$. We first prove the equivalence. Observe that $F$ is the set of isolated vertices in $H$, and without loss of generality, we can assume that no solution of $\cI$ shrinks a disk corresponding to a point in $F$. This implies that, if $\cI$ is a yes-instance then any solution to $\cI$ translates to a solution to $\cI'$ of the same size (and cost for $\Pi_2$). 
	
	In the other direction, suppose that $\cI'$ is a yes-instance with a solution $(S',r')$. We claim that $(S', r)$ is also a solution for $\cI$, where $r$ is an extension of $r$ to $P$ by setting $r(p) = 1$ for all points $p \in P \setminus T$. Note that for any points $p, q \in T$, the disks $D(p, r(p))$ and $D(q, r(q))$ do not intersect each other. Using \Cref{obs:vc-properties}, it follows that for any distinct $p, q \in I$, or $p \in F$ and $q \in U$, $|pq| > 2$, so the disks $D(p, r(p))$ and $D(q, r(q))$ cannot intersect each other. Therefore, $\sigma$ is a solution for the original instance $\cI$. This shows that the instances $\cI$ and $\cI'$ of $\Pi_1$ are equivalent.  Further, observe that the costs of the two solutions are equal. 
	
	Now, notice that for every $p \in U$, let $Q_p \subseteq I$ be the set of points $q$ such that $|pq| < 2$. Note that $pq \in E(H)$. However, note that for any distinct $q_1, q_2 \in Q_p$, $|q_1q_2| > 2$ due to \Cref{obs:vc-properties}. It follows that $|Q_p| \le 6$ -- this follows since $H$ is a unit disk graph. Now, since $|U| \le 2k$, it follows that $|T| \le |U| + 6|U| = \Oh(k)$. 
	
\end{proof}

\subsection{A Polynomial Kernel for \probEdgelessnessDown} \label{subsec:polykernel-vc}

Proceeding as in the proof of \Cref{thm:vc-partialpolykernels}, we obtain the partial linear kernel, which is given by the set of points $T$, which is the union of closed neighborhoods of a $2$-approximate vertex cover $U$ for the original intersection graph $H = G(P, \unitvec)$. As shown in the proof of \Cref{thm:vc-partialpolykernels}, $(P, \alpha, k)$ and $(T, \alpha, k)$ are equivalent. In order to obtain a truly polynomial kernel, we give a compression to an intermediate annotated graph problem, which we call $\Pi'_1$.

Before discussing the compression, we need the following definitions.
\begin{definition} \label{defn:removal}
	We say that $(S, r)$ is a \emph{solution} for an instance of (\textsc{Min}) \probEdgelessnessDown, if $S \subseteq P$ is a set of size at most $k$, and $r(p) = \alpha$ (resp. $\alpha \le r(p) < 1$) for all $p \in S$ and $r(p) = 1$ for all $p \not\in S$, and if the resulting intersection graph $G(P, r)$ is edgeless. We say that a solution $\sigma = (S, r)$ \emph{removes} an edge $e = uv$ of $H$ if $e$ does not exist in the \emph{resulting} intersection graph $G(P, r)$. 
\end{definition}
Consider an edge $uv \in E(H)$, then $0 < |uv| \le 2$. To encode whether a solution to $\Pi_1$ can remove an edge $uv$, and if so, how many vertices need to be shrunk; we define an auxiliary function $\nu: [0, 2] \to \{1, 2, \bot\}$ for designing the compression for $\Pi_1$.

\begin{definition} \label{defn:nu}
	The function $\nu: (0, 2] \to \{1, 2, \bot\}$ is defined as follows.
	$$\nu(d) = \begin{cases}
		\bot &\text{ if } 0 < d < 2\alpha \text{\qquad\ \ \small $\backslash\backslash$ \emph{no solution can remove the edge}}
		\\2 &\text{ if } 2\alpha \le d < 1 + \alpha \text{\quad \small $\backslash \backslash$ \emph{edge can only be removed by shrinking both endpoints to $\alpha$}}
		\\1 &\text{ if } 1+\alpha \le d \le 2  \text{\qquad \small $\backslash \backslash$ \emph{edge can be removed by shrinking one of the endpoints to $\alpha$}}
		
	\end{cases}$$
\end{definition}

\begin{observation} \label{obs:nu-properties}
	Let $p$ and $q$ be two points with $0 < d = |pq| \le 2$. Then, the following are true.
	\begin{itemize}
		\item $\nu(d) = 1$ iff $D(p, 1) \cap D(q, \alpha) = \emptyset$.
		\item $\nu(d) = 2$ iff $D(p, \alpha) \cap D(q, \alpha) = \emptyset$, but $D(p, 1) \cap D(q, \alpha) \neq \emptyset$.
		\item $\nu(d) = \bot$ iff $D(p, \alpha) \cap D(q, \alpha) \neq \emptyset$.
	\end{itemize}
\end{observation}
See \Cref{fig:nufunction} for an illustration.

\begin{figure}
	\centering
	\includegraphics[scale=1,page=7]{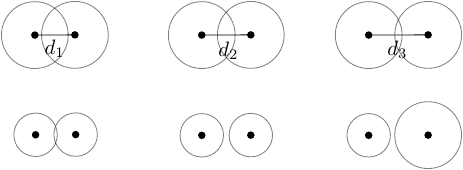}
	\caption{Top row: Pairs of disks with varying distances between the centers, with $\nu(d_1) = \bot, \nu(d_2) = 2$ and $\nu(d_3) = 1$. Here, $\alpha = 0.65$. 
		\\Bottom row: The corresponding edge can be removed by shrinking no, $2$ and $1$ disk, respectively.} \label{fig:nufunction}
\end{figure}

Let $H' = H[T] = G(T, \unitvec)$. For each edge $e = uv \in E(H')$, we compute the value $\nu(|uv|)$, and annotate the edge with $\nu(|uv|)$, which we denote by $\nu(e)$. Let $H_A$ denote the resulting annotated graph---note that $H_A$ has same set of vertices and edges as $H'$, except that each edge $e \in E(H_A)$ is annotated with $\nu(e) \in \LR{\bot, 1, 2}$. 

If for any edge $e \in E(H_A)$, $\nu(e) = \bot$, then no solution can remove the edge $e$. Thus, we stop and output a trivial no-instance of $\Pi_1$. Thus, we may assume that each $e \in E(H_A)$ is annotated with $\nu(e) \in \{1, 2\}$. 

Now, the intermediate problem $\Pi'_1$ is defined as follows. The input to $\Pi'_1$ is the instance $(H, k)$, where $H$ is the annotated graph as defined above. The task is to determine whether there exists a solution of the form $(S, {\sf cover})$, where 
\begin{itemize}
	\item $S \subseteq V(H_A)$ is a subset of vertices that is shrunk,
	\item ${\sf cover}: E(H_A) \to V(H_A) \cup (V(H_A) \times V(H_A))$ is a function that encodes for each edge $e \in E(H_A)$, whether $e$ is removed by shrinking one or both of its endpoints.
\end{itemize}
The solution must satisfy the following properties:
\begin{itemize}
	\item $|S| \le k$
	\item For each edge $e = uv \in E(H_A)$, ${\sf cover}(e) \subseteq S$, and,
	\begin{itemize}
		\item If $\nu(e) = 2$, then ${\sf cover}(e) = \{u, v\}$. That is, if the edge can only be removed by shrinking both of its endpoints, both of its endpoints must be shrunk,
		\item If $\nu(e) = 1$, then $\emptyset \neq {\sf cover}(e) \subseteq \{u, v\}$. That is, if the edge can be removed by shrinking one of its endpoints, then at least one of its endpoints must be shrunk.
	\end{itemize}
\end{itemize}

Observe that the problem $\Pi'_1$ is trivially in \classNP. Furthermore, the size of the instance $(H_A, k)$ is bounded by $\Oh(k^2)$, since $|E(H_A)| = \Oh(k^2)$, and each edge is annotated with one bit of information. Finally, since $\Pi_1$ is \classNPC\ via \Cref{thm:np-hardness-vc}, there exists a polynomial-time reduction from $(H_A, k)$ to an equivalent instance $(P', \alpha', k')$ of $\Pi_1$ -- this follows from the fact that there exists a polynomial-time reduction from any problem in \classNP\ to a problem in \classNPC; furthermore, the size of the resulting instance is polynomial in the size of the original instance. Thus, size of the resulting instance of $\Pi_1$ is bounded by a polynomial in the original parameter $k$. 

\vcpolykernel*

\subsection{An FPT Algorithm for \probEdgelessnessMinDown} \label{subsec:fpt-vc-ksum}
Given an instance $\cI = (P,\alpha, k, \mu)$ of $\Pi_2$, and use \Cref{thm:vc-partialpolykernels} to obtain an equivalent instance $\cI' = (T, \alpha, k, \mu)$ such that $|T| = \Oh(k)$. Note that the costs of the optimal solutions for $\cI$ and $\cI'$ -- if they are yes-instances -- are equal. Thus, here we focus on finding an optimal solution to the instance $\cI'$, if one exists, or conclude that $\cI'$ is a no-instance. The decision version of $\Pi_2$ can then be solved by checking whether the cost of the optimal solution found is at most $\mu$. 

We iterate over all subsets $S' \subseteq T$ of size at most $k$ -- note that there are at most $|T|^{\Oh(k)} = k^{\Oh(k)}$ such subsets. Consider an iteration corresponding to a fixed subset $S'$, which we call as \emph{shrinkable}. The remaining points $L = T \setminus S'$ are marked as \emph{unshrinkable}. 

Now, we define a linear program (LP) corresponding to each guess $(S', L)$. In this LP, we have a variable $r(p)$ for each $p \in T$. For each $p \in L$, we set $r(p) = 1$. On the other hand, for each $p \in S'$, we add a linear constraint $\alpha \le r(p) \le 1$. Finally, for each pair of points $p, q \in T$, with $|pq| \le 2$, we add a constraint $r(p) + r(q) \le |pq|$. The linear objective to minimize is $\sum_{p \in T} (1-r(p))$. It is easy to see the equivalence between a feasible solution that shrinks disks from $S'$ and a feasible solution to the LP. 

Note that the number of variables is $|T| = \Oh(k)$ and the number of constraints is $\Oh(k^2)$. Thus, using the algorithm of \Cref{prop:lpsolver} runs in time $k^{\Oh(k)} \cdot \Oh(k^2) = k^{\Oh(k)}$ time, we can find an optimal solution minimizing $\sum_{p \in S'} r(p)$, if the LP is feasible; or conclude that the guess $S'$ is incorrect.  We consider the minimum-cost solution found over all iterations, if one exists. Otherwise, we correctly conclude that $\cI'$, and thus $\cI$, does not admit any solution that shrinks at most $k$ disks, and of cost at most $\mu$.

\vcsumk*

\section{Compressibility of ({\sc Min}) \probAcyclicityDown and Consequences} \label{sec:fvs-probs}

In this section, we give polynomial kernel (resp.\ polynomial partial kernel) for \probAcyclicityDown (resp.\ \probAcyclicityMinDown). For \probAcyclicityMinDown, we then use the (partial) kernel to obtain an FPT algorithm. Since the initial arguments towards obtaining (partial) kernel are common for both the problems, we keep the discussion common for the initial part, and diverge subsequently. For the sake of brevity, we refer to \probAcyclicityDown as $\Pi_1$, and  to \probAcyclicityMinDown as $\Pi_2$.

\subsection*{Obtaining Annotated Graph via Reduction Rules} \label{subsec:annotated}

As in the previous section, let $H \coloneqq G(P, \unitvec)$ denote the original (open) unit disk graph. We have the following reduction rule. Further, we adapt \Cref{defn:removal} of a \emph{solution} appropriately for (\textsc{Min}) \probAcyclicityDown.

\begin{redrule} \label{redrule:fvs0}
	Suppose that in the given instance $\cI = (P, \alpha, k)$ of $\Pi_1$ (resp.\ $\Pi_2$), we have that $k \le 0$, and the graph $H$ is not a forest, then conclude that $\cI$ is a no-instance of $\Pi_1$ (resp.\ $\Pi_2$).
\end{redrule}

\begin{lemma} \label{lem:fvs-degbound}
	If $\cI$ is a yes-instance of $\Pi_1$ (resp.\ $\Pi_1$), then the degree of any vertex $p$ in the \emph{original} intersection graph $G$ is upper bounded by $25k + 50 = \Oh(k)$.
\end{lemma}
\begin{proof}
	We overlay a unit grid (i.e., side-length $1$) on the real plane. Note that the distance between any two points of $P$ that lie in the same grid cell is at most $\sqrt{2} \le 2$. Thus, they induce a clique in $H$. Note that from any clique $K$ in $H$ of size $s \ge 3$, the disks corresponding to at least $s-2$ points in $K$ must be shrunk in the solution -- otherwise, there will be a cycle in the resulting intersection graph. Since $\cI$ is a yes-instance, we can assume that each cell contains at most $k+2$ points.
	
	Consider any point $p$ lying in a cell $C$ with center $(x, y)$. Then, for any edge $pq \in E(H)$, $q$ belongs to a cell $C'$ with center $(x', y')$ such that $|x - x'| + |y - y'| \le 2$. That is, $p$ can have neighbors in one of at most $25$ cells. From the previous paragraph, each cell contains at most $k+2$ points. This concludes the proof of the lemma.
\end{proof}
From \Cref{lem:fvs-degbound}, it is easy to see the correctness of the following two reduction rules.

\begin{redrule} \label{redrule:fvs1}
	If the degree of any vertex in $H$ is more than $25k+50$, then conclude that $\cI$ is a no-instance of $\Pi_1$ (resp.\ $\Pi_2$).
\end{redrule}

Thus, in the following, we assume that the degree of all the vertices in $H$ is bounded by $\Delta = \Oh(k)$. Next, we apply the following three reduction rules, whose correctness is immediate.

\begin{redrule} \label{redrule:fvs2}
	If any vertex corresponding to $p \in P$ has degree at most $1$ in $H$, then obtain an equivalent instance of $\Pi_1$ (resp.\ $\Pi_2$) by deleting $p$ from $P$.
\end{redrule} 

\begin{redrule} \label{redrule:fvs3}
	If any connected component $F$ of $H$ is acyclic (i.e., induces a tree), then obtain an equivalent instance of $\Pi_1$ (resp.\ $\Pi_2$) by deleting all such points of $F$ from $P$.
\end{redrule}

After \Cref{redrule:fvs1}, \Cref{redrule:fvs2}, and \Cref{redrule:fvs3}, we ensure that (i) the minimum degree in $H$ is at least $2$, (ii) the maximum degree is $\Delta = \Oh(k)$, and (iii) each connected component of $H$ contains a cycle. Assuming $\cI$ is a yes-instance, if we bound the number of degree-$2$ vertices in $H$, we immediately bound the total number of vertices in $H$, using standard arguments for feedback vertex set kernels. However, since shrinking a disk is not the same as removing the corresponding vertex from the graph, we have to use different arguments for bounding the number of vertices. This is also the reason we aim for a polynomial compression (for the case of $\Pi_1$, we then infer a polynomial kernel).

Before discussing the compression, we recall \Cref{defn:removal} of \emph{removal of an edge}, \Cref{defn:nu} of the \emph{``$\nu$ function''}, and \Cref{obs:nu-properties} for the properties of $\nu$ function. Additionally, the following observation is relevant for $\Pi_2$.

\begin{observation} \label{obs:dist-properties}
	Let $p$ and $q$ be two points with $0 < d = |pq| \le 2$, then $D(p, x_p) \cap D(q, x_q) = \emptyset$ iff $x_p + x_q \le d$. In particular, consider an edge $e = pq \in E(H)$, 
	\begin{itemize}
		\item If $0 < |pq| < 2\alpha $, then no solution that is allowed to shrink radii within range $[\alpha, 1]$ can remove the edge. We say that such an edge is \emph{irremovable}.
		\item If $2\alpha \le |pq| < 1+\alpha$, then $D(p, x_p) \cap D(q, x_q) = \emptyset$ iff $2\alpha \le x_p + x_q \le |pq| \le 1 + \alpha$. In particular, any solution must shrink at least one of the disks centered at $p$ and $q$ to $\alpha$.
		\item If $1+\alpha \le |pq| \le 2$, then $D(p, x_p) \cap D(q, x_q) = \emptyset$ iff $2\alpha \le x_p + x_q \le |pq| \le 2$. That is, at least one of the disks must be shrunk.
	\end{itemize}
\end{observation}

For any connected component $C$ of $H$, contains at least one vertex of degree at least $3$, then we say that $C$ is an \emph{interesting} connected component. Otherwise, $C$ contains all degree-$2$ vertices (recall that the minimum degree is at least $2$), i.e., $C$ is a cycle. In this case, we say that $C$ is an \emph{isolated cycle}. We handle the two types of connected components separately.

\paragraph{Handling isolated cycles.} Consider an isolated cycle $C$ in $H$, then in any solution to $\Pi_1$ (resp.\ $\Pi_2$) must shrink at least one disk corresponding to vertices on $C$ such that at least one edge of $C$ is removed. We have the following claim.

\begin{claim}\label{cl:fvs-cycles}
	Suppose $\cI$ be a yes-instance of $\Pi_1$ (resp.\ $\Pi_2$), and let $\sigma = (S, r)$ be a solution for $\Pi_1$ (resp.\ $\Pi_2$). Let $C$ be an isolated cycle, and suppose $\sigma$ shrinks disks corresponding to at least $3$ vertices of $V(C)$. Then, there exists another solution $\sigma' = (S', r')$ for $\Pi_1$ (resp.\ $\Pi_2$) that shrinks disks corresponding to at most $2$ vertices $V(C)$. For $\Pi_2$, this claim holds with an additional property that the cost of $\sigma'$ is at most that of $\sigma$.
\end{claim}
\begin{proof}
	Note that the solution $\sigma$ must remove at least one edge of the cycle $C$, thus suppose this edge is $v_{i}v_{i+1}$ for some pair of consecutive vertices $v_i, v_{i+1}$ on the cycle. Then, at least one of the two disks corresponding to $v_i$ and $v_{i+1}$ is shrunk in the solution $\sigma$. Let $W = S \cap \{v_i, v_{i+1}\}$. We define a new solution $\sigma' = (S', r')$, such that $S' = S \setminus (V(C) \setminus W)$, and $r'(p) = 1$ for all points $p \not\in S'$, and $r'(p) = r(p)$ otherwise. Clearly, the solution $\sigma'$ shrinks only a fewer number of disks (and for $\Pi_2$ has cost no greater than that of $\sigma$). Furthermore, the resulting graph $G(P, r)$ does not contain the edge $v_i v_{i+1}$, since $S \cap \{v_i, v_{i+1}\} = S' \cap \{v_i, v_{i+1}\}$. 
\end{proof}
\Cref{cl:fvs-cycles} implies that from any isolated cycle $C$, the disks corresponding to at most two vertices need to be shrunk. Define $\dmax(C)$ as the maximum of the euclidean distances between the consecutive pairs of vertices on an isolated cycle $C$. We use the following reduction rule to get rid of all vertices on an isolated cycle $C$.

\begin{redrule} \label{redrule:fvs-4-isolated-cycles}
	Consider an isolated cycle $C$ in $H$. 
	\begin{itemize}
		\item If $\nu(\dmax(C)) = \bot$, i.e., if all the edges of $C$ are \emph{irremovable}, then conclude that $\cI$ is a no-instance of $\Pi_1$. 
		\item Otherwise, $\nu(\dmax(C)) \in \{1, 2\}$. In this case, obtain a new equivalent instance $\cI' = (P \setminus V(C), k-f(\dmax(C))$ of $\Pi_1$.
	\end{itemize}
\end{redrule}

We replace an isolated cycle $C$ with a special vertex $v(C)$, which we call a \emph{isolated cycle vertex}. We annotate it with $(\#(v(C))$ and $\cost(v(C)) \coloneqq 2-\dmax(C))$, where $\#(v(C)) = \nu(\dmax(C))$, i.e., the minimum number of disks that need to be shrunk in order to remove an edge from $C$. Note that if $\dmax(C) \le 2\alpha$, then $\#(v(C)) = \infty$, i.e., all the edges of $C$ are irremovable. In this case, we conclude that it is a no-instance of $\Pi_2$ via \Cref{obs:dist-properties}. Otherwise, any solution must contribute at least $\cost(v(C))$ in order to remove at least one edge from $C$, via \Cref{obs:dist-properties}. Thus, we have the following reduction rule.

\begin{redrule} \label{redrule:fvs-4'-isolated-cycles}
	If the number of isolated cycle vertices is greater than $k$, or if for some isolated cycle vertex $v(C)$, $\#(v(C)) = \infty$, then conclude that $\cI$ is a no-instance of $\Pi_2$.
\end{redrule}
Now we discuss how to handle interesting components, i.e., a connected component of $H$ that contains at least one vertex of degree at least $3$. First, we define an arbitrary total order $\preceq$ on the vertices of degree at least $3$ in $H$.

\paragraph{Handling long degree-$2$ paths.} Consider a sequence of vertices $\pi = (v_0, v_1, \ldots, v_\ell, v_{\ell+1})$ in $H$ such that (i) $\ell \ge 1$, (ii) for every $0 \le i \le \ell$, $v_iv_{i+1} \in E(H)$, (iii) the endpoints $v_0$ and $v_{\ell+1}$ have degree at least $3$ in $H$, with $v_0 \preceq v_{\ell+1}$, and (iv) all internal vertices $v_1, v_2, \ldots, v_\ell$ have degree exactly $2$ in $H$. Note that if $v_0$ and $v_{\ell+1}$ are distinct, then $\pi$ induces a path in $H$; otherwise if $v_0 = v_{\ell+1}$, then $\pi$ induces a cycle in $H$. In both the cases, we say that $\pi$ is a \emph{reducible path}. We denote the set of internal degree-$2$ vertices of $\pi$ by $I(\pi)$. If $|I(\pi)| \ge 3$, we say that $\pi$ is a \emph{long} reducible path; otherwise we say that $\pi$ is a \emph{short} reducible path. First, we claim that it suffices to shrink at most $2$ internal vertices on a long reducible path $\pi$. 
\begin{claim} \label{cl:fvs-longpath}
	Suppose $\cI$ be a yes-instance of $\Pi_1$ (resp.\ $\Pi_2$), and let $\sigma = (S, X)$ be a solution. Let $\pi$ be a long reducible path in $H$, and suppose $\sigma$ shrinks disks corresponding to at least $3$ vertices of $I(\pi)$. Then, there exists a solution $\sigma' = (S', r')$ that shrinks disks corresponding to at most $2$ vertices of $I(\pi)$. For $\Pi_2$, this claim holds with an additional property that the cost of $\sigma'$ is at most that of $\sigma$.
\end{claim}
\begin{proof}
	Consider a solution $\sigma = (S, r)$ with the property specified in the claim. Without loss of generality, we can assume that at least one edge of $\pi$ is not present in the resulting acyclic graph $G(P, r)$ -- otherwise, we can obtain a solution $(S', X')$ of smaller size (resp.\ smaller size \emph{and} cost for $\Pi_2$), where $S' = S \setminus I(\pi)$ and setting $r'(p) = 1$ for points in $S \cap I(\pi)$. Note that the resulting graph $G(P, r')$ will still be acyclic.
	
	Thus, at least one edge of $v_iv_{i+1}$ of $\pi$ does not exist in the resulting graph $G(P, r)$. In this case, the argument is almost identical to the proof of \Cref{cl:fvs-cycles}. Consider the solution $\sigma' = (S', X')$ as defined in the proof of \Cref{cl:fvs-cycles}, and let $G(P, r')$ be the new intersection graph. Now, suppose for contradiction that $G(P, r')$ contains a cycle $C$. Then, $C$ must contain a vertex from $I(\pi) \setminus W$. However, any cycle passing through an internal vertex of $\pi$ must have $\pi$ as a subgraph. However, this is a contradiction since the edge $v_iv_{i+1}$ does not exist in $G(P, r')$. 
\end{proof}


\paragraph{Replacing Reducible Paths.} Now we seek to replace the internal vertices of a reducible path $\pi$ by remembering some information about how a solution can remove an edge of $\pi$. Fix a reducible path $\pi = (v_0, v_1, v_2, \ldots, v_\ell, v_{\ell+1})$, where $v_{0}$ and $v_{\ell+1}$ have degree at least $3$ in $H$. We replace $\pi$ by a special kind of an edge $e(\pi)$, which we call a \emph{reducible edge}. Note that there may be multiple reducible (short or long) paths $\pi_1, \pi_2, \ldots, \pi_t$ with common endpoints $p$ and $q$. We replace each such $\pi_i$ with an associated reducible edge $e(\pi_i)$. Furthermore, both the endpoints of a reducible path may be the same vertex $p$ with degree at least $3$, in which case we replace each such $\pi_i$ with a reducible edge $e(\pi_i)$ which is a self-loop (note that in this case the internal degree-$2$ vertices on $\pi_i$ must be at least $2$). Thus, the resulting graph may be a multigraph with multiple parallel edges between vertices of degree at least $3$, and multiple self-loops on vertices. 

Next, we annotate each reducible edge $e(\pi)$ with a tuple that encodes information about different ways in which a solution of $\Pi_1$ or $\Pi_2$ can shrink some disks in order to remove different edges of $\pi$. First, consider the case where $\pi$ is a long reducible path. Let $\dstart \coloneqq |v_0v_1|$, $\dend \coloneqq |v_\ell v_{\ell+1}|$, and $\displaystyle \dmax \coloneqq \max_{1 \le i \le \ell-1} |v_i v_{i+1}|$. For $\Pi_2$, we simply annotate $e(\pi)$ with the tuple $(\dstart, \dend, \dmax)$. For $\Pi_1$, we annotate the edge with the tuple $(\nu(\dstart), \nu(\dend), \nu(\dmax))$. 

Note that if $\pi$ is a short reducible path, then note that $\dmax$ may not be defined as per the definition above, in which case we define it to be $\bot$. For $\Pi_2$, we annotate $e(\pi)$ with $(\dstart, \dend, \dmax)$. For $\Pi_1$, we annotate $e(\pi)$ with $(\nu(\dstart), \nu(\dend), \nu(\dmax))$, where we define $\nu(\bot) \coloneqq \bot$. For each reducible edge $e$, we use the shorthand $\nu_{\sfstart}(e), \nu_\sfend(e), $ and $\nu_{\sfmax}(e)$ to denote the annotated values.

\paragraph{Annotating Remaining Edges.} Consider an edge $e = pq \in E(H)$ in the original graph such that $p$ and $q$ are distinct vertices of degree at least $3$. We say that $e$ is an \emph{original edge}. We annotate each original edge as follows. For $\Pi_2$, we simply annotate it with $|pq|$, whereas for $\Pi_1$, we annotate the edge with $\nu(|pq|)$, which we refer to as $\#(e)$.

After replacing all reducible paths with reducible edges, and annotating all original edges, we have the following observation. 
\begin{observation} \label{obs:fvs-degbound-loops}
	After replacing all reducible paths with reducible edges, the degree of every vertex is upper bounded by $25k+50$, where we assume that a self-loop incident to a vertex $v$ contributes $2$ to the degree of $v$.
\end{observation}
\begin{proof}
	Consider the operation of replacing a reducible path $\pi$ with a reducible edge $pq$. Suppose $p$ and $q$ are distinct, then the degree-$2$ vertices on $\pi$ contribute exactly one toward the degrees of $p$ and $q$ respectively. On the other hand, if $p = q$, then there are at least two degree-$2$ vertices on $\pi$, which contribute one each to the degree of $p$. Thus, the replacement operation preserves the degree of each vertex that had degree at least $3$ in the original graph $G$ in all cases. Since each vertex has degree at most $25k+50$ from \Cref{redrule:fvs1}, the observation follows.
\end{proof}

\subsection{A Polynomial Kernel for \probAcyclicityDown.} \label{subsec:polykernel-fvs}

We have an annotated multigraph $H_A$ with multiple parallel edges between pairs of vertices, and multiple self-loops on a vertex such that the degree of each vertex $v \in V(H_A)$ is at least $3$ and at most $25k+50$. Let $E_r$ denote the set of reducible edges of $E(H_A)$, and $E_o$ denote the set of original edges of $E(H_A)$. We want to decide the following problem on $H_A$, which we call $\Pi'_1$. 
\ \\\begin{tcolorbox}[colback=gray!5!white,colframe=gray!75!black]
	\begin{description}
		\item[Input:]  An annotated multigraph $H_A$ with multiple parallel edges between pairs of vertices, a total order $\prec$ on $V(H_A)$, and a parameter $k$.
		\item[Task:] Determine whether there exists a solution $\sigma_{\Pi'_1} = (W, (F, K_o, K_r), f_o, g_r, c_r, w_r)$, where
		\begin{itemize}[leftmargin=*]
			\item $W \subseteq V(H_A)$ {\footnotesize \em $\triangleright$ denotes the set of vertices of degree at least $3$ that are shrunk}
			\item $(F, K_r, K_o)$ is a partition of $E(H)$ with $K_r \subseteq E_r$, and $K_o \subseteq E_o$ \newline{\ \qquad \footnotesize \em $\triangleright$ $F$ induces a forest, and the edges of $K_r \cup K_o$ are removed}
			\item $f_o: K_o \to W \cup (W \times W)$, {\footnotesize \em $\triangleright$ $f_o$ specifies how an original edge is removed}
			\item $g_r: K_r \to \{{\sf start, end, max}\}$, $c_r: K_r \to \{1, 2\}$, $n_r: K_r \to \{0, 1\}$, and $w_r: K_r \to W$ \newline{\footnotesize \em  $\triangleright$ $g_r$ encodes how an edge is removed, $c_r$ encodes how many degree-$2$ vertices are implicitly removed, and $w_r$ encodes the endpoints of reducible path to be removed}
		\end{itemize}
		such that the following properties are satisfied:
		\begin{itemize}[leftmargin=*]
			\item The graph $(V(H_A), F)$ is simple (i.e., without self-loops or parallel edges), and acyclic,
			\item  If $e = pq \in K_o$, (i) $\emptyset \neq f_o(e) \subseteq \{p, q\}$, (ii) $f_o(e) \subseteq W$, and (iii) $|f_o(e)| = \#(e)$ \newline{\footnotesize \em $\triangleright$ an original edge can be removed by shrinking one or both of its endpoints, as specified by its annotation}
			\item If $e = pq \in K_r$, with $p \preceq q$ then at least one of $\nu_{\sfstart}, \nu_{\sfend}$, and $\nu_{\sfmax}$ is at least $1$, and:
			\begin{itemize}[leftmargin=*]
				\item If $g_r(e) = {\sf start}$, then $\nu_{\sfstart}(e) \ge 1$, and $w_r(e) = p$. 
				\newline If $\nu_{\sfstart}(e) = 2$, then $c_r = 1$ 
				\newline{\footnotesize \em $\triangleright$ If a $e$ is removed by removing the first edge on the reducible path, then only shrink $p$ if possible}
				\item If $g_r(e) = {\sf end}$, then $\nu_{\sfend}(e) \ge 1$, and $w_r(e) = q$. 
				\newline If $\nu_{\sfend}(e) = 2$, then $c_r = 1$
				\item If $g_r(e) = {\sf max}$, then $c_r(e) = \nu_{\sfmax}(e) \ge 1$.
			\end{itemize}
			\item $\displaystyle |W| + \sum_{e \in K_r} c_r(e) \le k$
		\end{itemize}
	\end{description}
\end{tcolorbox}

\begin{lemma} \label{lem:fvs-pi1-compression-bound}
	If $(H_A, \prec, k)$ is a yes-instance of $\Pi'_1$, then $|V(H_A)| = \Oh(k^2)$, and $|E(H_A)| = \Oh(k^3)$, accounting for multiplicities.
\end{lemma}
\begin{proof}
	If $(H_A, \prec, k)$ is a yes-instance of $\Pi'_1$, then there exists a solution $\sigma = (W, (F, K_o, K_r), f_o, g_r, c_r, w_r)$ with the specified properties. Let $K_r' \subseteq K_r$ be the set of edges with $g_r(e) = \sfmax$, i.e., the edges that are removed by the virtue of deletion of internal vertices. Let $K' = (K_o \cup (K_r \setminus K'_r)$, denote the remaining set of edges that are removed. We observe that for any edge $e \in K'$, at least one endpoint of $e$ belongs to the set $W$. Since $(V(H_A), F)$ is a forest, $W$ is a \emph{feedback vertex set} (in the normal sense) for the graph $H'_A = (V(H_A), F \cup K')$.
	
	Recall that in the graph $H_A$, the minimum degree is at least $3$. However, in graph $H'_A$, at most $2k$ vertices may have degree at most $2$, by the virtue of removal of the edges of $K'_r$. This follows from the fact that each edge $e \in K'_r$ has $c_r(e) \ge 1$, and $\sum_{e \in K'_r} c_r \le \sum_{e \in K_r} c_r \le k$. In the following, we use a slight variation of the argument that bounds the number of vertices in a graph with minimum degree $3$ and maximum degree $\Delta$, and feedback vertex set of size $k$.
	
	Let $m \coloneqq |V(H_A)| = |V(H'_A)|$, and let $E_W \subseteq F \cup K'$ denote the set of edges with at least one edge incident to a vertex of $W$. Note that $|E_W| \le \Delta |W|$. Finally, let $Q \coloneqq V(H) \setminus W$. Note that $H'_A[Q]$ is a forest, which implies that the number of edges in $H'_A[Q]$ is at most $|Q|-1$.  It follows that
	\begin{align*}
		\Delta |W| \ge |E_W| &\ge |E(H'_A)| - (m-|W| -1) 
		\\&\ge \frac{3(m-2k)}{2} - (m - |W| -1) \tag{Since all except at most $2k$ vertices of $H'_A$ have degree at least $3$ in $H'_A$}
		\\\implies (\Delta-1) \cdot |W| &\ge \frac{m}{2} - 3k + 1
		\\\implies k \ge |W| &\ge \frac{m/2 - 3k+1}{(\Delta-1)}
	\end{align*}
	It follows that $m \le k \cdot (\Delta-1) + 3k - 1$, which is upper bounded by $\Oh(k^2)$, since $\Delta = \Oh(k)$. Since each vertex of $V(H_A)$ has degree at most $\Oh(k)$, it follows that the number of edges in $H$ (accounting for multiplicities) is $\Oh(k^3)$.
\end{proof}

\begin{lemma} \label{lem:fvs-pi1-equivalence}
	The original instance $\cI$ of $\Pi_1$ is a yes-instance iff the instance $(H_A, \prec, k)$ of $\Pi'_1$ as defined above is a yes-instance.
\end{lemma}
\begin{proof}
	It suffices to focus on the instance $\cI$ obtained after repeatedly applying reduction rules \Cref{redrule:fvs0}, \Cref{redrule:fvs1}, \Cref{redrule:fvs2}, \Cref{redrule:fvs3}, \Cref{redrule:fvs-4-isolated-cycles}. Consider a solution $\sigma = (S, X, Y)$ to $\Pi_1$ that shrinks at most $2$ disks corresponding to internal vertices on any reducible path, in accordance with \Cref{cl:fvs-longpath}. 
	
	\textbf{Forward direction.} We obtain a solution for the instance $(H, \prec, k)$ of $\Pi'_1$.	Let $W = V(H) \cap S$.
	If an original edge $pq$ is removed by $\sigma$ such that $p$ and $q$ have degree at least $3$ in $H$, then we add $pq$ to $K_o$, and define $f_o(e) = \{p, q\} \cap S$. Furthermore, from \Cref{obs:nu-properties}, $e$ can be removed by shrinking exactly $\nu(|pq|) = \#(e)$ endpoints. Note that $\#(e)$ also contributes toward $|W|$.
	
	Now fix a reducible path $\pi$ with endpoints $v_0 = p$ and $v_{\ell + 1} = q$ (with $v_0 \preceq v_{\ell+1}$) such that $\sigma$ removes at least one edge of $\pi$. Then, we add $e(\pi)$ to $K_r$. We consider different cases.
	\begin{itemize}
		\item If $\sigma$ removes the edge $e = v_0v_1$. We define $g_r(e) = \sfstart$. Via \Cref{obs:nu-properties}, if $\nu(|v_0v_1|) = \nu_{\sfstart} = 1$, then we can wlog assume that it shrinks $v_0$ (since it may also potentially eliminate other edges incident to $v_0$; whereas shrinking $v_1$ only affects $\pi$). Then, we set $w_r(e) = \{v_0\}$. 
		\\If $\nu(|v_0v_1|) = \nu(\sfstart) = 2$, then $\sigma$ must shrink both $v_0$ as well as $v_1$. In this case, we set $w_r(e) = \{v_0\}$, and $c_r = 1$. Note that we account $v_1$ in the form of $c_r(e)$ in the objective function.
		\item The case where $\sigma$ removes the edge $e = v_\ell v_{\ell+1}$ is symmetric to the previous case, and thus we omit it.
		\item If $\sigma$ removes an edge $e = v_i v_{i+1}$, where both $v_i$ and $v_{i+1}$ are internal degree-$2$ vertices of $\pi$. Without loss of generality, we can assume that $e = v_i v_{i+1}$ realizes the maximum distance $\dmax = \min_{1 \le j \le \ell-1} |v_j v_{j+1}|$. We set the value of $c_r(e)$ as $\nu_{\sfmax}$, since \Cref{obs:nu-properties} implies that $e$ can be removed by shrinking exactly $\nu_{\sfmax}$ disks.
	\end{itemize}
	We note that the subgraph induced by (i) edges of reducible paths $\pi$ such that no edge of $\pi$ is removed by $\sigma$, and (ii) original edges that are not removed by $\sigma$, induce a forest in the resulting graph $G(P, r)$. We let $F$ to be the set of corresponding reducible and original edges, which implies the graph $(V(H_A), F)$ is simple, and acyclic.	It is not difficult to verify that every disk that is shrunk by $\sigma$ contributes exactly $1$ to the objective of $\Pi'_1$. 
	
	\textbf{Backward direction.} Suppose $(H_A, \prec, k)$ is a yes-instance of $\Pi'_1$, and consider a solution $\sigma' = (W, (F, K_o, K_r), f_o, g_r, c_r, w_r)$ to $(H, \prec, k)$. We retrieve a solution $\sigma$ to the instance $\cI$ of $\Pi_1$. For every edge $e = pq \in K_o$, we shrink one or both of its endpoints, as specified by $f_o(e)$. For an edge $e = pq \in K_r$ with $p \preceq q$ corresponding to a reducible path $\pi = (p = v_0, v_1, \ldots, v_\ell, v_{\ell+1} = q)$, we consider the following different cases.
	\begin{itemize}
		\item If $g_r(e) = \sfstart$, and $\nu_{\sfstart} = 1$, then we shrink $p$ in the solution $\sigma$. If $\nu_{\sfstart} = 2$, then we shrink both $p$ and $v_1$ in $\sigma$.
		\item The case when $g_r(e) = \sfend$ is analogous to the previous case.
		\item If $g_r(e) = \sfmax$, then we shrink one or both endpoints of the edge $v_i v_{i+1}$ realizing the maximum euclidean distance along $\pi$, as specified by $c_r(e)$.  
	\end{itemize}
	Note that we can subdivide each reducible edge in $H$ appropriate number of times to obtain the original intersection graph $H$. However, the graph $(V(H_A), F)$ is a forest, and each edge of $F$ corresponds to an original edge, or a reducible path for which no edge is removed by $\sigma$. Furthermore, for every other reducible path $\pi$, $\sigma$ removes \emph{at least} one edge of $\pi$, disconnecting $\pi$. Therefore, the resulting intersection graph $G(P, r)$ corresponding to $\sigma$ is acyclic. Finally, it is not difficult to verify that the number of disks shrunk by $\sigma$ is at most $k$.
\end{proof}

The following observation is easy to verify, since a verification algorithm simply checks whether all the conditions in the definition of $\Pi_1$ are satisfied. 
\begin{observation} \label{obs:fvs-pi1-np}
	The problem $\Pi'_1$ is in \classNP.
\end{observation}

\fvspolykernel*
\begin{proof}
	Given an instance $\cI$ of \probAcyclicityDown, we obtain an equivalent instance $(H_A, \prec, k)$ of $\Pi'_1$ using \Cref{lem:fvs-pi1-equivalence}. Using \Cref{lem:fvs-pi1-compression-bound}, either (i) $|V(H_A)| = \Oh(k^2)$ and $|E(H_A)| = \Oh(k^3)$ -- otherwise we can conclude that $\cI$ is a no-instance. If at any point we conclude that $\cI$ is a no-instance, then we simply output a trivial no-instance of $\Pi_1$.
	
	Now assume that the number of vertices and edges in the instance $(H_A, \prec, k)$ is bounded. Furthermore, each edge is annotated with $\Oh(1)$ bits of information. \Cref{obs:fvs-pi1-np} implies that $\Pi'_1$ is in \classNP; whereas \Cref{thm:np-hardness-fvs} implies that \probAcyclicityDown is \classNPC. Therefore, there exists a polynomial-time reduction from the instance $(H, k)$ of $\Pi'_1$ that outputs an equivalent instance $\cI'$ of \probAcyclicityDown. Note that the size of the instance, and the parameter in the new instance $\cI'$ can be bounded by a polynomial in $k$. This implies the existence of a polynomial kernel for \probAcyclicityDown.
\end{proof}

\subsection{An FPT algorithm for \probAcyclicityMinDown via Partial Compression} We proceed as in \Cref{subsec:annotated} and summarize the situation after (i) replacing all isolated cycle components with isolated cycle vertices, (ii) replacing all reducible paths with reducible edges, and (iii) annotating all isolated cycle vertices, and all the edges (reducible as well original). We have an annotated multigraph $H_A$ with multiple parallel edges between vertices, and multiple self-loops on vertices. Note that $V(H_A) = V_b \uplus V_c$, where all vertices in $V_b$ have degree at least $3$, whereas all vertices in $V_c$ are isolated cycle vertices. Similarly, $E(H_A) = E_r \uplus E_o$, where $E_r$ and $E_o$ denote the sets of reducible and original edges respectively.
\begin{itemize}
	\item Each $e(\pi) = pq \in E_r$ with $p \preceq q$ corresponding to an reducible path $\pi$ is annotated with a tuple $(\dstart, \dend, \dmax)$. We refer to these values by $\dstart(e), \dend(e)$, and $\dmax(e)$ respectively.
	\item For each isolated cycle vertex $v(C)$ is annotated with $\#(v(C))$ and $\cost(v(C))$.
	\item Each $e = pq \in E_o$ is annotated with $d(e) = |pq|$.
\end{itemize}
Now we define the following annotated problem $\Pi'_2$, and prove that this is equivalent to $\Pi_2$.
\ \\\begin{tcolorbox}[breakable,colback=gray!5!white,colframe=gray!75!black]
	\begin{description}[]
		\item[Input:] An annotated multigraph $H_A$ with (possibly) multiple parallel edges between pairs of vertices, a total order $\prec$ on $V(H_A)$, a parameter $k$, and budget $\mu$.
		\item[Task:] Determine whether there exists a solution $\sigma_{\Pi'_2} = (W, (F, K_o, K_r), f_o, g_r, c_r, w_r, y_w, y_r)$, where 
		\begin{itemize}[leftmargin=3pt]
			\item $W \subseteq V(H_A)$ {\footnotesize \em $\triangleright$ denotes the set of vertices of degree at least $3$ that are shrunk}
			\item $(F, K_r, K_o)$ is a partition of $E(H_A)$ with $K_r \subseteq E_r$, and $K_o \subseteq E_o$ \newline{\ \qquad \footnotesize \em $\triangleright$ $F$ induces a forest, and the edges of $K_r \cup K_o$ are removed}
			\item $f_o: K_o \to W \cup (W \times W)$, {\footnotesize \em $\triangleright$ $f_o$ specifies how an original edge is removed}
			\item $g_r: K_r \to \{{\sf start, end, max}\}$, $c_r: K_r \to \{1, 2\}$, $n_r: K_r \to \{0, 1\}$, and $w_r: K_r \to W$ \newline{\footnotesize \em  $\triangleright$ $g_r$ encodes how an edge is removed, $c_r$ encodes how many degree-$2$ vertices are implicitly removed, and $w_r$ encodes the endpoints of reducible path to be removed}
			\item $\tilde{r}: V(H) \to [0, 1]$, and $y_e: E_r \to [0, 2]$ \newline{\footnotesize \em  $\triangleright$ $\tilde{r}$ encodes the new radius of a vertex in $V(H_A)$, and $y_e$ encodes the cost contributed by implicit shrinking of degree-$2$ vertex that helps in removal of an edge}
		\end{itemize}
		such that the following properties are satisfied.
		\begin{itemize}[leftmargin=3pt]
			\item The graph $(V(H_A), F)$ is simple (i.e., without self-loops or parallel edges), and acyclic,
			\item $\tilde{r}(p) <1$ iff $p \in W$; and for all remaining $q \in V(H) \setminus W$, $\tilde{r}(q) = 1$.
			\\$y_e(e) > 0$ iff $e \in K_r$; otherwise $y_e(e) = 0$ for $e \in E_r \setminus K_r$.
			\item  If $e = pq \in K_o$, then $d(e) \ge 2\alpha$, and \newline
			(i) $\emptyset \neq f_o(e) \subseteq \{p, q\}$, (ii) $f_o(e) \subseteq W$, and (iii) $\tilde{r}(p) + \tilde{r}(q) \le d(e)$ \\{\footnotesize \em $\triangleright$ at least one of the endpoints must be shrunk as specified by the annotation}
			\item Suppose $e = pq \in K_r$, with $p \preceq q$, and $g_r(e) = \sfstart$.
			\\Then $\dstart(e) \ge 2\alpha$, $w_r(e) = p$ with $p \in W$, and
			\newline (A) If $2\alpha \le \dstart(e) < 1 + \alpha$, then 
			\newline\hspace*{20mm}(i) $\tilde{r}(p) = \alpha$, (ii) $y_e(e) \ge \dstart(e)-\alpha$, (iii) $c_r(e) = 1$
			\newline (B) If $1+\alpha \le \dstart(e) \le 2$, then 
			\newline\hspace*{20mm} (i) $\tilde{r}(p) \ge 2-\dstart(e)$.
			\item Suppose $e = pq \in K_r$, with $p \preceq q$, and $g_r(e) = \sfend$.
			\\Then $\dend(e) \ge 2\alpha$, $w_r(e) = q$ with $q \in W$, and
			\newline (A) If $2\alpha < \dend(e) < 1+\alpha$, then 
			\newline\hspace*{20mm} (i) $\tilde{r}(q) = \alpha$, (ii) $y_e(e) \ge \dstart(e)-\alpha$, (iii) $c_r(e) = 1$
			\newline (B) If $1 \le \dend(e) \le 2$, then 
			\newline\hspace*{20mm}(i) $\tilde{r}(q) \ge 2-\dend(e)$.
			\item Suppose $e = pq \in K_r$, with $p \preceq q$, and $g_r(e) = \sfmax$.
			\\Then $\dmax(e) \ge 2\alpha$, and
			\\\hspace*{20mm}(i) $y_e(e) \ge 2 - \dmax(e)$, (ii) $w_r(e) = \emptyset$, (iii) $c(e) = \nu(\dmax(e))$
			\item $\displaystyle |W| + \sum_{e \in K_r} c_r(e) + |V_C| \le k$, and
			\item $\cost(\sigma_{\Pi'_2}) \coloneqq \sum_{p \in V(H_A)} y_w(p) + \sum_{e \in K_r} y_e(e) + \sum_{v(C) \in V_C} \cost(v(C)) \le \mu$. 
		\end{itemize}
	\end{description}
\end{tcolorbox}

The proof of the following lemma is almost identical to that of \Cref{lem:fvs-pi1-compression-bound}, except that we apply the argument for $H_1 = (V(H_A) \setminus V_C, E(H_A))$, and then account for at most $k$ vertices of $V_C$ separately (note that $|V_C| \le k$ via \Cref{redrule:fvs-4'-isolated-cycles})
\begin{lemma} \label{lem:fvs-pi2-compression-bound}
	If $(H_A, \prec, k, \mu)$ is a yes-instance of $\Pi'_2$, then $|V(H_A)| = \Oh(k^2)$, and $|E(H_A)| = \Oh(k^3)$, accounting for multiplicities.
\end{lemma}
Now we prove the equivalence between $\Pi_2$ and $\Pi'_2$.
\begin{lemma} \label{lem:fvs-pi2-equivalence}
	The original instance $\cI = (P, \alpha, k, \mu)$ of $\Pi_2$ is a yes-instance iff the instance $\cI' = (H, \prec, k, \mu)$ of $\Pi'_2$ as defined above is a yes-instance. Furthermore, if $\cI$ and $\cI'$ are yes-instances, then the optimal costs of the two instances are equal.
\end{lemma}
\begin{proof}
	The proof of equivalence for the two instances is almost identical to the proof of \Cref{lem:fvs-pi1-equivalence}, except that we need to account for the isolated cycle vertices $V_C$. Consider an isolated cycle $C$, and the corresponding vertex $v(C) \in V_C$. Since we seek to minimize the cost, it suffices to remove an edge $e = v_iv_{i+1}$ realizing the maximum euclidean distance between consecutive vertices, $\dmax(v(C))$. If $2\alpha < \dmax(v(C)) < 1+\alpha$, then a solution for $\cI$ can remove $e$ of $C$ only by shrinking at least two vertices on $C$. Otherwise, if $1+\alpha \le \dmax(v(C)) < 2$, then a solution can remove $e$ by shrinking one vertex.
	
	Now we argue about the equivalence between the costs of the solutions.
	
	\textbf{Forward direction.} Fix an optimal solution $\sigma = (S, r)$ for $\Pi_2$. First, \Cref{obs:dist-properties} implies that an edge $e = pq$ is removed iff $r(p) + r(q) \le |pq|$. Without loss of generality, we assume that for an edge $pq \in E(H)$ such that $q$ has degree $2$, and $p$ has degree at least $3$, it holds that $r(p) \le r(q)$. This is because shrinking the disk associated with $q$ only helps in disconnecting a reducible path; whereas shrinking the disk associated with $p$ can potentially help in removal of other cycles.
	
	Note that \emph{any} solution for $\Pi'_2$ assigns the optimal value for each isolated cycle $C$ by virtue of the annotated information. We therefore focus on the rest of the graph $H_1 = (V(H_A) \setminus V_C, E(H_A))$. First, we set $W = V(H_1) \cap S$, i.e., the set of vertices of degree at least $3$ (in $H$) that is shrunk by the solution. 
	
	Consider an original edge $e = pq$ that is removed by $\sigma$. We can set $f_o(e)$ to be equal to the set of shrunken endpoints of $pq$, and the $x$-values of $p$ and $q$ are set accordingly.
	
	Now consider a reducible path $\pi = (p = v_0, v_1, \ldots, v_\ell, v_{\ell+1} = q)$ in $G$ with $p \preceq q$ that does not exist in the final intersection graph $G(P, r)$. 
	
	Suppose the edge $e = pv_1$ is removed by $\sigma_2$. Then we set $g_r(e) = \sfstart$. Note that by assumption $p$ must be shrunk by $\sigma_2$, and $v_1$ is shrunk if $\dstart(e) < 1$. We define $\tilde{r}(p) = r(p)$, $y_e(e) = 1 - r(v_1)$. Furthermore, we set $c_r(e) = 1$ if $v_1$ is shrunk, otherwise we set $c_r(e) = 0$. The case when $e = v_\ell v_{\ell+1}$ is removed by $\sigma_2$ is symmetric. 
	
	Now consider the case where some edge $e = v_{i} v_{i+1}$ is removed, where $1 \le i < \ell$, i.e., $v_i$ and $v_{i+1}$ are degree-$2$ vertices. By optimality of $\sigma_2$, we can assume that $e$ is the edge realizing maximum euclidean distance $\dmax(e)$. In this case, we set $c_r(e)$ to be $1$ or $2$ depending on how many endpoints of $e$ are shrunk; and we set $y_e(e) = (1 - r(v_i)) + (1-r(v_{i+1})) = 2 - \dmax$. Finally, accounting for the costs of isolated cycle vertices, it follows that the cost of the constructed solution for $\Pi'_2$ is equal to that of $\sigma$.
	
	\textbf{Backward direction.} Given a solution $\sigma'$ for the instance $\cI'$ for $\Pi'_2$, we construct a solution $\sigma = (S, r)$ for the instance $\cI$ of $\Pi_1$ in a way analogous to the proof of \Cref{lem:fvs-pi1-equivalence}. Note that we also need to handle isolated cycles, which we do in a manner as described above. It is straightforward to verify that each step in the reconstruction process preserves the cost, which implies that the cost of $\sigma$ is equal to that of $\sigma'$. Thus, the costs of optimal solutions for $\cI$ and $\cI'$, if they exist, must be equal.
\end{proof} 
From \Cref{lem:fvs-pi2-equivalence}, we get the following theorem.

\begin{theorem} \label{thm:fvs-min-compression}
	There exists partial polynomial compression from \probAcyclicityMinDown into the problem $\Pi'_2$ as defined above, parameterized by $k$, the number of disks that we are allowed to shrink. 
	\\In particular, given an instance $(P, \alpha, k, \mu)$, there exists a polynomial-time algorithm that outputs an equivalent instance $(H_A, \prec, k, \mu)$ of $\Pi'_2$, where the number of vertices and edges (including multiplicities) in the multigraph $H_A$ is bounded by $\Oh(k^2)$ and $\Oh(k^3)$ respectively.
\end{theorem}

\paragraph{Solving the instance $(H_A, \prec, k, \mu)$ of $\Pi'_2$.}  It is straightforward to handle the isolated cycle vertices using the annotated information, so we focus on $H' = (V(H_A) \setminus V_C, E(H_A))$. From \Cref{lem:fvs-pi2-compression-bound}, $|V(H')| = \Oh(k^2)$, and $|E(H')| = \Oh(k^3)$. First, we guess a subset $F$ of size at most $|V(H')| - 1$ such that the graph $(V(H'), F)$ is simple and acyclic. Note that there are $\binom{\Oh(k^3)}{\Oh(k^2)} = k^{\Oh(k^2)}$ such guesses. Now fix one such guess $F$.

Let $K_r = E_r \setminus F$ and $K_o = E_o \setminus F$ denote the remaining reducible and original edges respectively. Let $W$ denote the set of shrinkable vertices, initialized to an empty set.  Consider an original edge $e = pq \in K_o$. If $d(e) < 1$, then we add both endpoints to the set $W$. Otherwise, if $d(e) \ge 1$, we \emph{guess} an endpoint, and add it to the set $W$.

Now, consider a reducible edge $e = pq \in K_r$ with $p \preceq q$. First, we \emph{guess} $g_r(e)$ from the set $\{\sfstart, \sfend, \sfmax\}$. We also set $w_r(e)$, $c_r(e)$ such that they are compatible with $g_r(e)$ as stipulated in the definition of $\Pi'_2$. Note that there are at most $\Oh(k^3)$ edges in $K_r \cup K_o$, and for each edge we need to guess from a constant number of possibilities. Thus, the total number of guesses is upper bounded by $2^{\Oh(k^3)}$. If $|W| + \sum_{e \in K_r} c_r(E) +  |V_C| > k$, then we discard this guess. 

Now, we write a linear program that consists of the appropriate linear constraints that is compatible with (i) the guesses for $(W, (F, K_r, K_o), g_r, c_r, w_r)$ that we have made, and (ii) the constraints imposed by the definition of $\Pi'_2$ b. Finally, we also add the following constraints for each $p \in V(H')$: (i) If $p \in W$, then $1 \ge r(p) \ge \alpha$, and (ii) If $p \not\in W$, then $r(p) = 1$.

Note that we have $|V(H')| + |E(H')| = \Oh(k^3)$ variables, and (i) at most two linear constraints for every edge in $K_r \cup K_o$, and (ii) at most two linear constraints for every $p \in V(H')$, and at most two constraints, which amounts to total $\Oh(k^3)$ constraints. Thus, using the algorithm of \Cref{prop:lpsolver}, we can find an optimal solution to the resulting linear program in time $(k^{3})^{\Oh(k^3)} \cdot \Oh(k^3) = 2^{\Oh(k^3 \log k)}$; or determine that the linear program is infeasible. Next, we add the cost for handling isolated cycle vertices to the cost of the LP. Finally, if in any of the $2^{\Oh(k^3)}$ guesses, if we find a solution (if one exists) of cost at most $\mu$, then such a solution can be translated for the original instance $\cI$ of $\Pi_2$ of the same cost using the approach from the proof of \Cref{lem:fvs-pi2-equivalence}, which implies that $\cI$ is a yes-instance. On the other hand, if each of the guesses that we do not discard, the linear program turns out to be infeasible, then we conclude that $\cI$ is a no-instance of $\Pi_2$. In fact, it is straightforward to see this algorithm in fact finds an \emph{optimal solution} for an instance of the \emph{optimization version} of the problem (if one exists). This leads to the following theorem.

\begin{theorem} \label{thm:fvs2-fpt}
	There exists an algorithm that takes as an input an instance $\cI = (P, k, \alpha, \mu)$ of \probAcyclicityMinDown, and in $2^{\Oh(k^3 \log k)} \cdot n^{\Oh(1)}$ time, either correctly concludes that it is a no-instance; or finds a solution of the smallest cost that is at most $\mu$.
\end{theorem}

\section{Subexponential FPT Algorithms}\label{sec:subexpFPT} 
Recall that an instance of \probEdgelessnessDown (resp.~\probAcyclicityDown/\probConnectivityFracDown) is denoted by $(P, k, \alpha)$, where $0 \le \alpha \le 1$. For \probEdgelessnessDown/{\sc Acyclicity}, the goal is to shrink \emph{at most} $k$ disks to radius $\alpha$, such that the resulting intersection graph is edgeless (resp.~acyclic). For \probConnectivityFracDown, the goal is to shrink \emph{at least} $k$ disks to radius $\alpha$ such that the resulting graph remains connected. In this section, we prove these problems are not only in \classFPT when parameterized by $k$, but also solvable by fixed-parameter algorithms whose time complexities are subexponential in $k$, if $\alpha$ is an absolute constant. For clarity, we make the dependence of the running time on $\alpha$ explicit.


\subsection{Clique-Grid Graphs and Nice Clique Tree Decompositions}

Towards the presentation of our algorithms, we introduce several notations, definitions and propositions.  Recall that for any set $P$ of points in the plane, $\cD(P, \unitvec)$ denotes the set of unit disks centered at each point of $P$, and $G(P, \unitvec)$ denotes the corresponding intersection graph. For the sake of simplicity and consistency with \cite{DBLP:journals/dcg/FominLPSZ19}, we suppose that disks are closed disks in this subsection. However, these results continue to hold for \emph{open disks}, which is required in the later subsections regarding \probEdgelessnessDown. 

Next, we present the definition of a clique-grid graph and its representation, and a known result concerning the computation of a representation for a unit disk graph.

\begin{definition}[{\bf Clique-Grid Graph}, Definition 3.1 in \cite{DBLP:journals/dcg/FominLPSZ19}] A graph $G$ is a {\em clique-grid graph} if there exists a function $f : V(G) \rightarrow [t]\times [t']$, for some $t,t' \in \mathbb{N}$, such that the following conditions are satisfied:
	\begin{enumerate}
		\item For all $(i, j) \in [t] \times [t']$, it holds that $f^{-1}(i, j)$ is a clique.
		\item For all $\{u,v\}\in E(G)$, it holds that if $f(u)=(i,j)$ and $f(v)=(i',j')$, then $|i-i'|\leq 2$ and $|j-j'|\leq 2$.
	\end{enumerate}
	Such a function $f$ is called a {\em representation} of $G$. A pair $(i,j)\in[t]\times[t']$ is called a {\em cell}.
\end{definition}

\begin{proposition}[Lemma 3.2 in \cite{DBLP:journals/dcg/FominLPSZ19}, Rephrased]\label{prop:represent} Let $P$ be a set of points on the plane. Then, a representation $f$ of the UDG $G(P, \unitvec)$ can be computed in polynomial time.
\end{proposition}

We proceed to present the definition of a notion of a tree decomposition that comes in handy in the context of clique-grid graphs.

\begin{definition}[{\bf $\ell$-NCTD}, Definition 3.5 in \cite{DBLP:journals/dcg/FominLPSZ19}]
	A tree decomposition ${\cal T} = (T, \beta)$ of a clique-grid graph $G$ with representation $f$ is a {\em nice $\ell$-clique tree decomposition} ({\em $\ell$-NCTD}) if for the root $r$ of $T$, it holds that $\beta(r) = \emptyset$, and for each node $v \in V(T)$, it holds that
	\begin{enumerate}
		\item There exist at most $\ell$ cells, $(i_1,j_1),\ldots,(i_\ell,j_\ell)$, such that $\beta(v)=\bigcup_{p=1}^\ell f^{-1}(i_p,j_p)$, and
		\item The node $v$ is of one of the following types:
		\begin{itemize}
			\item {\bf Leaf}: $v$ is a leaf in $T$ and $\beta(v)=\emptyset$.
			\item {\bf Forget}: $v$ has exactly one child $u$, and there exists a cell $(i, j) \in [t]\times[t']$ such that $f^{-1}(i, j) \subseteq \beta(u)$ and $\beta(v) = \beta(u) \setminus f^{-1}(i, j)$.
			\item {\bf Introduce}: $v$ has exactly one child $u$, and there exists a cell $(i, j) \in [t]\times[t']$ such that $f^{-1}(i, j) \subseteq \beta(u)$ and $\beta(v)\setminus f^{-1}(i, j) = \beta(u)$.
			\item {\bf Join}: $v$ has exactly two children, $u$ and $w$, and $\beta(v) = \beta(u) = \beta(w)$.
		\end{itemize}
	\end{enumerate}
\end{definition}

We have the following immediate observation.

\begin{observation}\label{obs:nctdToTd}
	Let ${\cal T}$ be a nice $\ell$-clique tree decomposition of a clique-grid graph $G$ with representation $f$. Let $c$ be the maximum size of a clique in $G$. Then, the width of $\cal T$ is bounded from above by $c\cdot\ell$.
\end{observation}

The main tool that we require in the context of the definitions presented so far is provided by the following proposition.

\begin{proposition}[Corollary 4.10 in \cite{DBLP:journals/dcg/FominLPSZ19}, Rephrased]\label{prop:nctd} Let $P$ be a set of points on the plane. Given  a representation $f$ of $H=G(P, \unitvec)$ and $\ell \in \mathbb{N}$, in time $2^{O(\ell)}\cdot n^{O(1)}$, one can either correctly conclude that $H$ contains a $\gamma\ell \times \gamma\ell$ grid as a minor, or compute a $5\ell$-NCTD of $H$, where $\gamma = \frac{1}{100\cdot 599^3}$. In the former case, one can also compute the corresponding minor model.
\end{proposition}

However, for \probConnectivityFracDown, we need an extension of the above-mentioned tool. Towards its statement, we introduce the following definitions. 

\begin{definition}[{\bf Well-Behaved Minor Model of a Grid}]\label{def:wellbehaved}
	Let $a,b\in\mathbb{N}$. Let $P$ be a finite set of points on the plane.  Suppose that $G(P, \unitvec)$ contains $H_{a\times b}$ as a minor, and let $\varphi$ be a corresponding minor model. Let $d_{a\times b}$ be the natural plane embedding of $\Gamma_{a\times b}$. Then, $\varphi$ is {\em well-behaved} if it satisfies the two following conditions.
	\begin{enumerate}
		\item\label{item:wellbehaved1} Let $C_1,C_2$ be two cycles in $\Gamma_{a\times b}$ such that  $C_2$ is embedded by $d_{a\times b}$ in the interior and boundary $C_1$. Then, the area enclosed by $\bigcup_{v\in V(C_2)}\varphi(v)$ is contained in the area enclosed by $\bigcup_{v\in V(C_1)}\varphi(v)$.
		\item\label{item:wellbehaved2} Let $C_1,C_2$ be two cycles in $\Gamma_{a\times b}$ such that $V(C_1)\cap V(C_2)=\emptyset$ and there do not exist vertices $u\in V(C_1),v\in V(\Gamma_{a\times b}),w\in V(C_2)$ such that $\{u,v\},\{v,w\}\in E(H_{a\times b})$. Then, the areas enclosed by $\bigcup_{v\in V(C_1)}\varphi(v)$ and $\bigcup_{v\in V(C_2)}\varphi(v)$ are disjoint.
	\end{enumerate}
\end{definition}

\begin{definition}[{\bf Enclosed Area}]
	Let $P$ be a finite set of points on the plane. Let ${\cal D}'\subseteq {\cal D}(P, \unitvec)$. Let $\cal C$ be the set of maximally connected regions of the plane once every point contained in at least one disk in ${\cal D}'$ is removed. Let $\widehat{\cal C}$ be the set of bounded regions in $\cal C$ (which includes all regions in $\cal C$ except for one). Then, the {\em area enclosed by ${\cal D}'$} is the set of every point $p$ contained in either at least one disk in ${\cal D}'$ or one region in $\widehat{\cal C}$.
\end{definition}

It can be verified that the proof of Corollary 4.10 in \cite{DBLP:journals/dcg/FominLPSZ19} yields the following extension of the~corollary.

\begin{proposition}[Extension of Corollary 4.10 in \cite{DBLP:journals/dcg/FominLPSZ19}]\label{prop:nctdExtended}
	Let $P$ be a set of points on the plane. Given  a representation $f$ of $H = G(P, \unitvec)$ and $\ell \in \mathbb{N}$, in time $2^{O(\ell)}\cdot n^{O(1)}$, one can compute either  a minor model of a $\gamma\ell \times \gamma\ell$ grid in $H$ that is well-behaved, or a $5\ell$-NCTD of $H$, where $\gamma = \frac{1}{100\cdot 599^3}$.
\end{proposition}

\subsection{Subexponential FPT Algorithm for \probEdgelessnessDown/{\sc Acyclicity}} \label{subsec:fpt-sqrt}

\noindent{\bf Elimination of Large Cliques.} Similarly to other subexponential-time fixed-parameter algorithms for problems on unit disk graphs (see, e.g., \cite{DBLP:journals/jacm/FominLS18,DBLP:journals/dcg/FominLPSZ19,DBLP:journals/jocg/ZehaviFLP021}), we also start by the elimination of large cliques in the input graph $H$. Further, similarly to the analysis of {\sc Vertex Cover} \iffvs(resp., {\sc Feedback Vertex Set})\fi, we can observe that every solution will have to select all but at most one \iffvs(resp., two)\fi of the vertices of any clique. However, unlike the other algorithms, we would not like to branch on which vertices these are, since, having this information at hand, we would not like to shrink (assign a radius of $\alpha$ instead of $1$) all but the corresponding at most two disks---this will require us to later deal with a disk graph (with two types of disks) instead of a unit disk graph. Fortunately, we observe that if we deal with a large enough clique (depending on $\alpha$), then we can simply say No. Specifically, we prove the following lemma.

\begin{lemma}\label{lem:vcfvsLargeDegree}
	Let $(P,\alpha,k)$ be an instance of \probEdgelessnessDown \iffvs(resp., \probAcyclicityDown)\fi. Suppose the $G(P, \unitvec)$ contains a vertex $v$ of degree larger than $(\frac{2}{\alpha}+1)^2-1$ (resp., $6(\frac{2}{\alpha}+1)^2-1$). Then, $(P,\alpha,k)$ is a no-instance.
\end{lemma}

Towards the proof of this lemma, we have the following results.

\begin{observation}\label{obs:is}
	A disk of radius $r$ can contain at most $(r/t)^2$ pairwise non-intersecting disks of radius $t$.
\end{observation}

\begin{proof}
	The area of a disk of radius $r$ is $\pi r^2$, and the area of a disk of radius $t$ is $\pi t^2$. Thus, a disk of radius $r$ can contain at most $(\pi r^2)/(\pi t^2)=(r/t)^2$ pairwise non-intersecting disks of radius~$t$.
\end{proof}

\begin{proposition}[\cite{DBLP:journals/dcg/AtminasZ18}]\label{prop:maxDeg}
	Let $\cal D$ be a set of unit disks. Let $D\in{\cal D}$. Suppose that the disks in ${\cal D}\setminus\{D\}$ are pairwise non-intersecting, and that each disk in ${\cal D}\setminus\{D\}$ intersects $D$. Then, $|{\cal D}\setminus\{D\}|\leq 5$.
\end{proposition}

We are now ready to prove \Cref{lem:vcfvsLargeDegree}.

\begin{proof}[Proof of \Cref{lem:vcfvsLargeDegree}]
	Notice that each disk in $N_G[v]$ is at distance at most $2$ from $v$. Let $A$ denote the disks in $N_G[v]$ after they are shrunk (from having radii $1$ to having radii $\alpha$). So, all disks in $A$ are contained in a single disk of radius $2+\alpha$ (whose center is the center of $v$). Let $I$ be a maximum-sized set of pairwise non-intersecting disks in $A$. By \Cref{obs:is}, $|I|\leq(\frac{2+\alpha}{\alpha})^2=(\frac{2}{\alpha}+1)^2$.  So, if $|N_G(v)|>(\frac{2}{\alpha}+1)^2-1$ (meaning that $|N_G[v]|>(\frac{2}{\alpha}+1)^2$), we have a no-instance of \probEdgelessnessDown.
	\iffvs
	Now, notice that each disk in $A$ intersects at least one disk in $I$ (by the maximality of $I$). Moreover, by \Cref{prop:maxDeg}, for each disk in $I$, there can be at most $5$ pairwise non-intersecting disks in $A$ that intersect it. However, $A$ cannot contain $3$ pairwise intersecting disks (which correspond to a triangle). Thus, we conclude that, unless we have a no-instance of \probAcyclicityDown, it must hold that $|A|\leq 5|I|+|I|\leq 6(\frac{2}{\alpha}+1)^2$, and hence $|N_G(v)|\leq 6(\frac{2}{\alpha}+1)^2-1$. 
	\fi
	This completes the proof.
\end{proof}

In particular, we have the following corollary of \Cref{lem:vcfvsLargeDegree}.

\begin{corollary}\label{lem:vcfvsLargeClique}
	Let $(P,\alpha,k)$ be an instance of \probEdgelessnessDown \iffvs(resp., \probAcyclicityDown)\fi. Suppose the UDG $G(P, \unitvec)$ contains a clique of size larger than $(\frac{2}{\alpha}+1)^2$ \iffvs(resp., $6(\frac{2}{\alpha}+1)^2$)\fi. Then, $(P,\alpha,k)$ is a no-instance.
\end{corollary}

\medskip
\noindent{\bf Dealing with a Large Grid as a Minor.} Similarly to most Bidimensional problems, when we have a large grid as a minor, we can directly solve the problem. In case of \probEdgelessnessDown \iffvs and \probAcyclicityDown\fi, we can simply say No:

\begin{observation}\label{obs:vcfvsGrid}
	Let $(P,\alpha,k)$ be an instance of \probEdgelessnessDown \iffvs or \probAcyclicityDown\fi. Suppose that $G(P, \unitvec)$ contains a $\lceil 2\sqrt{k+1}\rceil\times \lceil 2\sqrt{k+1}\rceil$  grid as a minor. Then, $(P,\alpha,k)$ is a no-instance.
\end{observation}

\begin{proof}
	Observe that, because $G(P, \unitvec)$  contains a $\lceil 2\sqrt{k+1}\rceil\times \lceil 2\sqrt{k+1}\rceil$  grid as a minor, then $G(P, \unitvec)$  contains \iffvs $k+1$ vertex-disjoint cycles (and, hence, also a matching of size $k+1$) \else a matching of size $k+1$\fi. Since we cannot shrink at most $k$ disks (irrespective of the value of $\alpha$) and thereby eliminate \iffvs all of these cycles (or edges of the matching)\else all the edges of the matching\fi, we conclude that $(P,\alpha,k)$ is a no-instance of \probEdgelessnessDown \iffvs or \probAcyclicityDown\fi.
\end{proof}

\medskip
\noindent{\bf Dynamic Programming for the Case of Bounded Treewidth.} Having eliminated both large cliques and large grids, it remains to deal with the case where the treewidth of the graph is small, specifically, bounded by a polynomial in $1/\alpha$ and $\sqrt{k}$. For this purpose, we have the following lemma, which can be proved by a straightforward dynamic programming algorithm over tree decompositions (similar to the one for \textsc{Vertex Cover}); see e.g., \cite{DBLP:books/sp/CyganFKLMPPS15}. Thus, we omit the details.
\begin{lemma}\label{lem:vcDP}
	Let $(P,\alpha,k)$ be an instance of \probEdgelessnessDown. Suppose that we are given a tree decomposition of $G(P, \unitvec)$ of width $w$. Then, $(P,\alpha,k)$ is solvable in time $\Oh(2^w\cdot n)$.
\end{lemma}
\iffvs
\begin{lemma}\label{lem:fvsDP}
	Let $(P,\alpha,k)$ be an instance of \probAcyclicityDown. Suppose that we are given a tree decomposition of $G(P, \unitvec)$ of width $w$. Then, $(P,\alpha,k)$ is solvable in time $w^{\Oh(w)}\cdot n$.
\end{lemma}

\begin{proof}
	The lemma can be proved by the design of a straightforward dynamic algorithm over tree decompositions (similar to the one for {\sc Feedback Vertex Set}; see, e.g., \cite{DBLP:books/sp/CyganFKLMPPS15}. Thus, the details are omitted.
\end{proof}
\fi

\noindent{\bf Conclusion of the Algorithm.} We are now ready to present our algorithm. To this end, let $(P,\alpha,k)$ be an instance of \probEdgelessnessDown \iffvs(resp., \probAcyclicityDown)\fi. Then, the algorithm performs the following operations:
\begin{enumerate}
	\item\label{step:vc1} Use the (polynomial-time) algorithm in \cite{clark1990unit} to compute the maximum size of a clique in $H \coloneqq G(P, \unitvec)$. If this size is larger than $(\frac{2}{\alpha}+1)^2$\iffvs (resp., $6(\frac{2}{\alpha}+1)^2$) \fi, then answer No.
	\item\label{step:vc2} Use the algorithm in  \Cref{prop:represent} to compute a representation of $H$. Then, use the algorithm in  \Cref{prop:nctd} with $\ell=\frac{2}{\gamma}\lceil\sqrt{k+1}\rceil$ to either correctly conclude that $H$ contains a $\gamma\ell \times \gamma\ell$ grid as a minor, or compute a $5\ell$-NCTD of $H$, ${\cal T}$, where $\gamma = \frac{1}{100\cdot 599^3}$. 
	\item\label{step:vc3} If the algorithm in \Cref{prop:nctd} concluded that $H$ contains a $\gamma\ell \times \gamma\ell$ grid as a minor, then answer No.
	\item\label{step:vc4} Otherwise, use the algorithm in \Cref{lem:vcDP} \iffvs(resp., \Cref{lem:fvsDP}) \fi to solve the problem using $\cal T$.
\end{enumerate}

The correctness of the algorithm directly follows from \Cref{lem:vcfvsLargeClique} and \Cref{obs:vcfvsGrid}. For the time complexity, first notice that Steps \ref{step:vc1}  and \ref{step:vc3} take polynomial time, and Step \ref{step:vc2} takes time $2^{\Oh(\sqrt{k})}\cdot n^{\Oh(1)}$. So, next suppose that we have $\cal T$ at hand. 
Now, suppose that we have reached Step \ref{step:vc4}. By \Cref{obs:nctdToTd} (and due to Step~\ref{step:vc1}), the width of $\cal T$ is at most $\Oh((\frac{1}{\alpha})^2\sqrt{k})$. Hence, by \Cref{lem:vcDP} \iffvs(resp., \Cref{lem:fvsDP})\fi, the algorithm runs in time $2^{\Oh((\frac{1}{\alpha})^2\sqrt{k})}\cdot n^{\Oh(1)}$ \iffvs (resp.~$(\frac{1}{\alpha}k)^{\Oh((\frac{1}{\alpha})^2\sqrt{k})}\cdot n^{\Oh(1)}$)\fi. Thus, we conclude the correctness of the following theorems.

\subexponentialvck*
\iffvs
\subexponentialfvsk*
\fi

\subsection{Subexponential FPT Algorithm for \probConnectivityFracDown} \label{subsec:fpt-subexp-conn}

We first set up some additional notation and convention that will be used throughout the section. Consider an instance $(P,\alpha,k)$ of \probConnectivityFracDown. Let $H \coloneqq G(P, \unitvec)$ denote the corresponding \emph{original} unit disk graph. For any subset $Q \subseteq P$, we use ${\cal S}_{Q}$ to denote the set of \emph{shrunken disks} corresponding to $Q$, i.e., the set of radius-$\alpha$ disks centered at each point in $Q$. We use the term {\em solution} to refer to $S\subseteq P$ such that assigning $x_p=\alpha$ for all $p\in S$ and $x_p=1$ for all $p\in P\setminus S$ yields a connected disk graph. Further, we use the term {\em optimal solution} to refer to a solution of maximum size.  When no confusion arises, we use $S$ also to refer to the set of disks of radius $1$ corresponding to its points. Recall that, for any set of disks $\cal D'$ (not necessarily of the same radii), we use $G(\cal D')$ to denote the corresponding intersection graph. Throughout the section, we slightly abuse the notation and equate a point in $P$ with its corresponding vertex in an intersection graph. 

\noindent{\bf Elimination of Large Cliques.} Unlike the cases  of \probEdgelessnessDown \iffvs or \probAcyclicityDown\fi, here the existence of a large clique does not mean that we have a no-instance (or a yes-instance). Fortunately, unlike these cases, here we can eliminate(rather than shrink) some of the vertices of a large clique.  Towards the elimination, we start with the the following simple lemma.

\begin{lemma}\label{lem:connElim2}
	Let $(P,\alpha,k)$ be an instance of \probConnectivityFracDown. Fix a $v\in P$ and $U \subseteq N_G[v]$. 
	
	\begin{enumerate}
		\item There exists some $Y_U \subseteq U$ of size at most $9$, such that each vertex in $U$ is adjacent in $G(\cD(U, \unitvec))$ to at least one vertex of $Y_U$.
		\item There exists some $Z_U \subseteq U$ of size at most $(\frac{2}{\alpha}+1)^2$, such that each vertex in $U$ is adjacent in $G(\cS_U)$ to at least one vertex of $Z_U$.
	\end{enumerate}
\end{lemma}

\begin{proof}
	Notice that each point in $N_G[v]$ is at distance at most $2$ from $v$. So, all disks in $\cD(U, \unitvec)$ are contained in a single disk of radius $3$ whose center is $v$. Since, from \Cref{obs:is}, the maximum-size of a subset of pairwise non-intersecting unit disks in $\cD(U, \unitvec)$ is at most $9$, the first part follows.
	
	The proof for the second part of the lemma is analogous to the first part, except that the arguments involve the set of shrunken disks $\cS_U$ of radius $\alpha$. Here, all the disks in $\cS_U$ are contained in a single disk of radius $2 + \alpha$ whose center is $v$. Then, \Cref{obs:is} implies that the maximum size of a subset of pairwise non-intersecting disks of radius $\alpha$ in $\cS_U$ is at most $(\frac{2+\alpha}{\alpha})^2 = (\frac{2}{\alpha}+1)^2$.
\end{proof}

Intuitively, the following \Cref{lem:connElim1} asserts that given a high-degree vertex, a substantial portion of a solution is contained in its neighborhood.
However, notice that even if we know which vertices of the neighborhood of some high-degree vertex to shrink, we cannot simply eliminate all of them, since we must ensure that they will stay connected to the rest of disks. Fortunately, with the help of the \Cref{lem:connElim2}, we will be able to argue that, for this purpose, it suffices to keep only some constantly many representatives from the shrunk set (and eliminate the rest).

\begin{lemma}\label{lem:connElim1}
	Let $(P,\alpha,k)$ be an instance of \probConnectivityFracDown and $v \in P$. Then, every optimal solution of $(P,\alpha,k)$ contains all except for at most $9(\frac{4}{\alpha}+1)^2+9$ disks from $N_{H}[v]$.
\end{lemma}

\begin{proof}
	Consider some optimal solution $S$. Let $U = N_{G}[v] \setminus S$ that are \emph{not} shrunk by $S$. By \Cref{lem:connElim2} part 1, there exists some $Y = Y_U \subseteq U$ of size at most $9$, such that each vertex in $U$ is adjacent to at least one vertex of $Y$ in $G(U, \unitvec)$. Let $U^\star \subseteq U \setminus Y$ be the maximum-size set of the form $N_H(y) \setminus Y$, over all $y \in Y$. It follows that, $|U^\star|\geq \frac{|U|-9}{9} \geq \frac{|U|}{9}-1$. 
	
	Due to the optimality of $S$, for every $u\in U^\star$, there must exist $w_u\in P$ satisfying the following.
	\begin{enumerate}
		\item If $w_u \not\in S$, then the unit disk $D(w_u, 1)$ (A) intersects the unit disk $D(u, 1)$, and (B) does not intersect any other unit disk $D(u', 1)$ where $u' \in U^\star$ and $u' \neq u$. In this case, observe that $D(w_u, \alpha)$ also does not intersect any other unit disk, since it is contained inside $D(w_u, 1)$. 
		\\or
		\item If $w_u \in S$, then the resulting shrunken disk, $D(w_u, \alpha)$ (A) intersects the unit disk $D(u, 1)$, and (B) does not intersect any other unit disk $D(u', 1)$ where $u' \in U^\star$ and $u' \neq u$.
	\end{enumerate}
	In either case, $D(w_u, 1)$ intersects $D(u, 1)$ and $D(w_u, \alpha)$ does not intersect any unit disk centered at a point of $U^\star$ other than $u$. This implies that the disks $D(w_u, \alpha)$ are distinct for each $u \in U^\star$. Let $W=\{D(w_u, \alpha): u\in U^\star\}$. From the optimality of $S$, we further obtain that the disks in $W$ are pairwise non-intersecting. However, all of the disks of $W$ are contained in a single disk of radius $4+\alpha$ (whose center is $v$). Thus, from \Cref{obs:is}, $|W|\leq (\frac{4+\alpha}{\alpha})^2=(\frac{4}{\alpha}+1)^2$. Since $|U^\star|=|W|$, this implies that $\frac{|U|}{9}-1\leq|U^\star|\leq (\frac{4}{\alpha}+1)^2$. So, $|U|\leq 9(\frac{4}{\alpha}+1)^2+9$.
\end{proof}

We now state an immediate corollary of \Cref{lem:connElim1}.

\begin{corollary}\label{cor:elimLarge}
	Let $(P,\alpha,k)$ be an instance of \probConnectivityFracDown. Let $v\in P$ be a vertex of degree at least $9(\frac{4}{\alpha}+1)^2+9+k$ in $H$. Then, $(P,\alpha,k)$ is a yes-instance.
\end{corollary}

Before stating the main lemma regarding high-degree vertices (and, hence, also large cliques), we set up some notation. Fix an instance $(P, \alpha, k)$ of \probConnectivityFracDown and fix some $v \in P$ be a vertex of degree at least $\mu \coloneqq 9(\frac{4}{\alpha}+1)^2+9+(\frac{2}{\alpha}+1)^2\cdot((\frac{1}{\alpha}+4)^2+1)+(\frac{2}{\alpha}+1)^2$ and at most $\nu \coloneqq 9(\frac{4}{\alpha}+1)^2+9+k$. We say that a subset $U \subseteq N_{G}[v]$ is \emph{special} if it satisfies the following two properties: (1) $|U| = |N_{G}[v]| - \mu$, and (2) Every disk in $\cS_U$ intersects at least one disk in $\cS_{N_{G}[v]} \setminus \cS_U$. 
Let ${\cal F}_v$ denote a family of reduced instances defined w.r.t. each special subset, defined as follows. ${\cal F}_v \coloneqq \LR{ (P \setminus U, \alpha, k - |U|) : U \text{ is special}}$. Now, we are ready to state our lemma using the notation we have thus defined.

\begin{lemma} \label{lem:connLargeDegree}
	$(P, \alpha, k)$ is a yes-instance if and only if for some $v \in P$ at least one of the instances in ${\cal F}_v$ is a yes-instance.
\end{lemma}
\begin{proof}
	In one direction, suppose that $(P,\alpha,k)$ is a yes-instance. By \Cref{lem:connElim1}, every optimal solution contains all except for at most $9(\frac{4}{\alpha}+1)^2+9$ disks from $N_{G}[v]$. So, in particular, there exists an optimal solution $S$ (recall that by convention adopted earlier, $S \subseteq P$ is the subset of points corresponding to the shrunken disks) and a set $X\subseteq N_{G(P, \unitvec)}[v]$ such that $|X|=|N_{G}[v]|-9(\frac{4}{\alpha}+1)^2-9$ and $X\subseteq S$. Let ${\cal D}' \coloneqq \cD(P \setminus S, \unitvec) \cup \cS_S$, i.e., $\cD'$ is the set of disks resulting by shrinking only the disks corresponding to the solution $S$ from radius $1$ to $\alpha$. We have the following claim.
	
	\begin{claim}
		There exists a subset $Y \subseteq X$ of size $|X|-(\frac{2}{\alpha}+1)^2\cdot((\frac{1}{\alpha}+4)^2+1)$ such that $G({\cal D}'\setminus \cS_Y)$ (as well as $G({\cal D}'\setminus \cS_{Y'})$ for any $Y' \subseteq Y$) is connected.
	\end{claim}
	
	\begin{proof}
		By \Cref{lem:connElim2} part 2, there exists some $Z \coloneqq Z_X \subseteq X$ of size at most $(\frac{2}{\alpha}+1)^2$, such that every vertex in $X$ is adjacent to at least one vertex in $Z$ in $H' = G(\cS_X)$. Let $C = N_{H'}(z)$ for some $z \in Z$. We argue that, if $|C| > (\frac{1}{\alpha}+4)^2$, then at least one shrunken disk centered at a point $u \in C$ can be removed from $\cD'$, while the resulting graph remaining connected. Once we show this statement, we can repeatedly apply it, showing that all but at most $(\frac{1}{\alpha}+4)^2$ shrunken disks centered at points in $C$ can be removed. In turn, this will yield the claim.
		
		So, suppose that $|C|>(\frac{1}{\alpha}+4)^2$. Targeting a contradiction, suppose that none of the disks in $C$ can be removed. Now, the proof proceeds similarly to the proof of \Cref{lem:connElim1}. We deduce that for every $u\in C$, there must exist $w_u\in P \setminus \LR{u}$ satisfying:
		\begin{enumerate}
			\item If $w_u \notin S$, then $D(w_u, 1)$ (A) intersects the shrunken disk $D(u, \alpha)$, and (B) does not intersect any other shrunken disk $D(u', \alpha)$ for any other $u' \in C$, $u' \neq u$, or
			\item If $w_u \in S$, then $D(w_u, \alpha)$ (A) intersects the shrunken disk $D(u, \alpha)$, and (B) does not intersect any other shrunken disk $D(u', \alpha)$ for any other $u' \in C$, $u' \neq u$. Note that in this case, (A) implies that $D(w_u, 1)$ also intersects $D(u, \alpha)$, since the former disk contains $D(u, \alpha)$. 
		\end{enumerate}
		In either case, $D(w_u, 1)$ intersects $D(u, \alpha)$ and $D(w_u, \alpha)$ does not intersect any other shrunken disks centered at points of $C$ other than $u$. This implies that the points $w_u$ are distinct for each $u \in C$. Let $W=\{D(w_u, \alpha): u\in C\}$. Because no shrunken disk centered at a point in $C$ can be removed, we further conclude that the disks in $W$ are pairwise non-intersecting -- otherwise, they can be used for connectivity via each other while still being able to remove some shrunken disk. However, all of the disks in $W$ are contained in a single disk of radius $1+4\alpha$ (whose center is $z$). Thus, from \Cref{obs:is}, $|W|\leq (\frac{1+4\alpha}{\alpha})^2=(\frac{1}{\alpha}+4)^2$. Since $|C|=|W|$, we have reached a contradiction.
	\end{proof}
	
	Consider $Y \subseteq X$ as guaranteed by the claim. Observe that all disks in $\cS_Y$ are contained in a disk of radius $2+\alpha$ (whose center is $v$). So, by \Cref{obs:is}, there exists $Z\subseteq Y$ of size $(\frac{2+\alpha}{\alpha})^2=(\frac{2}{\alpha}+1)^2$ such that every disk in $\cS_Y$ intersects at least one disk in $\cS_Z$. Finally, let $Q = Y \setminus Z$. Note that $|Q|=|N_{G}[v]|-9(\frac{4}{\alpha}+1)^2-(\frac{2}{\alpha}+1)^2\cdot((\frac{1}{\alpha}+4)^2+1)-(\frac{2}{\alpha}+1)^2$, and hence, by the choice of $Z$, $Q$ is special. Further, $Q \subseteq X\subseteq S$ and hence $|S\setminus Q|\geq k-|Q|$, and, due to the choice of $Y$, $S\setminus Q$ is a solution for $(P \setminus Q,\alpha,k-|Q|)$. So, at least one of the instances in ${\cal F}_v$ is a yes-instance.
	
	In the other direction, suppose that at least one of the instances in ${\cal F}_v$, say, $(P-U,\alpha,k-|U|)$, is a yes-instance, where $U$ is special. Let $S$ be a solution of this instance of size $k-|U|$. Let $S'=S\cup U$. So, $|S'|=k$. Further, since $U$ is special, it follows that every disk in $\cS_U$ intersects at least one disk in $\cS_{N_{G}[v]} \setminus \cS_U$. However, this means that $S'$ is a solution of $(P,\alpha,k)$. Thus, the proof of the lemma is complete.
\end{proof}

\medskip\noindent{\bf Reduction for an Area Enclosed by Unshrinkable Disks.} With respect to an instance $(P,\alpha,k)$ of \probConnectivityFracDown, every $p \in P$ such that shrinking the corresponding disk (i.e., changing its radius from $1$ to $\alpha$) splits the resulting intersection graph into at least two connected components is marked as ``\emph{unshrinkable}''. We say that such an instance is {\em marked}. Additionally, once a point $p$ is marked as unshrinkable, it never gets unmarked (even if some points are removed from the instance so that $p$ no longer has the property that shrinking the disk splits the graph). Furthermore, we say that a solution is {\em compatible} if it does not contain any point that is marked as unshrinkable. Additionally, we have the following definition. Since any point marked as unshrinkable will have its radius fixed at $1$, we will slightly abuse the notation and say that the corresponding (unit) \emph{disk} is marked as unshrinkable.

\begin{definition}[{\bf Unshrinkable and Non-Empty Closed Walks}] 
	Let $(P,\alpha,k)$ be a marked instance of \probConnectivityFracDown. Let $C \subseteq P$ be a closed walk in $G(P, \unitvec)$, and let $A$ be the area enclosed by $C$.  Let $R \subseteq P$ be subset of points such that the corresponding unit disks intersect the area $A$ but are not contained in $A$ (i.e., the disk is not a subset of the interior of $A$), and let $Q \subseteq P \setminus R$ be the set of points such that the corresponding unit disks are contained inside the area $A$. We refer to $R$ and $Q$ are the {\em relevant set} of $C$ and the {\em interior set} of $C$, respectively. If all of the points in $R$ are marked as unshrinkable, then $C$ is said to be {\em unshrinkable} (and otherwise it is {\em shrinkable}). If $Q\neq\emptyset$, then $C$ is said to be {\em non-empty} (and otherwise it is {\em empty}). Moreover, if at least one point in $Q$ is shrinkable, then $C$ is said to be {\em reducible}.
\end{definition}

Observe that a non-empty shrinkable closed walk might not be reducible. First, we present a lemma concerning non-empty unshrinkable closed walks.

\begin{lemma}\label{lem:nonEmptyCycUnshrinkable}
	Let $(P,\alpha,k)$ be a marked instance of \probConnectivityFracDown such that $H = G(P, \unitvec)$ is connected. Let $C$ be a non-empty unshrinkable closed walk in $H$. Then, $(P,\alpha,k)$ has a compatible solution of size at most $k$ if and only if $(P\setminus Q,\alpha,k)$ has a compatible solution of size at most $k$, where $Q$ is the interior set of $C$.
\end{lemma}
\begin{proof}
	Let $R$ be the relevant set of $C$. In one direction, suppose that $(P,\alpha,k)$ has a compatible solution $S$ of size at most $k$. Then, $S\cap R=\emptyset$.  Let ${\cal D}'$ be the collection of disks obtained from $\cD(P, \unitvec)$ by shrinking the disks centered at $S$ to radius $\alpha$ (and the radii of all other disks are unit). Since $S\cap R=\emptyset$, $C$ is also a closed walk in $G({\cal D}')$. Further, by the definition of an interior set and since $G({\cal D}')$ is connected, for every two vertices $u,v$ of   ${\cal D}'\setminus Q$, there exists a path between $u$ and $v$ in $G({\cal D}')$ that does not traverse any vertex in $Q$---indeed, if we traverse a vertex in $Q$, then the path goes inside and outside the region enclosed by $C$, and so that part of the path can be simply replaced by a subpath of $C$ itself. This implies that after the removal of $Q$ from $G({\cal D}')$, the graph stays connected. In particular, we get that $S$ is also a compatible solution of size at most $k$ for $(P\setminus Q,\alpha,k)$.
	
	In the other direction, suppose that $(P\setminus Q,\alpha,k)$ has a compatible solution $Z$ of size at most $k$. Then, $Z\cap (R\setminus Q)=\emptyset$. Let $\cS_Z$ be the collection of disks that contains disks of radius $\alpha$ centered at each point in $Z$, and disks of radius $1$ centered at each point in $(P \setminus Q) \setminus Z$. Observe that, by the definition of an interior set and since $H$ is connected, for every two vertices $u,v$ of $Q$ there exists a path between $u$ and $v$ in $H$ that does not traverse any vertex outside $R$---indeed, if we traverse such a vertex, then the path goes outside and inside the region enclosed by $C$, and so that part of the path can be simply replaced by a subpath of $C$ itself. Because $Z\cap (R\setminus Q)=\emptyset$, this means that when we add $Q$ to $G(\cS_Z)$, the graph stays connected. In particular, we get that $Z$ is also a compatible solution of size at most $k$ for $(P,\alpha,k)$.
\end{proof}

Second, we present a lemma concerning non-empty shrinkable (in fact, reducible) closed walks.

\begin{lemma}\label{lem:nonEmptyCycShrinkable}
	Let $(P,\alpha,k)$ be a marked instance of \probConnectivityFracDown such that $H = G(P, \unitvec)$ is connected. Let ${\cal C}$ be a collection of reducible closed walks in $H$ whose relevant sets are pairwise disjoint. Then, there exists a compatible solution of $(P,\alpha,k)$ (possibly of size smaller than $k$) that contains one vertex from the interior set of each of the closed walks in $\cal C$.
\end{lemma}

\begin{proof}
	We prove this claim by induction on $i=|{\cal C}|$. In the base case, when $i=0$, the claim is trivially true (because $H$ is connected). Now, suppose correctness for $i-1\geq 0$, and let us prove it for $i$. Let ${\cal C}'$ be some subcollection of $\cal C$ of size $i-1$. Then, by the inductive hypothesis, there exists a compatible solution, say, $S'$, of $(P,\alpha,k)$ that contains one vertex from the interior set of each of the closed walks in ${\cal C}'$, and, without loss of generality, suppose that it contains no other vertex (so, $|S'|=i-1$). Let $C$ be the closed walk in ${\cal C}\setminus{\cal C}'$. Since $C$ is reducible, there exists a shrinkable vertex, say, $s$, in the interior set of $C$. Let ${\cal D}' = \cD(P\setminus S', \unitvec) \cup \cS_{S'}$ be the set of disks obtained after shrinking the disks in $S'$ to radii $\alpha$. So, $G({\cal D}')$ is connected. We claim that $S'\cup\{s\}$ is a compatible solution of $(P,\alpha,k)$, which will conclude the proof. To this end, it suffices to show that when we shrink $s$ in ${\cal D}'$, obtaining a set of disks ${\cal D}''$, the graph (turned from $G({\cal D}')$ to $G({\cal D}'')$) stays connected.
	
	Let $R$ be the relevant set of $C$. Let $\widehat{\cal D} \coloneqq \cD(P \setminus \LR{s}, \unitvec) \cup D(s, \alpha)$ be the set of disks obtained from $\cal D(P, \unitvec)$ by shrinking the disk centered at $s$ to radius $\alpha$.  Since $H$ is connected and $s$ is shrinkable, we get that $G(\widehat{\cal D})$ is connected. Let $F$ be the subgraph of $G(\widehat{\cal D})$ induced by $R$ where $s$ is shrunk.
	Notice that, because $G(\widehat{\cal D})$ is connected, it follows that $F$ is connected---for every two vertices $u,v\in R$, there exists a path in $G(\widehat{\cal D})$, and if it uses vertices outside $R$, then it goes outside and inside $C$, and we can replace these parts of the path by subpaths of $C$ itself. Moreover, since $S'\cap R=\emptyset$ (by the choice of $S'$ and since the relevant sets of the closed walks in $\cal C$ are pairwise disjoint), $F$ also exists as a subgraph in $G({\cal D}'')$. Now, for every two vertices $u,v$ in $G({\cal D}')$ that do not belong to $F$, there exists a path in $G({\cal D}')$ that does not use any vertex from the interior of $C$ (and in particular it does not use $s$)---indeed, since $G({\cal D}')$ is connected, there exists a path between $u,v$ in $G({\cal D}')$, and if it is uses vertices from the interior of $C$, then we can replace these parts of the path by subpaths of $C$ itself. In turn, this means that $G({\cal D}'')$ is connected, and hence the proof is complete.
\end{proof}

\medskip
\noindent{\bf Dealing with a Large Grid as a Well-Behaved Minor.} Ideally, similarly to most Bidimensional problems, when we have a large grid as a minor, we would have liked to directly solve the problem. However, for \probConnectivityFracDown, we are not able to do this. Still, if the minor is well-behaved (as in \Cref{def:wellbehaved}) {\em and} we are given a marked instance that does not contain any non-empty unshrinkable closed walk, then we are able to do this. However, it is unclear how to turn the instance to have this property. In either case, we observe that we can focus only on some particular set of closed walks. We identify this set in the following definition.

\begin{definition}[{\bf $C_{i,j},W_{i,j}$}]\label{def:specialCycs}
	Let $(P,\alpha,k)$ be a marked instance of \probConnectivityFracDown. Let $\Gamma = \Gamma_{t\times t}$ grid with $V(\Gamma)=\{v_{i,j}: i,j\in[t]\}$, where $t=7\lceil \sqrt{k}\rceil$. Suppose that $G(P, \unitvec)$ contains $\Gamma$ as a well-behaved minor, and let $\varphi$ be a corresponding minor model. Then, we define the following notations. For every $i,j\in[\lceil\sqrt{k}\rceil]$, let $C_{i,j}$ be the cycle (on 8 vertices) in $\Gamma$ on vertex set $\{v_{a,b}: (a,b)\in\{(1+x+7(i-1), 1+y+7(j-1)): x,y\in[3]\}\setminus\{(3+7(i-1),3+7(j-1))\}\}$. Additionally, for every $i,j\in[\lceil\sqrt{k}\rceil]$, let $W_{i,j}$ be some closed walk in $G(P, \unitvec)$ that traverses all vertices in $\bigcup_{v\in V(C_{i,j})}\varphi(v)$ (it might traverse the same edge twice consecutively).\footnote{The existence of such a closed walk follows from the definition of a tree decomposition.}
\end{definition}

Now, we present the lemma that handles large grids as well-behaved minors, if they do not contain certain cycles.

\begin{lemma}\label{lem:connGrid}
	Let $(P,\alpha,k)$ be a marked instance of \probConnectivityFracDown. Let $\Gamma \coloneqq \Gamma_{t\times t}$ grid with $V(\Gamma)=\{v_{i,j}: i,j\in[t]\}$, where $t=7\lceil \sqrt{k}\rceil$. Suppose that $G(P, \unitvec)$ is connected and contains $\Gamma$ as a well-behaved minor for $t=7\lceil \sqrt{k}\rceil$, and let $\varphi$ be a corresponding minor model. For every $i,j\in[\lceil\sqrt{k}\rceil]$, let $C_{i,j}$ and $W_{i,j}$ be as defined in \Cref{def:specialCycs}. Additionally, suppose that for every $i,j\in[\lceil\sqrt{k}\rceil]$, $W_{i,j}$ is not a non-empty unshrinkable closed walk. Then, $(P,\alpha,k)$ is a yes-instance.
\end{lemma}

\begin{proof}
	We start with the claim that each of the closed walks $W_{i,j}$ is shrinkable:
	
	\begin{claim}\label{claim:connGrid1}
		For every $i,j\in[\lceil\sqrt{k}\rceil]$, $W_{i,j}$ is shrinkable.
	\end{claim}
	
	\begin{proof}
		Fix some $i,j\in[\lceil\sqrt{k}\rceil]$. Targeting a contradiction, suppose that $W_{i,j}$ is unshrinkable.
		Let $C$ be the cycle (on 4 vertices) in $\Gamma$ on vertex set $\{v_{a,b}: (a,b)\in\{(1+x+7(i-1), 1+y+7(j-1)): x,y\in[2]\}\}$. Observe that $C$ is embedded inside $C_{i,j}$ by $d_{t\times t}$. So, by Property \ref{item:wellbehaved1} in \Cref{def:wellbehaved}, the area enclosed by $\bigcup_{v\in V(C)}\varphi(v)$ is contained in the area enclosed by $V(W_{i,j})$. In particular, this means that $\varphi(v_{3+7(i-1),3+7(j-1)})$ is a subset of the interior set of $W_{i,j}$. Thus, $W_{i,j}$ is non-empty. However, this is a contradiction, since we suppose that $W_{i,j}$ is not a non-empty unshrinkable closed walk.
	\end{proof}
	
	Additionally, for every $i,j\in[\lceil\sqrt{k}\rceil]$, let $C^\star_{i,j}$ be the cycle (on 16 vertices) in $\Gamma$ on vertex set $\{v_{a,b}: (a,b)\in\{(x+7(i-1), y+7(j-1)): x,y\in[5]\}\setminus(V(C_{i,j})\cup\{(3+7(i-1),3+7(j-1))\})\}$. Additionally, for every $i,j\in[\lceil\sqrt{k}\rceil]$, let $W^\star_{i,j}$ be some closed walk in $G(P, \unitvec)$ that traverses all vertices in $\bigcup_{v\in V(C_{i,j})}\varphi(v)$. We claim that each of these closed walks is reducible:
	
	\begin{claim}\label{claim:connGrid2}
		For every $i,j\in[\lceil\sqrt{k}\rceil]$, $W^\star_{i,j}$ is reducible.
	\end{claim}
	
	\begin{proof}
		By Property \ref{item:wellbehaved1} in \Cref{def:wellbehaved}, the area enclosed by $V(W_{i,j})$ is contained in the area enclosed by $V(W^\star_{i,j})$. Moreover, by  \Cref{claim:connGrid1}, $W_{i,j}$ is shrinkable. Thus, we deduce that $W^\star_{i,j}$ is reducible.
	\end{proof}
	
	Additionally, for the application of \Cref{lem:nonEmptyCycShrinkable}, we also need the following claim.
	
	\begin{claim}\label{claim:connGrid3}
		For every $i,j,i',j'\in[\lceil\sqrt{k}\rceil]$ such that $(i,j)\neq(i',j')$, the relevant sets of $W^\star_{i,j}$ and $W^\star_{i',j'}$ are disjoint.
	\end{claim}
	
	\begin{proof}
		Due to our definition of the cycles $C^\star_{a,b}$, $a,b\in[\lceil\sqrt{k}\rceil]$, the claim is an immediate consequence of Property \ref{item:wellbehaved2} in \Cref{def:wellbehaved}.
	\end{proof}
	
	We now return to the proof of the lemma. Due to \Cref{claim:connGrid2} and \Cref{claim:connGrid3}, we can use \Cref{lem:nonEmptyCycShrinkable} to derive that, for every $i,j\in[\lceil\sqrt{k}\rceil]$, we can pick some vertex, say, $w_{i,j}$, from the relevant set of $W_{i,j}$, such that there exists compatible solution, say, $S$, of $(P,\alpha,k)$ that, for every $i,j\in[\lceil\sqrt{k}\rceil]$, contains $w_{i,j}$.  
	However, this means that $|S|\geq k$. Thus, the proof is complete.
\end{proof}

\medskip
\noindent{\bf Dynamic Programming for the Case of Bounded Treewidth.} Having eliminated both large cliques and large grids (under certain assumptions), we deal with the case where the treewidth of the graph is small, specifically, bounded by a polynomial in $1/\alpha$ times $k^{1-\epsilon}$ for some fixed $\epsilon>0$. For this purpose, we have the following lemma.

\begin{lemma}\label{lem:connDP}
	Let $(P,\alpha,k)$ be a marked instance of \probConnectivityFracDown. Suppose that we are given a tree decomposition of $G(P, \unitvec)$ of width $w$. Then, $(P,\alpha,k)$ is solvable (where we seek only compatible solutions) in time $w^{\Oh(w)}\cdot n$.
\end{lemma}

\begin{proof}
	The lemma can be proved by the design of a straightforward dynamic algorithm over tree decompositions; see, e.g., \cite{DBLP:books/sp/CyganFKLMPPS15}. Thus, the details are omitted.
\end{proof}

\medskip
\noindent{\bf Conclusion of the Algorithm.} We are now ready to present our algorithm, which we call $\mathsf{ALG}$. The algorithm will be a recursive algorithm, working on marked instances (so, when we make a recursive call, we preserve the annotations that have been done so far). To present its execution, let $(P,\alpha,k)$ be an instance of \probConnectivityFracDown. Then, $\mathsf{ALG}$ performs the following operations:
\begin{enumerate}
	\item\label{step:cc1} If $P=\emptyset$, return Yes.
	\item\label{step:cc2} If $H = G(P, \unitvec)$ is not connected, return No.
	\item\label{step:cc3} For every $p \in P$ such that shrinking its radius from $1$ to $\alpha$ splits the resulting intersection graph into at least two connected components, mark it as unshrinkable (unless it is already marked as such).
	\item\label{step:cc4} Use the (polynomial-time) algorithm in \cite{clark1990unit} to compute a clique $K$ of maximum size in $H$.
	\item\label{step:cc5} If $|V(K)|$ is at least $9(\frac{4}{\alpha}+1)^2+9+k+1$, then answer Yes.
	\item\label{step:cc6} Else, if $|V(K)| \ge k^{1/4}+9(\frac{4}{\alpha}+1)^2+9+(\frac{2}{\alpha}+1)^2\cdot((\frac{1}{\alpha}+4)^2+1)+(\frac{2}{\alpha}+1)^2$, then:
	\begin{enumerate}
		\item Let $v\in V(K)$.
		\item Let ${\cal F}_v=\{(P\setminus U,\alpha,k-|U|): U \text{ is special}\}$ (recall that the \emph{special} property is defined before \Cref{lem:connLargeDegree}) 
		\item Call $\mathsf{ALG}$ recursively on each of the instances in ${\cal F}_v$. If at least one of these calls returns Yes, then return Yes, and otherwise return No. 
	\end{enumerate}
	
	\item\label{step:cc7} Use the algorithm in \Cref{prop:represent} to compute a representation of $H$. Let $\Gamma$ be the $t\times t$ grid with $V(\Gamma)=\{v_{i,j}: i,j\in[t]\}$, where $t=7\lceil \sqrt{k}\rceil$. Then, use the algorithm in Proposition \ref{prop:nctd} with $\ell=\frac{7}{\gamma}\lceil\sqrt{k}\rceil$ (then, $t=\gamma\ell$) to  compute either  a minor model $\varphi$ of $\Gamma$  in $H$ that is well-behaved, or a $5\ell$-NCTD of $H$, ${\cal T}$, where $\gamma = \frac{1}{100\cdot 599^3}$. 
	
	\item\label{step:cc8} If the algorithm in \Cref{prop:nctdExtended} returned the minor model $\varphi$, then:
	\begin{enumerate}
		\item\label{step:cc8a}  For every $i,j\in[\lceil\sqrt{k}\rceil]$, let $C_{i,j}$ and $W_{i,j}$ be as defined in \Cref{def:specialCycs}. 
		\item\label{step:cc8b} If there exist $i,j\in[\lceil\sqrt{k}\rceil]$ such that $W_{i,j}$ is a non-empty unshrinkable closed walk, then call $\mathsf{ALG}$ recursively on $(P\setminus Q,\alpha,k)$ where $Q$ is the interior set of $W_{i,j}$, and return its answer.
		\item\label{step:cc8c} Otherwise, answer Yes. 
	\end{enumerate} 
	
	\item\label{step:cc9} Otherwise, use the algorithm in \Cref{lem:connDP} to solve the problem using $\cal T$.
\end{enumerate}

We first consider the correctness of the algorithm.

\begin{lemma}\label{lem:algCorrect}
	Algorithm $\mathsf{ALG}$ correctly solves \probConnectivityFracDown.
\end{lemma}

\begin{proof}
	For correctness, we need the slightly stronger claim that $\mathsf{ALG}$ solves \probConnectivityFracDown even if these instances are marked (and we seek only compatible solutions). Since initially no disk is marked, this will yield the lemma. The proof of this claim is by induction on the size $s=|P|+k$ of the instance. In the base case, when $s=0$, then $P$ is empty, and correctness follows directly from Step \ref{step:cc1}. Now, suppose that the claim is correct for instances of size at most $s-1\geq 0$, and let us prove it for~$s$.
	
	First, notice that the correctness of Steps \ref{step:cc2} and \ref{step:cc3} is immediate. The correctness of Step \ref{step:cc5} follows from \Cref{cor:elimLarge}. Next, the correctness of Step \ref{step:cc6} follows from \Cref{lem:connLargeDegree} and the inductive hypothesis. Additionally, the correctness of Step \ref{step:cc8b} follows from \Cref{lem:nonEmptyCycUnshrinkable} and the inductive hypothesis, and the correctness of Step \ref{step:cc8c} follows from \Cref{lem:connGrid}. Lastly, the correctness of Step \ref{step:cc9} follows from Lemma \Cref{lem:connDP}. Thus, we conclude correctness for $s$.
\end{proof}

Next, we consider the time complexity of the algorithm.

\begin{lemma}\label{lem:algTime}
	Algorithm $\mathsf{ALG}$, called with instance of \probConnectivityFracDown runs in time $(\frac{1}{\alpha}k)^{\Oh((\frac{1}{\alpha})^2k^{3/4})}\cdot n^{\Oh(1)}$.
\end{lemma}

\begin{proof}
	Observe that Steps \ref{step:cc1}--\ref{step:cc5} require polynomial time. For Step \ref{step:cc6}, observe that we derive the following recursion, where $T(n,k)$ denotes the runtime of the algorithm on an $n$-vertex graph (i.e., $|P|=n$) and parameter $k$:
	\[\begin{array}{ll}
		T(n,k) &\leq \displaystyle{\max_{k^{1/4}+c_{\alpha}\leq i\leq k+d_{\alpha}} {i\choose i-c_{\alpha}}\cdot T(n,k-(i-c_{\alpha})) + n^{\Oh(1)}.}
	\end{array}\]
	where $c_{\alpha}=9(\frac{4}{\alpha}+1)^2+9+(\frac{2}{\alpha}+1)^2\cdot((\frac{1}{\alpha}+4)^2+1)+(\frac{2}{\alpha}+1)^2$ and $d_{\alpha}=9(\frac{4}{\alpha}+1)^2+9$. From this recursion, supposing that Steps \ref{step:cc7}--\ref{step:cc9} can be performed in time $(\frac{1}{\alpha}k)^{\Oh((\frac{1}{\alpha})^2k^{3/4})}\cdot n^{\Oh(1)}$, we derive that $T(n,k) \leq (\frac{1}{\alpha}k)^{\Oh((\frac{1}{\alpha})^2k^{3/4})}\cdot n^{\Oh(1)}$. So, it remains to consider Steps \ref{step:cc7}--\ref{step:cc9}.
	
	According to Propositions \ref{prop:represent} and \ref{prop:nctdExtended}, Step \ref{step:cc7} can be performed in time $2^{\Oh(\sqrt{k})}n^{\Oh(1)}$. In Step \ref{step:cc8}, we perform computations in polynomial time and possibly make a single recursive call on a smaller instance. Lastly, consider Step \ref{step:cc9}. Due to Steps \ref{step:cc5} and \ref{step:cc6}, the maximum size of a clique in $G(P, \unitvec)$ (once we reach Step \ref{step:cc9}) is $\Oh((\frac{1}{\alpha})^2+k^{1/4})$. So, by \Cref{obs:nctdToTd}, the width of $\cal T$ is  $\Oh((\frac{1}{\alpha})^2k^{1/2}+k^{3/4})\leq\Oh((\frac{1}{\alpha})^2k^{3/4})$. Hence, by \Cref{lem:connDP}, the algorithm runs in time $(\frac{1}{\alpha}k)^{\Oh((\frac{1}{\alpha})^2k^{3/4})}\cdot n^{\Oh(1)}$.
	This completes the proof of the lemma.
\end{proof}

From \Cref{lem:algCorrect} and \Cref{lem:algTime}, we conclude the correctness of the following theorem.
\subexpconnectivity*

\section{Approximation Schemes for (Min) \probEdgelessnessDown} \label{sec:vc-eptas}
In this section, we design approximation schemes for (\textsc{Min}) \probEdgelessnessDown. Since many of the arguments for both problems are similar, we present the algorithm in a unified manner, and only highlight the differences. As in \Cref{sec:vc}, we adopt the abbreviation $\Pi_1$ (resp.\ $\Pi_2$) for \probEdgelessnessDown (resp.\ \probEdgelessnessMinDown). 

We assume that the given instance $\cI = (P, k, \alpha)$ of $\Pi_1$ (resp.\ $\cI = (P, k, \alpha, \mu)$ of $\Pi_2$) is a yes-instance. Then, let $(S^*, r^*)$ be a solution to $\cI$ (resp.\ an optimal solution for $\Pi_2$), where $S^* \subseteq P$ with $|S^*| \le k$, and $\alpha \le r^*(p) \le 1$ iff $p \in S^*$, and $r^*(p) = 1$ otherwise. Finally, for $\Pi_2$, let $\mathsf{OPT} = \cost(S^*, r^*) \sum_{p \in S^*} (1-r^*(p))$. 

Without loss of generality, let us assume that $P$ is completely contained in the first quadrant. Let $\ell = \lceil 16/\epsilon \rceil$. For all $0 \le i, j < 2\ell$, let $G(i, j)$ be a grid of sidelength $2\ell \times 2\ell$ with origin at $(-2i, -2j)$. By slight perturbation of the grid, we assume that no point of $P$ lies on the boundary of $G(i, j)$ for any $0 \le i, j < 2\ell$.

For an index $0 \le i < 2\ell$, let $S^*_i \subseteq S^*$ be the set of points with cartersian coordinates $(x, y)$ such that $x \mod 2\ell \in (i-1, i+1]$. 

Let $A = \{ 0 \le i < 2\ell : |S^*_i| \le k/\ell \}$. For $\Pi_1$, let $B = \{0, 1, \ldots, 2\ell\}$, whereas for $\Pi_2$, let $B = \{ 0 \le i < 2\ell : \sum_{p \in S^*_i} (1-r^*(p)) \le \mathsf{OPT}/\ell \}$. Since the sets $S^*_i$ form a partition of $S^*$, it follows from an averaging argument that $|A| > \ell$ and $|B| > \ell$. Therefore, $A \cap B \neq \emptyset$. Let $\imath \in A \cap B$.  

By repeating the preceding argument w.r.t.\ the indices $j$, it follows that for there exists an index $0 \le \jmath < 2\ell$ such that the number of points of $S^*$ such that the unit disks centered at the point intersects the boundary of the grid $G(\imath, \jmath)$ is at most $k - k (1-\frac{1}{\ell})^2 = k(\frac{2}{\ell} - \frac{1}{\ell^2}) \le 2k/\ell$ -- note that this argument holds for both $\Pi_1$ and for $\Pi_2$. For $\Pi_2$, we have the additional property that the total contribution of the shrinking factors corresponding to such points is at most $\mathsf{OPT} (\frac{2}{\ell} - \frac{1}{\ell^2}) \le 2\mathsf{OPT}/\ell$. Due to the definition of $\ell$, we have that $2/\ell \le \epsilon/8$.

The first step is to ``guess'' the pair $(\imath, \jmath)$, by running the following procedure on every pair $(i, j)$ with $0 \le i, j < 2\ell$, and returning the minimum cost solution. 

Consider a $2\ell \times 2\ell$ grid cell $C$ in a grid $G(i, j)$ for some pair $(i, j)$. Let $P' \subseteq P$ be the set of points $p$ such that the closed unit disk centered at $p$ is either completely contained in $C$, or \emph{intersects} the boundary of $p$. Let $C'$ denote the \emph{expanded version} of $C$, i.e., the Minkowski sum of $C$ with a unit disk centered at the origin. Note that $P' \subseteq C'$.

Now, we further subdivide $C'$ into $\Oh(\ell^2)$ sub-cells of unit sidelength. Note that the distance between any pair of points belonging to a sub-cell is at most $\sqrt{2} < 2$, thus, the set of points in a sub-cell form a clique in the \emph{original} unit disk graph. Observe that $S^*$ must contain at least $|K|-1$ points from a clique $K$ -- otherwise there exist points $p, q \in K$ such that the neither of the unit disks centered at $p$ and $q$ is shrunk. Thus, the edge $pq$ persists in the resulting disk graph $G(P, r^*)$, which is a contradiction.  Now, two cases depending on the number of points contained in a sub-cell. 

If the number of points in a sub-cell is less than $1 + \frac{4}{\epsilon}$, then we say that this is a \emph{small} sub-cell. For every small sub-cell $c$, we guess the intersection of $S^*$ with $c$. For a fixed guess, we mark the set of points in the $c$ that belong to the guessed intersection as \emph{shrinkable}, and mark the remaining points in $c$ as \emph{unshrinkable}. From the argument in the previous paragraph, note that the number of possible intersections for a particular sub-cell is at most $\Oh(1/\epsilon)$, which implies that the total number of guesses corresponding to the grid cell $C$ is at most $2^{\Oh(1/\epsilon)}$. 

Note that each of the remaining sub-cells contains at least $1 + 4/\epsilon$ points. We say that such a sub-cell is \emph{large}. In this case, we ``mark'' all the points in the corresponding sub-cell as \emph{shrinkable}. Note that from every clique $K$ corresponding to a large sub-cell, any solution must shrink either $|K|-1$ or $|K|$ disks. Thus, by marking all points as shrinkable, our solution may shrink at most $\frac{|K|}{|K|-1} \le 1+\epsilon/4$ factor additional points.

Thus, for each fixed guess corresponding to each small sub-cell, we partition the points in $P'$ into \emph{shrinkable} and \emph{unshrinkable}. For $\Pi_1$, the situation is much simpler -- we simply iterate over all guesses, and check whether shrinking all disks corresponding to points marked as shrinkable results in an edgeless graph. Formally, for a fixed guess, we consider the radius assignment $r: P' \to \{\alpha, 1\}$, where $r(p) = \alpha$ iff $p$ is marked as shrinkable, and check whether the graph $G(P', r)$ is edgeless. By iterating over all guesses, we store a solution that shrinks the minimum number of disks, if any.

Now, we turn to $\Pi_2$. Note that here we cannot simply shrink all disks marked as shrinkable to $\alpha$, since we want to minimize the cost of the solution. Thus, we formulate the problem as a linear program. For each point $p \in P'$, we add a variable $r(p)$. For all unshrinkable points $p$, we add the constraint $r(p) = 1$; whereas for all shrinkable points $q$, we add the constraint $\alpha \le r(q) \le 1$. Finally, for each pair of points $p, q \in P'$ with $|pq| \le 2$, we add the constraint $r(p) + r(q) \le |pq|$. It is easy to see that the LP is an exact formulation of the subproblem.  In particular, if the LP is feasible, then an optimal solution to the LP is a solution to the sub-problem corresponding to the grid-cell $C$, and it does not shrink any unit disk centered at a point that was marked unshrinkable. Note that this LP has $\Oh(1/\epsilon)$ variables and $\Oh(1/\epsilon^2)$ constraints, and can be solved using the algorithm of \Cref{prop:lpsolver} in time $(1/\epsilon)^{\Oh(1/\epsilon)} \cdot \frac{1}{\epsilon^2} = (1/\epsilon)^{\Oh(1/\epsilon)}$.
For this guess, let $t$ denote the number of disks that are shrunk, i.e., the number of points $p \in P'$ with $r(p) < 1$. Note that $0 \le t \le k$ -- otherwise, we can safely discard this guess, since the total number of disks that are shrunk cannot exceed $k$. Now, we iterate over each of the $2^{\Oh(1/\epsilon)}$ guesses, and maintain a set of solutions as follows. For every $0 \le t \le k$, we denote by $A[C, t]$, the minimum cost of a solution that shrinks at most $t$ disks centered at $P'$ in the subproblem corresponding to $C$. Note that all the values $A[C, t]$ can be computed in $(1/\epsilon)^{\Oh(1/\epsilon)} \cdot n^{\Oh(1)}$ time. 

Now we use dynamic programming to combine the solutions computed for the subproblems corresponding to every cell $C$. Let us order the cells as $C_1, C_2, \ldots$ in an arbitrary manner -- say we order the rows from bottom to top, and in each row, we list the cells from left to right. For $i \ge 1$, and $0 \le k' \le k$, let $B[i, k']$ denote the minimum cost of a solution for the subproblem for the cells $C_1 \cup C_2 \cup \ldots \cup C_i$ that uses at most $k'$ disks. It is easy to compute this value using dynamic programming, via the following recurrence.
$$B[m, k'] = \min_{0 \le t \le k'} B[m-1, k'-t] + A[C_m, t].$$
It is easy to modify this computation so as to also compute the corresponding solution, although we need to be slightly careful because points near the boundaries of the cells may appear in multiple (at most $4$) subproblems. However, this is easily handled as follows -- for every point $p$, we let the final radius $r(p)$ to be the minimum of its radii from at most $4$ solutions that are combined by the dynamic program. The overall smallest cost solution for the grid $G(i, j)$ will be found at the entry $B[z, k]$, where $z$ is the index of the last cell. Finally, we simply iterate over all $\ell^2$ pairs $(i, j)$, and return the smallest cost solution over all $G(i, j)$'s. Note that dynamic programming computation is polynomial, therefore, the algorithm takes $(1/\epsilon)^{\Oh(1/\epsilon)} \cdot n^{\Oh(1)}$ time. It remains to show the correctness and the approximation guarantee of the algorithm.

Fix a pair of indices $(\imath, \jmath)$ such that the number of unit disks centered at an optimal solution $S^*$ intersects the boundary of cells is at most $\epsilon k/8$, and the total shrinking factors of such disks is at most $\epsilon \cdot \mathsf{OPT}/8$. Let $C_1, C_2, \ldots, C_y$ denote the cells of the grid $G(\imath, \jmath)$, and for any $1 \le m \le t$, let $S^*_m \subseteq S^*$ denote the set of points of $S^*$ that are contained in the expanded version $C'_m$ of the cell $C_m$. Let $S^*_m = T^*_m \uplus L^*_m$, where $T^*_m$ and $S^*_m$ denote the subsets of points from $S^*_m$ that belong to the tiny and large sub-cells respectively. 

Out of $(1/\epsilon)^{\Oh(1/\epsilon^2)}$ guesses corresponding to $C_m$, let us fix the guess where the set of points from the tiny sub-cells of $C_m$ that are marked as shrinkable is exactly $T^*_m$. Note that for large sub-cells, all the points are marked as shrinkable, and the total number of shrinkable points from large sub-cells is at most $(1+\epsilon/4) \cdot |L^*_m|$. Therefore, the total number of shrinkable points corresponding to this guess is at most $|T^*_m| + (1+\epsilon/4) \cdot |L^*_m| \le (1+\epsilon) \cdot |S^*_m| =: t^*_m$. This implies that the solution $(S^*_m, r^*_m)$ is a valid solution to the subproblem corresponding to $A[C_m, t^*_m]$, where $r^*_m$ is a restriction of $r^*$ to the points in $S^*_m$.

Let $R^*_m \subseteq S^*_m$ be the set of points $p$ such that the unit disk centered at $p$ intersects the boundary of a cell $C_m$ in the grid $G(\imath, \jmath)$, and let $Q^*_m \coloneqq S^*_m \setminus R^*_m$. Let $R^* = \bigcup_{m = 1}^y R^*_m$, and let $Q^* \coloneqq S^* \setminus R^*$. Note that for any $p \in R^*$, the number of cells $C_m$ such that the unit disk centered at $p$ intersects the boundary of $C_m$ is at most $4$. Therefore, $p$ is counted in at most $4$ subproblems corresponding to different cells $C_m$, i.e., $|\left\{ m : p \in S^*_m \right\}| \le 4$. On the other hand, for any $p \in S^* \setminus R^*$, it occurs in exactly one $S^*_m$. Furthermore, we know that $\sum_{m = 1}^y |R^*_m| \le 4 \cdot |R^*| \le 4 \epsilon \cdot k/8 \le \epsilon \cdot k/2$, and $\sum_{p \in R^*} (1-r^*(p)) \le \epsilon \cdot \mathsf{OPT}/8$. Now, note that

\begin{align*}
	\sum_{m = 1}^y A[C_m, t_m] &\le \sum_{m = 1}^y \sum_{p \in S^*_m} (1-r^*_m(p)) \tag{Since $(S^*_m, r^*_m)$ is a valid solution for $A[C_m, t_m]$}
	\\&= \sum_{m = 1}^y \lr{ \sum_{p \in R^*_m} (1-r^*_m(p)) + \sum_{p \in S^*_m \setminus R^*_m} (1-r^*_m(p))}
	\\&\le \sum_{p \in Q^*} (1-r^*(p)) + 4 \sum_{p \in R^*} (1-r^*(p)) 
	\\&\le \mathsf{OPT} + 4 \epsilon/8 \cdot \mathsf{OPT}
	\\&= (1+\epsilon/2) \cdot \mathsf{OPT}
\end{align*}

For every $1 \le m \le y$ and $0 \le t_m \le (1+\epsilon)k$, let $S[C_m, t_m] \subseteq P \cap C'_m$ be the solution corresponding to the subproblem $A[C_m, t_m]$. Then, note that,
\begin{align*}
	\sum_{m = 1}^y |S[C_m, t^*_m]| &\le \sum_{m = 1}^y (1+\epsilon/4) \cdot |S^*_m| 
	\\&\le (1+\epsilon/4)\cdot\sum_{m = 1}^y |R^*_m| + |Q^*_m|
	\\&\le (1+\epsilon/4) \cdot \lr{|Q^*| + \epsilon \cdot k/2 }
	\\&\le (1+\epsilon) \cdot k \tag{Since $|Q^*| \le |S^*| \le k$}
\end{align*}
This concludes the analysis for $\Pi_2$. For $\Pi_1$, a similar (but much simpler) argument shows that the solution found by combining the optimal solution for each subproblem shrinks at most $(1+\epsilon)k$ disks. Thus, we obtain the following results.

\begin{theorem}
	There exists an algorithm that, given an instance $(P, k, \alpha)$ of \probEdgelessnessDown, and a fixed $\epsilon > 0$, runs in time $2^{\Oh(\frac{1}{\epsilon} \log(\frac{1}{\epsilon}))} \cdot n^{\Oh(1)}$, and either concludes that there exists no solution of size $k$; or returns a solution $(S, r)$, such that (i) $|S| \le (1+\epsilon)k $, (ii) $r(p) = \alpha$ for all $p \in S$, and $r(q) = 1$ for all $q \not\in S$, and (iii) $G(P, r)$ is edgeless.
\end{theorem}

\begin{theorem}
	There exists an algorithm that, given an instance $(P, k, \alpha, \mu)$ of \probEdgelessnessMinDown, and a fixed $\epsilon > 0$, runs in time $2^{\Oh(\frac{1}{\epsilon} \log(\frac{1}{\epsilon}))} \cdot n^{\Oh(1)}$, and either concludes that there exists no solution of size $k$; or returns a solution $(S, r)$, such that (i) $|S| \le (1+\epsilon)k $, (ii) $\sum_{p \in S} (1-r(p)) \le (1+\epsilon) \cdot {\sf OPT} \le (1+\epsilon) \mu$, (iii) $\alpha \le r(p) < 1$ for all $p \in S$; otherwise $r(q) = 1$ for all $q \in P \setminus S$, and (iv) $G(P, r)$ is edgeless.
\end{theorem}

\section{Lower Bounds}\label{sec:lower} 

\subsection{NP-hardness of  {\sc $k$-Shrinking to Connectivity}}\label{subsec:conn-NP-hard}  
In this subsection, we show that  \probConnectivityFracDown
is  \classNP-hard.

Our lower bound is obtained by a reduction from \textsc{Monotone Planar $3$-SAT}.  Given a Boolean formula $\varphi$ in the conjunctive normal form with $n$ variables $v_1,\ldots,v_n$ and $m$ clauses $C_1,\ldots,C_m$, we associate with $\varphi$ the bipartite graph $G_\varphi$ with the vertex set $\{v_1,\ldots,v_n\}\cup\{C_1,\ldots,C_m\}$ such that $v_i$ is adjacent to $C_j$  if and only if the clause $C_j$ contains either $v_i$ or its negation. 
By the well-known result of Lichtenstein~\cite{Lichtenstein82}, \textsc{$3$-SAT} is \classNP-complete for the instances $\varphi$ with planar $G_\varphi$. Moreover (see~\cite{Lichtenstein82}), the problem remains hard if (a) $\varphi$ is monotone in the sense that each clause of $\varphi$ contain either only positive occurrences  of variables or only negations, (b) the graph $G_\varphi'$ obtained from $G_\varphi$ by the additions of the edges $v_{i-1}v_i$ for $i\in\{1,\ldots,n\}$ (assuming that $v_0=v_n$) forming the cycle $C$ is planar, and (c) $G_\varphi'$ has a planar embedding such that each $C_j$ is in the inner face of $C$ if the clause $C_j$ contains positive literals and is in the outer face otherwise (see Figure~\ref{fig:sat} a). Throughout this subsection, whenever we refer to \textsc{Monotone Planar $3$-SAT}, we mean the restricted variant of \textsc{$3$-SAT} with the  instances satisfying (a)--(c).  

Furthermore, we use the observation of  Knuth and
Raghunathan~\cite{KnuthR92} that $G_\varphi$ can be drawn in a special way. Following~\cite{KnuthR92}, we give an informal description of such a drawing (see Figure~\ref{fig:sat} b). The variables $v_1,\ldots,v_n$ are represented by segments of a line $L$,  the clauses $C_1,\ldots,C_m$ are drawn as points above or below $L$ depending on whether $C_j$ is positive or negative, and the edges incident to each $C_j$ is drawn as rectilinear   ``three-legged'' structures.  These three-legged structures are properly nested so that none of the legs between clauses and variables cross each other. Such a drawing in a rectilinear grid can be constructed in polynomial time. We start our reductions from such drawings.

\begin{figure}[ht]
	\centering
	\scalebox{0.7}{
		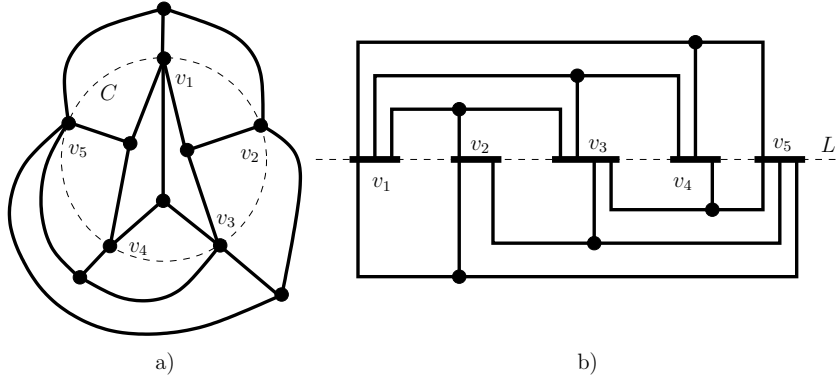
	}
	\caption{Graphs $G_\varphi$ and $G_\varphi'$ for $\varphi=(v_1\vee v_2\vee v_3)\wedge(v_1\vee v_3\vee v_4)\wedge(v_1\vee v_4\vee v_5)\wedge (\overline{v}_1\vee\overline{v}_2\vee \overline{v}_5)\wedge (\overline{v}_2\vee\overline{v}_3\vee \overline{v}_5)\wedge (\overline{v}_3\vee\overline{v}_4\vee \overline{v}_5))$ (a) and the drawing of $G_\varphi$ (b). }\label{fig:sat}
\end{figure}

It is convenient to prove hardness for the dual problem, namely, \probConnectivityFracUp. 
Observe that expanding $k$ disks of radius $1$ to radius $\beta$  is equivalent to shrinking  $|P|-k$ disks of radius $\beta$ to radius $1$. Thus, \probConnectivityFracDown and  \probConnectivityFracUp are indeed dual.

\probConnectivityFracUp is trivial if $\beta=1$. We show that the problem is \classNP-hard for every $\beta>1$.

\begin{theorem}\label{thm:hard-one}
	For every rational $\beta>1$, \probConnectivityFracUp  is  \classNP-hard. Furthermore, for any $\beta\geq 2$, the problem is \classNP-hard when restricted to the sets of points with integer coordinates. 
\end{theorem}

\begin{proof}
	We reduce from \textsc{Monotone Planar $3$-SAT} with instances satisfying (a)--(c). Consider such an instance $\varphi$ with variables $v_1,\ldots,v_n$ and clauses $C_1,\ldots,C_m$. Let $\ell$ be be the maximum number of occurrences of a variable in the clauses. Let also 
	$\delta>0$ be a rational constant and let
	$s=2\delta n(\ell+1)$, $q=2(\lceil \delta/2\rceil+1)$ and $d=q+6$. 
	For the proof for $\beta<2$, we take $\delta=\beta-1$ and for the case $\beta\geq 2$, we set $\delta=\lfloor \beta\rfloor-1$.

	\begin{sidewaysfigure}[p]
		\centering
		\scalebox{0.75}{
			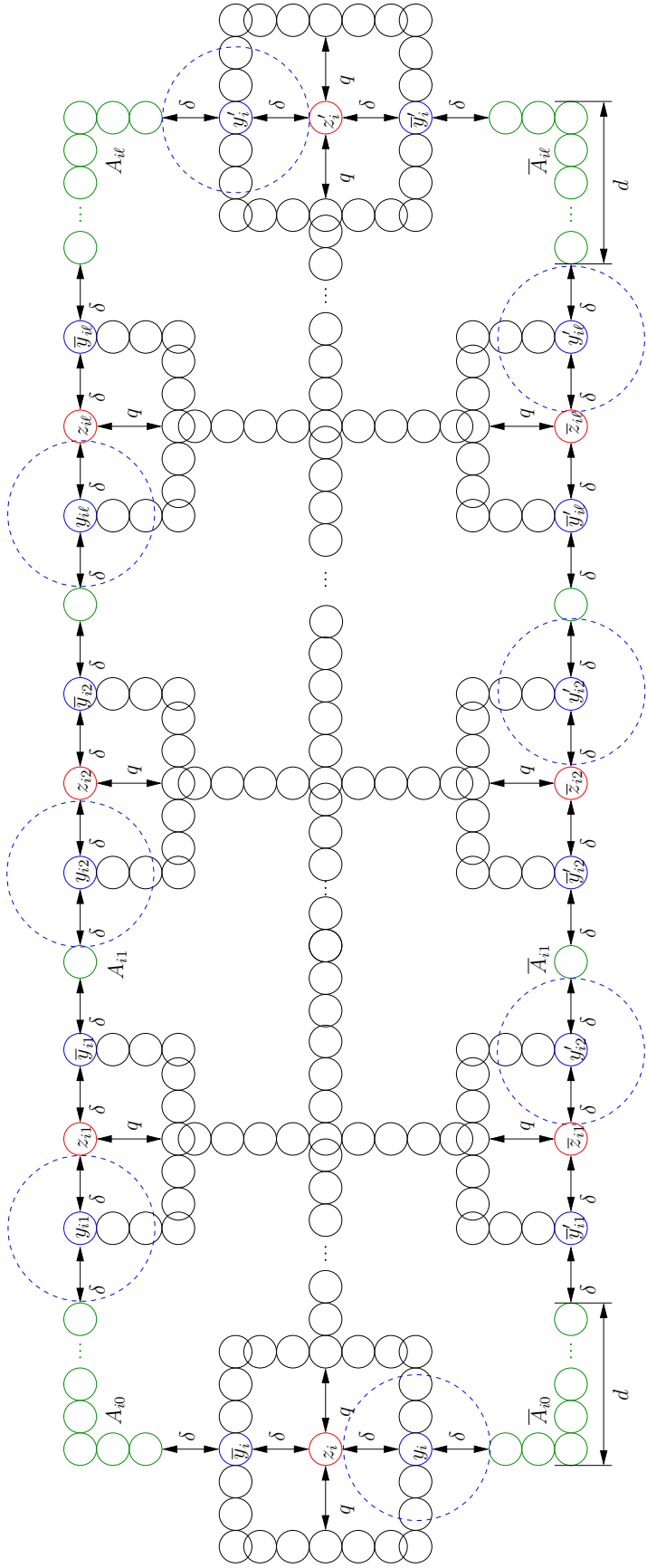
		}
		\caption{The variable gadget for $v_i$. The radius scaling corresponding to the positive assignment of $v_i$ is shown by dashed blue circles.}\label{fig:var-one}
	\end{sidewaysfigure}
	
	For each variable $v_i$, we construct the set of points as it is shown in Figure~\ref{fig:var-one}. Since we are constructing an instance of \probConnectivityFracUp, we draw disks of radius one having their centers in the points instead of points themselves to show the connectivity.  Denote $P_i$ the set of points in the gadget. The points $y_{ij},\overline{y}_{ij},y'_{ij},\overline{y}'_{ij}$ are used to connected the variable gadget with the gadgets that will be constructed for clauses. The crucial idea behind the construction of the gadget is that if we are forced to set $r(p)=\beta$ to at most $2\ell+2$ points $p$ in the gadget then (a)~we have to set $r(p)=\beta$ for exactly $2\ell+2$ points, 
	(b)~because of the points $z_i,z'_{i}$ and $z_{ij},z'_{ij}$ for $j\in\{1,\ell\}$, for each of these $z$-points, we have to assign $r(p)=\beta$ to one of closest the $y$-points $p$, and moreover (c)~the assignment for $y$-points has to be ``synchronized'' as it is shown in   Figure~\ref{fig:var-one}. Thus, we obtain exactly two possible choices for the assignment corresponding to the value of the variable $v_i$.

	Recall that $G_\varphi$ has a special planar embedding demonstrated in Figure~\ref{fig:sat}. We use this and construct the set of points $P$ as it is shown in Figure~\ref{fig:clause-one}:
	\begin{itemize}
		\item arrange variable gadgets for $v_1,\ldots,v_n$ along a line at distances two and construct $m-1$  ``connecting'' greenpoints $b_1,\ldots,b_{m-1}$ between them,
		\item for each $k\in\{1,\ldots,m\}$, construct a set of points $S_k$ (shown in magenta) corresponding to the clause $C_k$ in such a way that the unit disk graph with disks having their centers in $S_k$ form a  ``three-legged'' structure; the structures for distinct clauses should be non-intersecting. 
	\end{itemize}
	We set $N=|P|$ and define $k=2n(\ell+1)$. This completes the construction of the instance $(P,k,\beta)$ of \probConnectivityFracUp.
	Since $\delta$ is rational, we have that the points of $P$ have rational coordinates, and if $\delta$ is integer, then 
	the coordinates of the points can be chosen to be integer.
	By~\cite{KnuthR92}, it can be seen that the construction of $(P,k,\beta)$ from $\varphi$
	can be done in polynomial time.
	
	\begin{sidewaysfigure}[p]
		\centering
		\scalebox{0.75}{
			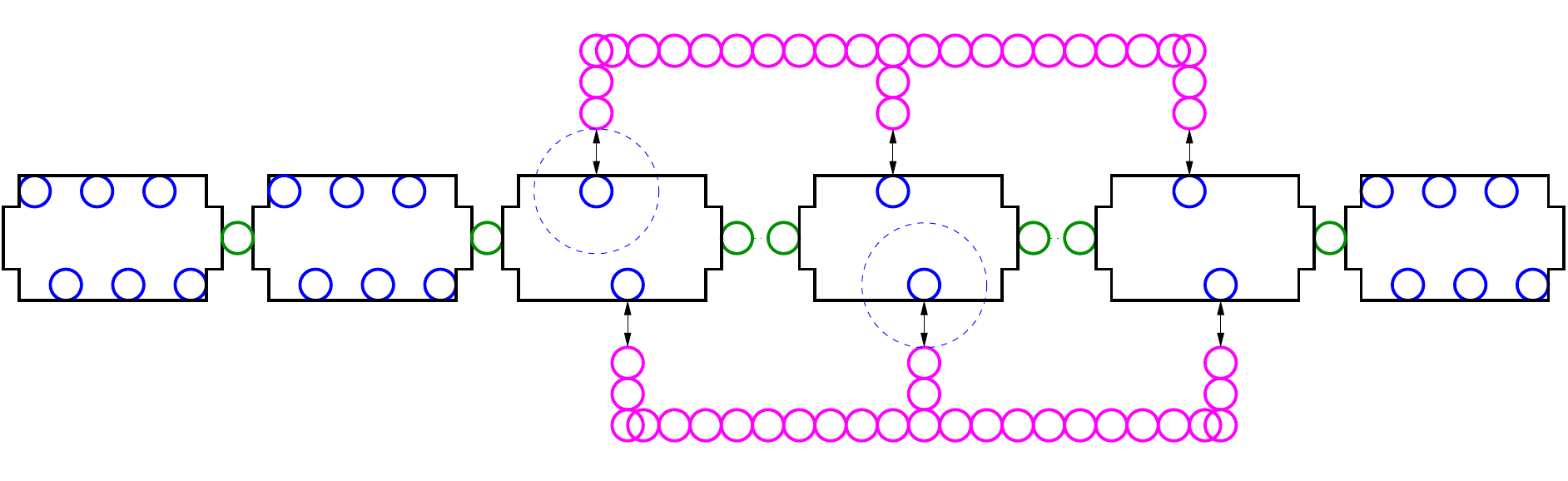
		}
		\caption{The connection of  the variable gadgets and the construction of the clause gadgets; here $C_r=v_h\vee v_i\vee v_j$ and 
			$C_r=\overline{v}_h\vee \overline{v}_i\vee \overline{v}_j$.
			The radius epansion corresponding to the positive assignment of $v_h$ and the negative assignment of $v_i$ is shown by dashed blue circles.}\label{fig:clause-one}
	\end{sidewaysfigure}
	
	We claim that the Boolean variables $v_1,\ldots,v_n$ have an assignment satisfying $\varphi$ if and only if  $(P,k,\beta)$ is a yes-instance of \probConnectivityFracUp.
	
	First, we assume that $\varphi$ is satisfied by an assignment of $v_1,\ldots,v_n$. For each $i\in\{1,\ldots,n\}$, we set 
	$S_i=\{y_i, y_{i'}\}\cup\{y_{ij},y_{ij}'\mid j\in\{1,\ldots,\ell\}\}$ if $v_i=true$. Symmetrically, we set 
	$S_i=\{\overline{y}_i, \overline{y}_i'\}\cup\{\overline{y}_{ij},\overline{y}_{ij}\mid j\in\{1,\ldots,\ell\}\}$ if $v_i=false$ (see Figure~\ref{fig:var-one}). Then we define $S=\bigcup_{i=1}^nS_i$ and 
	define the function $r\colon P\rightarrow \{1,\beta\}$ by setting 
	$r(p)=\beta$ for $p\in S$ and $r(p)=1$ for the other points $p$. Notice that $|S|=2n(\ell+1)=k$. 
	
	For each $i\in\{1,\ldots,n\}$, we have that the disk graph $G(P_i,r_i)$, where  $r_i$ is the restriction of $r$ on $P_i$,
	is connected (see Figure~\ref{fig:var-one}). Due to the disks of radius one with centers in the points in $b_1,\ldots,b_m$, the subgraphs $G(P_i,r_i)$ are connected with each other in $G(P,r)$.
	For each $k\in \{1,\ldots,m\}$, $G(S_k,r_k')$, where $r_k'$ is the restriction of $r$ on $S_k$,  is connected (see Figure~\ref{fig:clause-one}). Moreover, because $\varphi$ is satisfied, we have that if $C_k=v_h\vee v_i\vee v_j$, then one of the graphs  $G(P_h,r_h)$, $G(P_i,r_i)$, $G(P_j,r_j)$ has a disks intersecting a disk of  $G(S_k,r_k')$. For clauses with negative literals, the arguments are the same. This implies that  $G(P,r)$ is connected. 
	
	For the opposite direction, assume that $S\subseteq P$ is a set of point of size at most $k$ such that $G(P,r)$ is connected for 
	$r(p)=
	\begin{cases}
		\beta&\mbox{if } p\in S,\\
		1&\mbox{if } p\in P\setminus S.
	\end{cases}
	$
	
	For each $i\in\{1,\ldots,n\}$, let 
	$Z_i=\{z_i,z_i',z_{i1},\ldots,z_{i\ell},\overline{z}_{i1},\ldots,\overline{z}_{i\ell}\}$ (see Figure~\ref{fig:var-one}; these points are shown in red), and let $Z=\bigcup_{i=1}^nZ_i$. 
	Because the disk graph $G(P,r)$ is connected, for every $z\in Z$, there is a point $p\in P$ distinct from $z$ such that $r(z)+r(p)\leq 1+\beta$, that is, either $z\in S$ or $p\in S$.
	We say that $p$ is \emph{assigned} to $z$ if $p\in S$ and $r(p)=\beta$. By the choice of $\delta$ and $q$, no point $p$ can be assigned to two distinct $z,z'\in Z$. Because $|Z_i|=2(\ell+1)$ for every $i\in\{1,\ldots,n\}$, we have that $|Z|=2n(\ell+1)=k$. We obtain that (a) for every $z\in Z$, at most one point $p_z\in P$ is assigned to $z$ and (b) for every $p\in P\setminus Z$ that is not assigned to any point of $Z$, $r(p)=1$.  Also by the choice of $\delta$ and $q$, we  have that for  every $i\in\{1,\ldots,n\}$, $p_{z_i}\in\{y_i,\overline{y}_i\}$, $p_{z_i'}\in\{y_i',\overline{y}_i'\}$, and $p_{z_{ij}}\in\{y_{ij},\overline{y}_{ij}\}$, $p_{z_{ij}'}\in\{y_{ij}',\overline{y}_{ij}'\}$ for $j\in\{1,\ldots,\ell\}$ (see Figure~\ref{fig:var-one}).

	Consider $i\in\{1,\ldots,n\}$. Suppose that either no point is assigned to $z_i$ or $p_{z_i}=y_i$. Then consider the points from the set $A_{i0}$ (again, we refer to Figure~\ref{fig:var-one}). Since $r(p)=1$ for $p\in A_{0i}$, we have that $p_{z_{i1}}=y_{i1}$ to ensure the connectivity. By the same arguments for $A_{i1}$, $p_{z_{i2}}=y_{i2}$. Iterating, we obtain that $r(y_i)=r(y_i')=r(y_{i1})=\ldots=r(y_{i\ell})=r(y_{i1}')=\ldots=r(y_{i\ell}')=\beta$ and 
	$r(\overline{y}_i)=r(\overline{y}_i')=r(\overline{y}_{i1})=\ldots=r(\overline{y}_{i\ell})=r(\overline{y}_{i1}')=\ldots=r(\overline{y}_{i\ell}')=1$. 
	If this holds, we set the value of the Boolean variable $v_i=true$. Symmetrically, if $p_{z_i}=\overline{y}_i$, then 
	$r(y_i)=r(y_i')=r(y_{i1})=\ldots=r(y_{i\ell})=r(y_{i1}')=\ldots=r(y_{i\ell}')=1$ and 
	$r(\overline{y}_i)=r(\overline{y}_i')=r(\overline{y}_{i1})=\ldots=r(\overline{y}_{i\ell})=r(\overline{y}_{i1}')=\ldots=r(\overline{y}_{i\ell}')=\beta$, and we set $v_i=false$ in this case. 
	We claim that this is a satisfying assignment for $\varphi$.   
	
	To see the claim, consider $C_k$ for some $k\in\{1,\ldots,m\}$. By symmetry, we can assume that $C_k=v_h\vee v_i\vee v_j$, that is, the clause contains only positive literals. Observe that the set of points $S_k$ corresponding to $C_k$ (see Figure~\ref{fig:clause-one}) contains exactly tree points that are at distance $2\alpha+\delta$ from $y_{hh'}$, $y_{ii'}$, and $y_{jj'}$, respectively, for some $h',i',j'\in\{1,\ldots,\ell\}$, and other points of $S_k$ are at distances that are bigger than $\beta+1$ from the points of the variable gadgets. Then because $G(P,r)$ is connected,
	$r(y_{hh'})=\beta$ or $r(y_{ii'})=\beta$ or $r(y_{jj'})=\beta$. This means that at least one of the variables $x_h$, $x_i$, $x_j$ is set $true$. This proves that $C_k$ is satisfied. 
	As $k\in\{1,\ldots,m\}$ is arbitrary, we obtain that $\varphi$ is satisfied. This completes the proof.
\end{proof}

%
%
%
%

By the duality between \probConnectivityFracUp and \probConnectivityFracDown, we obtain the following corollary.

\begin{corollary}\label{cor:hard-two}
	For every rational positive $\alpha<1$, \probConnectivityFracDown  is  \classNP-hard. 
\end{corollary}

\subsection{W[1]-Hardness of a generalization of {\sc $k$-Expanding to Connectivity}} \label{subsec:w1-hard-conn}
In this section, we consider a generalization of \probConnectivityFracDown the set of disks that may be expanded must come from a designated subset. A formal definition follows. 

\begin{tcolorbox}[colback=white!5!white,colframe=gray!75!black]
	\probConnGen
	\begin{description}
		\item[Input:]  A finite set of points $P \subset \real^2$, a subset $A \subseteq P$, a constant $\alpha \ge 1$, and a parameter $k \ge 0$.
		\item[Task:] Decide whether there exists a solution $(S, r)$ such that
		\begin{itemize}[leftmargin=*]
			\item $S \subseteq A$ is a set of size \textbf{at most} $k$ corresponding to expanded disks, which must come from the set of \emph{allowed} set $A$,
			\item \vspace{-0.25cm}corresponding radii $r: P \to \realplus$ such that $r(p) = \begin{cases}
				1 &\text{ if } p \not\in S
				\\\alpha &\text{ if } p \in S
			\end{cases}$
			\item \vspace{-0.4cm}$G(P, r)$ is connected.
		\end{itemize}
	\end{description}
\end{tcolorbox}

For convenience, we consider an formulation of \probConnGen, wherein we are given a collection of unit disks ${\cal Q}$ with the corresponding set of centers being $P$, along with the \emph{allowed subset} $\mathcal{A} \subseteq \mathcal{Q}$, $k$, and $\alpha \ge 1$. The initial UDG defined by $G({\cal Q})$ may not be connected. The question is whether there exists a subset $\mathcal{S} \subseteq \mathcal{A}$ of size at most $k$, such that after expanding the disks in $\mathcal{S}$ to radius $\alpha$ (and leaving all other disks unit), the resulting intersection graph is connected. Clearly, this is a reformulation of the original problem. We refer to this formulation as {\sc $k$-Expanding UDG to Connectivity}.

The reduction is shown from the {\sc Covering Points by Unit Disks} problem, where we are given a set $P$ of points in the Euclidean plane, a set of unit disks $\cal D$ and a parameter $k\in\mathbb{N}$, and the objective is to decide whether there exists ${\cal D}'\subseteq{\cal D}$ of size at most $k$ such that every point in $P$ belongs to at least one unit disk in ${\cal D}'$.\footnote{The formulation of the problem where $\cal D$ is not part of the input, but one can choose any $k$ unit disks on the plane, is not harder, since it is not difficult to see that there there exist only polynomially many centers that are sufficient to consider in order to find a solution, if one exists.} This problem is known to be W[1]-hard when parameterized by $k$~\cite{DBLP:conf/esa/Marx05,giannopoulos2008parameterized}.

Having defined the source problem for our reduction, we state the result of this section.

\begin{theorem} \label{thm:genw1hard}
	\probConnGen is W[1]-hard parameterized by $k$.
\end{theorem}

\begin{proof}
	We give a polynomial-time parameter-preserving reduction from the {\sc Covering Points by Unit Disks} problem to {\sc $k$-Expanding UDG to Connectivity}. 
	To this end, consider some instance $(C,{\cal D},k)$ of {\sc Covering Points by Unit Disks}. Let $A$ denote the set of centers of the disks in $\cD$. Let $\epsilon>0$ be small enough such that the following properties hold.
	\begin{enumerate}
		\item\label{prop1} Let ${\cal Q}_C=\{Q_p: p\in C\}$ where $Q_p$ is a disk with center $p$ and radius $\epsilon$. Then, for all distinct $p,p'\in C$, $Q_p$ and $Q_{p'}$ do not intersect.
		
		\item\label{prop2} For every unit disk $D\in{\cal D}$ and every point $p\in C$, $Q_p$ and $D$ intersect if and only if $p$ belongs to $D$. 
		
		\item\label{prop3} Let ${\cal Q}_{\cal D}=\{Q_D: D\in {\cal D}\}$ where $Q_D$ is a disk with the same center as $D$ and radius $\epsilon$. Then, for every point $p\in C$, if the center of $Q_D$ is not $p$, then the distance between $Q_D$ and $p$ is at least $10\epsilon$.
	\end{enumerate}
	
	\begin{figure}
		\begin{center}
			\fbox{\includegraphics[scale=1.1]{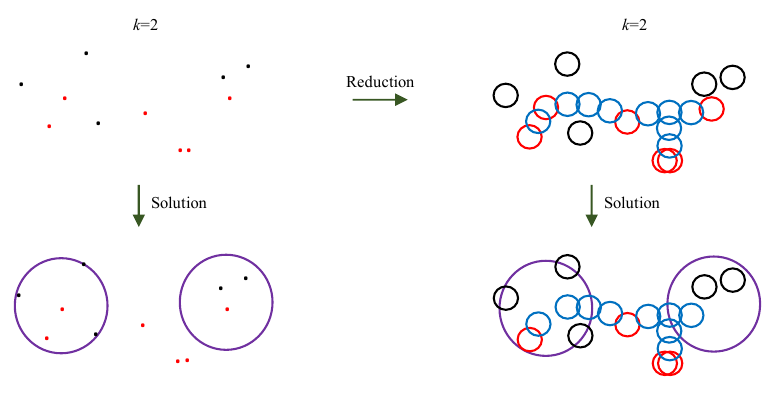}}
			\caption{Reduction of a Yes-instance of {\sc Covering Points by Unit Disks} to a Yes-instance of {\sc Scaling to Connected Graph}. On the left side, the points are colored black, and the possible poitions for centers are colored red. On the right, only the new-unit disks colored red can be enlarged.}
			\label{fig:WHard}
		\end{center}
	\end{figure}
	
	We will output an instance $({\cal Q},\mathcal{A}, k,\alpha = 1/\epsilon)$ of {\sc $k$-Expanding UDG to Connectivity} where we initially define ${\cal Q}$ as disks of radius $\epsilon$ (referred to next as new-unit for the sake of clarity, to avoid confusion with the different unit used with respect to $(C,{\cal D},k)$). Further, each disk in $\mathcal{A}$ can be expanded from $\epsilon$ to $1$; whereas all other disks must remain of radius $\epsilon$. By ultimately scaling the entire instance by a factor of $1/\epsilon$, we obtain a valid of input instance for {\sc $k$-Expanding UDG to Connectivity}. In the following we construct the instance before scaling.
	
	The construction of $\cal Q$ is done as follows. First, we insert all the new-unit disks in ${\cal Q}_C$ into $\cal Q$, and these disks are \emph{not} included in $\mathcal{A}$. Next we add all the new-unit disks in ${\cal Q}_{\cal D}$ into $\cal Q$, as well as in $\cal A$. Now, we add to $\cal Q$ (but not to $\cal A$) a set ${\cal Q}^{\mathsf{connect}}$ of polynomially many additional new-unit disks with scaling parameter $0$ so that {\em (i)} they do not intersect any new-unit disk $Q_p \in {\cal Q}_C$ centered at $p$ but they do intersect the (hypothetical) unit disk centered at $p$, and {\em (ii)}  the intersection graph of them together with all the new-unit disk already inserted into $\cal Q$ except for those in ${\cal Q}_C$ is connected (see Fig.~\ref{fig:WHard}), which is possible due to Property \ref{prop3} above. This completes the construction.

	As the reduction can be executed in polynomial time, and the parameter is unchanged, it remains to prove that $(C,{\cal D},k)$ is a yes-instance of {\sc Covering Points by Unit Disks} if and only if $({\cal Q},\mathcal{A}, k,\alpha)$ is a yes-instance of {\sc $k$-Expanding UDG to Connectivity}. In one direction, suppose that $(C,{\cal D},k)$ is a yes-instance of {\sc Covering Points by Unit Disks}, and let ${\cal D}'$ be a set of at most $k$ unit disks that covers all the points of $C$. Then, we scale by a factor of $1/\epsilon$ every new-unit disk in $Q_D\in{\cal Q}_D = {\cal A} \subseteq {\cal Q}$ such that $D\in{\cal D}'$. So, because ${\cal D}'$ is a solution, we have thus scaled at most $k$ new-unit disks, and every new-unit disk in ${\cal Q}_C $ is intersected by at least one of the new-unit disks that we scaled, so they ``join'' the connected component of the intersection graph formed by $({\cal Q}_{\cal D}\setminus{\cal R})\cup {\cal Q}^{\mathsf{connect}}$ (note that here we implicitly use that for every $R\in{\cal R}$, ${\cal Q}^{\mathsf{connect}}$ does intersect the (hypothetical) unit disk with the same center as $R$, else if such new-unit disks are enlarged, they could have possibly form their own connected components). Hence, we derive a connected intersection graph. Hence, we have a yes-instance of {\sc $k$-Expanding UDG to Connectivity}.
	
	In the other direction, suppose that $({\cal Q}, {\cal A}, k, \alpha)$ is a yes-instance of {\sc $k$-Expanding UDG to Connectivity}. Let ${\cal S}$ be the set of new-unit disks scaled by a solution for this instance. Then, by the choice of the scaling parameters, ${\cal S}\subseteq {\cal A} = {\cal Q}_{\cal D}$, and for each $Q_D\in{\cal S}$, the scaling yields a disk contained in $\cal D$. So, due to Properties \ref{prop1} and \ref{prop2} above and as after scaling we obtain a connected intersection graph, each point $p\in P$ must be covered by at least one of the new-unit disks in ${\cal S}$ after scaling. As we scale at most $k$ new-unit disks, this implies that $\{D\in{\cal D}: Q_D\in {\cal S}\}$ is a solution to $(P,{\cal D},k)$, and hence it is a yes-instance of {\sc Covering Points by Unit Disks}.
\end{proof}

\subsection{NP-Hardness of {\sc (Min)} \probEdgelessnessDown} \label{subsec:np-hard-vc}

In this section we show that (\textsc{Min}) \probEdgelessnessDown is \classNP-complete for any fixed rational $0 < \alpha < 1$. We fix a value of $0 < \alpha < 1$ that is rational. We give a reduction from \textsc{Independent Set} in cubic planar graphs, i.e., planar graphs of degree exactly $3$, which is known to be \classNPH \ (\cite{Mohar01}). First, we make use of the following crucial observation. Let $G$ be a graph, and $t_e$ be positive integers given for every edge $e \in E(G)$. Suppose $G'$ is a new graph obtained replacing every edge $e =  uv \in E(G)$ with an induced path $(u, w_1, w_2, \ldots, w_{2s_e}, v)$ in the new graph $G'$. We refer to this operation as ``subdividing an edge $e \in E(G)$ by $2s_e$ times''. It is easy to see that $G$ has an independent set of size $k$ iff $G'$ has an independent set of size $k + \sum_{e \in E(G)} s_e$. 

A \emph{rectilinear embedding} of a planar graph is a planar embedding where the vertices are mapped to points with integer coordinates, and each edge is mapped to a broken line, consisting of an alternate sequence of horizontal and vertical segments. We have the following result due to Liu, Morgana, and Simeone \cite{LiuMS98}. 

\begin{proposition}[\cite{LiuMS98}] \label{prop:planar-embedding}
	Every $n$-vertex planar graph of maximum degree at most $4$ admits a rectilinear embedding with at most $3$ bends for every edge with the area $\Oh(n^2)$. Furthermore, such an embedding can be constructed in $\Oh(n)$ time.
\end{proposition}
Let $r = \lceil \frac{2}{1-\alpha} \rceil$, and note that $r \ge 2$ is an integer. Given an instance $(G, k)$ of \textsc{Independent Set}, where $G$ is a cubic planar graph, we use \Cref{prop:planar-embedding} to obtain a rectilinear embedding. Next, we scale the embedding by an integer factor $\gamma \ge 1$, such that (1) the rectilinear distance between any two parallel segments is at least $2r+1$, (2) the length of every segment (horizontal or vertical) in any broken line corresponding to an edge is a positive integer divisible by $(2r+2)(2r-1)$. 

Consider an edge $e = uv \in E(G)$. Suppose the total length of the broken line (i.e., sum of lengths of all horizontal and vertical segments comprising the broken line) corresponding to $e$ in the rectilinear embedding is equal to $(2r-1) \cdot t_e$, for some positive integer $t_e$. We want to add an odd number of points along the embedding of $e$ separated by a specific distance. We consider two cases depending on whether $t_e$ is even or odd. If $t_e$ is odd with $t_e = 2s_e + 1$ for some positive integer $s_e$, then, we place a series of points $w_1, w_2, \ldots, w_{2s_e}$ along the broken line corresponding to $e$, each separated by a distance of exactly $2r-1$ along the broken edge (see \Cref{fig:nph-vc}). 

Otherwise, suppose that $t_e = 2s_e$ for some positive integer $s_e \ge r+1$ (this follows since the length of each segment is divisible by $(2r+2)(2r-1)$. In this case, we ``squeeze'' one additional point along the first edge so that the total number of points added along the embedding of $e$ is odd. To achieve this, we proceed as follows. Suppose we start from $u$, and walk along the edge. We place the first point $w_1$ at distance $2r+1$ from $u$ along the edge, and $w_2$ is placed at distance $2r+1$ after $w_1$. Next, we place $w_3, w_4, \ldots, w_{2r+1}$, each separated by distance $2r-2 \ge 2$ after $w_2$ along the edge. Note that the distance between $u$ and $w_{2r+3}$ is $2(2r-1) + (2r-2)(2r-1) = 2r(2r-1)$, thus the length of the remaining portion of the edge $e$ is an even multiple of $2r-1$. Thus, we can place an even number of points, each separated by distance $5$ along the edge. Thus, we place points $w_{2r+2}, w_{2r+3}, \ldots, w_{2s_e + 1}$ in the remaining portion of $e$, separated by distance $2r-1$. 

\begin{figure}[hbt]
	\centering
	\includegraphics[scale=1]{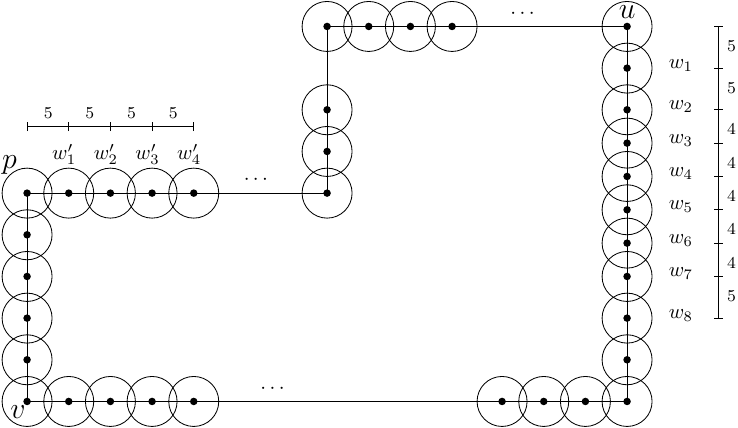}
	\caption{Example for construction of the instance of \probEdgelessnessDown. Here, the length of the broken line corresponding edge $pu$ is an odd number, therefore, all the points $w'_1, w'_2, \ldots$ are placed at distance $5$ along the drawing of the edge. On the other hand, suppose the length of the broken line corresponding to the edge $uv$ is an even number. Nevertheless, we need to place an odd number of points along the broken line. Therefore, we ``squeeze'' points $w_2, w_3, w_4, w_5, w_5, w_7$ in an interval of length $20$ by placing them at distance $4$. The rest of the points are placed at an even spacing of $5$ along the broken line. } \label{fig:nph-vc}
\end{figure}

After performing this operation for all edges $e \in V(G)$, the resulting embedding can be seen as a rectilinear embedding of a graph $G'$ that is obtained by subdividing each edge of $G$ exactly $2s_e$ times. Furthermore,  $p, q \in V(G')$ are adjacent in $G'$ iff the $\ell_1$ and $\ell_2$ distance between their images is either $2r-1$ or $2r-2$ (note that they have either the same $x$- or $y$-coordinate). On the other hand, $p, q \in V(G')$ are non-adjacent in $G'$ iff the $\ell_1$ and $\ell_2$ distance between their images is strictly greater than $2r-1$. From the previous observation, it follows that the original $(G, k)$ is a yes-instance of \textsc{Independent Set} iff $(G', k')$ is a yes-instance of \textsc{Independent Set}, where $k' \coloneqq k + \sum_{e \in E(G)} s_e$. 

Let $P$ be the image of $V(G')$ in the resulting embedding of $G'$ constructed as above. For each $p \in P$, we add a disk of radius $r$ centered at $p$. From the properties of the embedding of $G'$, it follows that the disks centered at $p$ and $q$ intersect iff there is an edge between the corresponding vertices in $G'$. Note that $G'$ is exactly the intersection graphs of disks of radius $r$ centered at points of $P$. By scaling down the entire instance by a factor of $r$, we obtain an instance with disks of unit radii. In the reduced instance, the distance between any neighbors is either $\frac{2r-1}{r}$ or $\frac{2r-2}{r}$. Thus, an edge can be removed by shrinking at least one of the disks to radius $\alpha \le \frac{2r-2}{2}$. 

Thus, we obtain the instance $\cI = (P, |P|-k', \alpha)$ \probEdgelessnessDown.

Now we argue that $\cI$ is a yes-instance of \probEdgelessnessDown iff $(G', k')$ is a yes-instance of \textsc{Independent Set}, which is equivalent to $(G, k)$. In the following, we slightly abuse the notation, and equate a vertex of $V(G')$ (or a subset of $V(G')$) with the corresponding image under the embedding of $G'$. Thus, $P$ is synonymous with $V(G')$ under this convention.

In the backward direction, let $U \subseteq V(G')$ be an independent set of size at least $k'$ in $G$. Let $S = V(G) \setminus V(G')$ that has size at most $|V(G')| - k' = |P| - k'$. Define $r(p) = \alpha$ for all points $p \in S$, and $r(q) = 1$ for all points $q \in U$. We claim that $(S, r)$ is a solution for $\cI$, by showing that the intersection graph $G(P, r)$ is edgeless. Suppose for contradiction that there exists an edge $pq \in E(G(P, r))$. Now, if either of $p$ or $q$ belongs to $S$, then $r(p) + r(q) < 2\alpha \le |pq|$, which implies that the edge $pq$ cannot exist in $G(P, r)$ (recall that we assume that the disks are open). Therefore, both $p$ and $q$ belong to $U$. However, this contradicts that $U$ is an independent set in $G'$. The proof of the forward direction is almost identical, and is therefore omitted.

We note that the same construction can be modified the \classNP-hardness of \probEdgelessnessMinDown. Indeed, let $\cI' = (P, |P|-k', \alpha, \mu)$ be the instance of \probEdgelessnessMinDown, where $P$ and $\alpha$ are as defined above, and $\mu = (1-\alpha) \cdot (|P|-k')$. The proof of the backward direction is same as above, where one can observe that the cost of the constructed solution $(S, r)$, is equal to $\sum_{p \in P} (1-r(p)) = (1-\alpha) \cdot |S| \le \mu$. 

For the forward direction, suppose $\cI'$ is a yes-instance of \probEdgelessnessMinDown. Then, it admits a solution $(S, r)$ such that $|S'| \le |P| - k'$, and $\alpha \le r(p) < 1$ iff $p \in S$. Let $U = P \setminus S$ be the set of size at least $k'$. We claim that $U$ is an independent set in $V(G')$. Indeed, suppose for some edge $pq \in E(G')$, both endpoints $p$ and $q$ belong to $U$. However, since $r(p') < 1$ only for points $p' \in S$, it follows that $r(p) = r(q) = 1$, which implies that $r(p) + r(q) = 2 > \frac{2r-1}{r} \ge |pq|$. However, this contradicts the assumption that the resulting intersection graph $G(P, r)$ is edgeless. 

\begin{theorem} \label{thm:np-hardness-vc}
	\probEdgelessnessDown and \probEdgelessnessMinDown are \classNPC\ for any fixed rational $0 < \alpha < 1$.
\end{theorem}

\subsection{NP-Hardness of {\sc (Min)} \probAcyclicityDown} \label{subsec:np-hard-fvs}

Note that \textsc{Feedback Vertex Set} remains \classNPH\ on instances $(G, k)$, where $G$ is a planar graph of maximum degree $4$ \cite{Speckenmeyer83,Speckenmeyer88}. Given an instance $(G, k)$ of \textsc{Feedback Vertex Set}, where $G$ is a planar graph of maximum degree $4$, we use \Cref{prop:planar-embedding} to obtain a rectilinear embedding. Next, we scale the embedding by a factor of $\gamma \ge 1$, such that (1) the rectilinear distance between any two parallel segments is at least $20$, (2) the length of every segment (horizontal or vertical) in any broken line corresponding to an edge is a positive integer divisible by $221 = 13 \cdot 17$. Note that the scaling factor $\gamma$ required to achieve that these properties are satisfied is polynomial in $n$.

Consider an edge $e = uv \in E(G)$, such that the length of the broken line segment corresponding to $e = uv$ is $2\cdot (17) + s_e \cdot (13)$, where $s_e$ is a positive integer. Then, we place a series of points $w_1, w_2, \ldots, w_{s_e}$ along the broken line corresponding to $e$ as follows. We start ``walking'' along $uv$, starting from $u$. We place the first point $w_1$ at a distance of $17$ along the rectilinear embedding of the edge. We place the subsequent points $w_2, w_3, \ldots, w_{s_e}$ along the embedding of $e$ such that the rectilinear distance between the consecutive pair of points is exactly $13$. Note that the rectilinear distance between $w_{s_e}$ and $v$ is exactly $17$.

After performing this operation for all edges $e \in V(G)$, the resulting embedding can be seen as a rectilinear embedding of a graph $G'$ that is obtained by subdividing each edge of $G$ exactly $s_e$ times. Furthermore, $p, q \in V(G')$ are adjacent in $G'$ iff the $\ell_1$ and $\ell_2$ distance between their images is at most $17$. On the other hand, $p, q \in V(G')$ are non-adjacent in $G'$ iff the $\ell_1$ and $\ell_2$ distance between their images is at least  $\sqrt{13^2 + 17^2} > 18$. We will refer to the vertices that originated from the embedding of $V(G)$ as \emph{original} vertices, and the vertices that are obtained as a result of subdividing the edges of $E(G)$ as \emph{chain} vertices. 

It is straightforward to see that $(G, k)$ is a yes-instance of \textsc{Feedback Vertex Set} iff $(G', k)$ is a yes-instance of \textsc{Feedback Vertex Set}, since subdividing an edge cannot increase the optimal size of a feedback vertex set. This also implies that if $(G', k)$ is a yes-instance of \textsc{Feedback Vertex Set}, then without loss of generality, it admits a solution $S \subseteq V(G')$ that does not contain a vertex that is obtained by subdividing an original edge from $E(G)$. In the following, we slightly abuse the notation, and equate a vertex of $V(G')$ (or a subset of $V(G')$) with the corresponding image under the embedding of $G'$. Thus, $P$ is synonymous with $V(G')$ under this convention.

Let $P$ be the image of $V(G')$ in the resulting embedding of $G'$ constructed as above. For each $p \in P$, we add a disk of radius $9$ centered at $p$. From the properties of the embedding of $G'$, it follows that the disks centered at $p$ and $q$ intersect iff there is an edge between the corresponding vertices in $G'$. Furthermore, the distance between an original vertex and an adjacent chain vertex is exactly $17$, i.e., the ``overlap'' between two disks centered at the respective disks is of ``width'' $1$. On the other hand, the overlap between two disks centered at adjacent chain vertices is $3$. Now we scale the entire instance by a factor of $9$ such that each disk becomes a unit disk. Also, if we set $\alpha = 8/9$, then only the disks corresponding to an original vertex, or a chain vertex adjacent to them can be shrunk. Thus, we define an instance $\cI= (P, k, \alpha)$ of \probAcyclicityDown by setting $\alpha = 8/9$. We claim that $\cI$ is a yes-instance of \probAcyclicityDown iff $(G', k)$ is a yes-instance of \textsc{Feedback Vertex Set}.

In the forward direction, we first claim that that $\cI$ admits a solution $(S, r)$ such that $S$ is a subset of original vertices. Suppose for contradiction that this does not hold, and $p \in S$ is a chain vertex. If $p$ is not adjacent to an original vertex, then note that even if we shrink the disk centered at $p$ to radius $8$, both the edges incident to $p$ are still present in $G(P, r)$, since the distance between $p$ and its adjacent chain vertices is equal to $13/9 < 8/9 + 8/9$. In this case, we simply remove $p$ from $S$, and obtain a solution of smaller size. On the other hand, if a chain vertex $p$ is adjacent to an original vertex $v$, then we obtain a solution by removing $p$ and adding $v$ (and also exchanging their $x$-values). Note that since $p$ has degree exactly $2$ in $G'$, it follows that this results in a valid solution for $\cI$. Thus, after performing all such operations, we can assume that $S$ is a subset of original vertices. Let $(S, r)$ be a solution with this property. Since we shrink all vertices of $S$ to radius $\alpha = 8/9$, and the distance between an original vertex and a chain vertex is $17/9$, it follows that all vertices of $S$ in the resulting intersection graph $G(P, r)$ are isolated. Furthermore, the graph $G(P, r) \setminus S$ is acyclic. However, $G(P, r) \setminus S$ is isomorphic to the graph $G' \setminus S$, which implies that $S$ is a feedback vertex set of size at most $k$ for $G'$.

In the backward direction, if $S \subseteq V(G')$ is a feedback vertex set for $G'$ of size at most $k$, then note that we may assume that $S$ only consists of original vertices. Furthermore, $G$ is acyclic. Then using an argument similar to the previous paragraph, if we set $r(p) = 8/9$ for all points $p \in S$, and $r(q) = 1$ for all $q \in P \setminus S$, then it follows that the resulting intersection graph $G(P, r)$ is acyclic, and the vertices of $S$ are isolated in $G(P, r)$. It follows that $(S, r)$ is a solution for $\cI$. 

Note that the same construction can be extended to show the \classNP-Hardness of \probAcyclicityMinDown, by constructing the instance $(P, k, \alpha = 8/9, \mu)$ as long as the budget $\mu$ is set to at least $8k/9$. 

Finally, we also observe that the same construction also shows that \probAcyclicity is \classNPH\ for the following reason. Recall that given an instance $(P, \alpha, \beta, \mu)$ of \probAcyclicity, the goal is to find a solution $(S, r)$ such that cost $\sum_{p \in P} (1-r(p))$ is at most $\mu$ (note that there is no cardinality constraint on $S$). We obtain the same instance $\cI' = (P, \alpha = 8/9, \mu)$, where $\mu = 8k/9$. Consider an optimal solution $(S, r)$ of cost at most $\mu$. As argued previously, if $S$ contains a chain vertex $p$ such that both of the neighbors of $p$ are also chain vertices, then we can obtain a solution of smaller cost by simply removing $p$ from $S$, and setting $r(p) = 1$; which contradicts the optimality of the solution $S$. On the other hand, if a chain vertex $p$ is adjacent to an original vertex $v$, then we consider two cases. If the edge $pv$ exists in $G(P, r)$, then we can remove $p$ from $S$, set $r(p) = 1$, and obtain a solution of the smaller cost -- note that the resulting intersection graph remains unchanged. Finally, if the edge $pv$ does not exist in $G(P, r)$, then it follows that $16/9 < r(p) + r(v) < |pv| = 17/9$. In this case, we simply set $r'(p) = 1/9$, and $r'(v) = 17-r(p)$. Since $8/9 \le r'(v) \le 1$, it follows that the resulting solution has the same cost as the original solution. Therefore, we can assume that if $\cI'$ is a yes-instance, then $S$ is a subset of original vertices. Next, we argue that in any optimal solution $(S, r)$ such that $S$ is a subset of original vertices, $x_p = 8$ for all $p \in S$. Indeed, if some original vertex $p \in S$ has $8/9 < r(p) < 1$, then all the edges incident to $p$ in the graph $G'$ are also present in $E(G(P, r))$. Therefore, we can obtain a solution of smaller cost by removing $p$ from $S$. We conclude the proof by observing that since $1- r(p) = 1$ for all $p \not\in S$, $(P, \alpha, \mu = k)$ admits a solution of cost at most $k$ iff $|S| \le k$. Now, the equivalence follows from the argument made for \probAcyclicityMinDown. 

Finally, it is not difficult to see that, by appropriately choosing the constants in the construction (as done in the previous section), the hardness can be generalized to any fixed rational $0 < \alpha < 1$. We omit the details and conclude with the following theorem.

\begin{theorem} \label{thm:np-hardness-fvs}
	\probAcyclicity, \probAcyclicityDown, and \probAcyclicityMinDown are \classNPC, for any fixed rational $0 < \alpha < 1$.
\end{theorem}

\section{Conclusion and Future Research} \label{sec:conclusion}

In this paper, we initiated the study of graph modification by \emph{scaling of disks}  from the perspective of Parameterized Complexity. Specifically, we focused on modifying the connectivity properties of the graph while shrinking or expanding a subset of disks. The (\textsc{Min}) \probEdgelessnessDown/\textsc{Acyclicity} problems correspond to shrinking at most $k$ disks to achieve a completely disconnected (i.e., edgeless), and minimally connected (i.e., acyclic) intersection graph, respectively. Whereas \probConnectivityFracDown(resp. \probConnectivityFracUp) corresponds to shrinking at least (resp. at most) $k$ disks while retaining (resp.~achieving) connectivity in the intersection graph.

Besides our NP/\classW{1}-hardness results for the above problems, we also gave several algorithmic results. For (\textsc{Min}) \probEdgelessnessDown, we designed (partial) polynomial kernels, FPT algorithms, and (bicriteria) approximation schemes. For \probEdgelessnessMinDown, we use a novel combination of linear programming with ideas from parameterized complexity and geometric approximation respectively. The FPT algorithm is obtained by first using kernelization to derive a bounded-size instance, and then using LP to solve this instance; whereas the approximation algorithm uses a combination of the classical shifting technique along with LP. We also obtained similar results for (\textsc{Min})\probAcyclicityDown; however our arguments and techniques are more sophisticated due to ``non-local'' nature of the problem. Finally, for \probConnectivityFracDown, we introduced a novel decomposition theorem that returns a grid minor of the disk graph along with a particular kind of embedding on the plane, and used this theorem to design a subexponential FPT algorithm for the problem. Similar ideas also gave subexponential FPT algorithms for \probEdgelessnessDown/\textsc{Acyclicity}.

Analogous to the results obtained for \probEdgelessnessDown, it is natural to ask whether \probConnectivityFracDown admits (partial) polynomial kernel, and approximation schemes. Furthermore, it is also possible to define cost-optimization version of \probConnectivityFracDown and study it from the perspective of fixed-parameter and approximation algorithms. 

At a broader level, we believe that a significant portion of our contribution lies in the conceptual realm, with the goal of pioneering a new domain within the field of geometric graph modification problems. Furthermore, it opens up a multitude of potential directions for future research. In this context, we propose a few specific avenues to explore.
\begin{itemize}
	\setlength{\itemsep}{-1pt}
	\item {\bf Modification to other graph Classes:} 
	In addition to graph modification to achieve edgelessness or connectivity, it is intriguing to consider modifying the given geometric graph to belong to other families of graphs, such as acyclic graphs (for minimal connectivity assurance), cluster graphs (for clustering purposes), cliques (for direct communication facilitation), and bounded-degree graphs (which generalize edgeless graphs).

	\item {\bf Various geometric intersection graphs as input:} Our study primarily focused on the simplest geometric intersection graphs, specifically disk graphs, where graph modification problems often become tractable within the realm of Parameterized Complexity despite being \classNP-hard. Beyond disk graphs, various other geometric intersection graphs, such as string graphs and map graphs, present substantial differences. Additionally, we can extend our investigations to three-dimensional spaces (e.g., unit ball graphs and more general ball graphs). Therefore, a natural research direction is to explore geometric modifications for graph classes within these broader and distinct geometric intersection graph categories.
	
	\item {\bf New geometric modification operations:} 
	While we utilized scaling as a quintessential example of a geometric modification operation, other geometric operations, like movement (recently studied in \cite{FominG0Z23, FominG0Z22}) or rotation, may be more suitable depending on the context and the nature of the geometric objects under consideration. Furthermore, one might contemplate simultaneously shifting, scaling, and rotating objects (or applying any combination of two out of these three operations). This parallels research on graph modification problems where significant efforts have been made to concurrently address edge deletions and insertions. Beyond these combinations, defining and researching other geometric operations can be highly relevant based on the objectives and the class of geometric intersection graphs at hand. This includes operations like smoothing and dimension reduction.

	\item {\bf Meta-theorems.} While the study of research questions on a problem-to-problem basis can be sometimes challenging, yet interesting in its own right, the ultimate goal is often to prove meta-theorems that simultaneously assert or refute, say, membership in \classFPT\ or the existence of a polynomial kernel, for a host of problems. For example, we may claim, for a whole class of problems, that if a specific constraint X in the definition of any problem in the class can be expressed in some logic (e.g., Monadic Second Order Logic), then the problem is FPT, or that if and only if the input graph (or one of the input graphs) has bounded treewidth, then the problem is FPT, or that if the problem admits an OR-cross-composition from an \classNP-hard problem, then it has no polynomial kernel unless the polynomial hierarchy collapses. We refer to the books \cite{downey2013fundamentals,DBLP:books/sp/CyganFKLMPPS15,fomin2019kernelization} for examples of such results. 
	Clearly, after gaining a better understanding of the problem class introduced in this paper, the pursuit of such meta-theorems represents an intriguing avenue for further research.	
\end{itemize}

\bibliographystyle{siam}
\bibliography{Scaling,references}

\end{document}

%% file: SAT.pdf_t
\begin{picture}(0,0)%
\includegraphics{SAT.pdf}%
\end{picture}%
\setlength{\unitlength}{3947sp}%
\begingroup\makeatletter\ifx\SetFigFont\undefined%
\gdef\SetFigFont#1#2#3#4#5{%
  \reset@font\fontsize{#1}{#2pt}%
  \fontfamily{#3}\fontseries{#4}\fontshape{#5}%
  \selectfont}%
\fi\endgroup%
\begin{picture}(7419,3365)(394,-2909)
\put(5476,-2836){\makebox(0,0)[lb]{\smash{{\SetFigFont{12}{14.4}{\rmdefault}{\mddefault}{\updefault}{\color[rgb]{0,0,0}b)}%
}}}}
\put(1900,-284){\makebox(0,0)[lb]{\smash{{\SetFigFont{12}{14.4}{\rmdefault}{\mddefault}{\updefault}{\color[rgb]{0,0,0}$v_1$}%
}}}}
\put(2480,-958){\makebox(0,0)[lb]{\smash{{\SetFigFont{12}{14.4}{\rmdefault}{\mddefault}{\updefault}{\color[rgb]{0,0,0}$v_2$}%
}}}}
\put(2267,-1538){\makebox(0,0)[lb]{\smash{{\SetFigFont{12}{14.4}{\rmdefault}{\mddefault}{\updefault}{\color[rgb]{0,0,0}$v_3$}%
}}}}
\put(1487,-1791){\makebox(0,0)[lb]{\smash{{\SetFigFont{12}{14.4}{\rmdefault}{\mddefault}{\updefault}{\color[rgb]{0,0,0}$v_4$}%
}}}}
\put(953,-911){\makebox(0,0)[lb]{\smash{{\SetFigFont{12}{14.4}{\rmdefault}{\mddefault}{\updefault}{\color[rgb]{0,0,0}$v_5$}%
}}}}
\put(1233,-431){\makebox(0,0)[lb]{\smash{{\SetFigFont{12}{14.4}{\rmdefault}{\mddefault}{\updefault}{\color[rgb]{0,0,0}$C$}%
}}}}
\put(7639,-891){\makebox(0,0)[lb]{\smash{{\SetFigFont{12}{14.4}{\rmdefault}{\mddefault}{\updefault}{\color[rgb]{0,0,0}$L$}%
}}}}
\put(3653,-1210){\makebox(0,0)[lb]{\smash{{\SetFigFont{12}{14.4}{\rmdefault}{\mddefault}{\updefault}{\color[rgb]{0,0,0}$v_1$}%
}}}}
\put(4527,-872){\makebox(0,0)[lb]{\smash{{\SetFigFont{12}{14.4}{\rmdefault}{\mddefault}{\updefault}{\color[rgb]{0,0,0}$v_2$}%
}}}}
\put(5574,-864){\makebox(0,0)[lb]{\smash{{\SetFigFont{12}{14.4}{\rmdefault}{\mddefault}{\updefault}{\color[rgb]{0,0,0}$v_3$}%
}}}}
\put(6327,-1205){\makebox(0,0)[lb]{\smash{{\SetFigFont{12}{14.4}{\rmdefault}{\mddefault}{\updefault}{\color[rgb]{0,0,0}$v_4$}%
}}}}
\put(7213,-858){\makebox(0,0)[lb]{\smash{{\SetFigFont{12}{14.4}{\rmdefault}{\mddefault}{\updefault}{\color[rgb]{0,0,0}$v_5$}%
}}}}
\put(1726,-2836){\makebox(0,0)[lb]{\smash{{\SetFigFont{12}{14.4}{\rmdefault}{\mddefault}{\updefault}{\color[rgb]{0,0,0}a)}%
}}}}
\end{picture}%

%% file: Variable-Up.pdf_t
\begin{picture}(0,0)%
\includegraphics{Variable-Up.pdf}%
\end{picture}%
\setlength{\unitlength}{3947sp}%
\begingroup\makeatletter\ifx\SetFigFont\undefined%
\gdef\SetFigFont#1#2#3#4#5{%
  \reset@font\fontsize{#1}{#2pt}%
  \fontfamily{#3}\fontseries{#4}\fontshape{#5}%
  \selectfont}%
\fi\endgroup%
\begin{picture}(14412,5866)(-2107,-3748)
\put(9138,-3130){\makebox(0,0)[lb]{\smash{{\SetFigFont{12}{14.4}{\rmdefault}{\mddefault}{\updefault}{\color[rgb]{0,0,0}$y_{i\ell}'$}%
}}}}
\put(1216,-3286){\makebox(0,0)[lb]{\smash{{\SetFigFont{12}{14.4}{\rmdefault}{\mddefault}{\updefault}{\color[rgb]{0,0,0}$\delta$}%
}}}}
\put(2041,-3286){\makebox(0,0)[lb]{\smash{{\SetFigFont{12}{14.4}{\rmdefault}{\mddefault}{\updefault}{\color[rgb]{0,0,0}$\delta$}%
}}}}
\put(1216,1214){\makebox(0,0)[lb]{\smash{{\SetFigFont{12}{14.4}{\rmdefault}{\mddefault}{\updefault}{\color[rgb]{0,0,0}$\delta$}%
}}}}
\put(2041,1214){\makebox(0,0)[lb]{\smash{{\SetFigFont{12}{14.4}{\rmdefault}{\mddefault}{\updefault}{\color[rgb]{0,0,0}$\delta$}%
}}}}
\put(1891,-2686){\makebox(0,0)[lb]{\smash{{\SetFigFont{12}{14.4}{\rmdefault}{\mddefault}{\updefault}{\color[rgb]{0,0,0}$q$}%
}}}}
\put(1891,914){\makebox(0,0)[lb]{\smash{{\SetFigFont{12}{14.4}{\rmdefault}{\mddefault}{\updefault}{\color[rgb]{0,0,0}$q$}%
}}}}
\put(-974,-436){\makebox(0,0)[lb]{\smash{{\SetFigFont{12}{14.4}{\rmdefault}{\mddefault}{\updefault}{\color[rgb]{0,0,0}$\delta$}%
}}}}
\put(-974,-1261){\makebox(0,0)[lb]{\smash{{\SetFigFont{12}{14.4}{\rmdefault}{\mddefault}{\updefault}{\color[rgb]{0,0,0}$\delta$}%
}}}}
\put(-749,-1036){\makebox(0,0)[lb]{\smash{{\SetFigFont{12}{14.4}{\rmdefault}{\mddefault}{\updefault}{\color[rgb]{0,0,0}$q$}%
}}}}
\put(-1649,-1036){\makebox(0,0)[lb]{\smash{{\SetFigFont{12}{14.4}{\rmdefault}{\mddefault}{\updefault}{\color[rgb]{0,0,0}$q$}%
}}}}
\put(376,1214){\makebox(0,0)[lb]{\smash{{\SetFigFont{12}{14.4}{\rmdefault}{\mddefault}{\updefault}{\color[rgb]{0,0,0}$\delta$}%
}}}}
\put(376,-3286){\makebox(0,0)[lb]{\smash{{\SetFigFont{12}{14.4}{\rmdefault}{\mddefault}{\updefault}{\color[rgb]{0,0,0}$\delta$}%
}}}}
\put(2851,1214){\makebox(0,0)[lb]{\smash{{\SetFigFont{12}{14.4}{\rmdefault}{\mddefault}{\updefault}{\color[rgb]{0,0,0}$\delta$}%
}}}}
\put(2851,-3286){\makebox(0,0)[lb]{\smash{{\SetFigFont{12}{14.4}{\rmdefault}{\mddefault}{\updefault}{\color[rgb]{0,0,0}$\delta$}%
}}}}
\put(3676,1214){\makebox(0,0)[lb]{\smash{{\SetFigFont{12}{14.4}{\rmdefault}{\mddefault}{\updefault}{\color[rgb]{0,0,0}$\delta$}%
}}}}
\put(6151,1214){\makebox(0,0)[lb]{\smash{{\SetFigFont{12}{14.4}{\rmdefault}{\mddefault}{\updefault}{\color[rgb]{0,0,0}$\delta$}%
}}}}
\put(6151,-3286){\makebox(0,0)[lb]{\smash{{\SetFigFont{12}{14.4}{\rmdefault}{\mddefault}{\updefault}{\color[rgb]{0,0,0}$\delta$}%
}}}}
\put(4501,-3286){\makebox(0,0)[lb]{\smash{{\SetFigFont{12}{14.4}{\rmdefault}{\mddefault}{\updefault}{\color[rgb]{0,0,0}$\delta$}%
}}}}
\put(5326,-3286){\makebox(0,0)[lb]{\smash{{\SetFigFont{12}{14.4}{\rmdefault}{\mddefault}{\updefault}{\color[rgb]{0,0,0}$\delta$}%
}}}}
\put(4501,1214){\makebox(0,0)[lb]{\smash{{\SetFigFont{12}{14.4}{\rmdefault}{\mddefault}{\updefault}{\color[rgb]{0,0,0}$\delta$}%
}}}}
\put(5326,1214){\makebox(0,0)[lb]{\smash{{\SetFigFont{12}{14.4}{\rmdefault}{\mddefault}{\updefault}{\color[rgb]{0,0,0}$\delta$}%
}}}}
\put(5176,-2686){\makebox(0,0)[lb]{\smash{{\SetFigFont{12}{14.4}{\rmdefault}{\mddefault}{\updefault}{\color[rgb]{0,0,0}$q$}%
}}}}
\put(5176,914){\makebox(0,0)[lb]{\smash{{\SetFigFont{12}{14.4}{\rmdefault}{\mddefault}{\updefault}{\color[rgb]{0,0,0}$q$}%
}}}}
\put(3676,-3286){\makebox(0,0)[lb]{\smash{{\SetFigFont{12}{14.4}{\rmdefault}{\mddefault}{\updefault}{\color[rgb]{0,0,0}$\delta$}%
}}}}
\put(11322,388){\makebox(0,0)[lb]{\smash{{\SetFigFont{12}{14.4}{\rmdefault}{\mddefault}{\updefault}{\color[rgb]{0,0,0}$\delta$}%
}}}}
\put(11322,-2087){\makebox(0,0)[lb]{\smash{{\SetFigFont{12}{14.4}{\rmdefault}{\mddefault}{\updefault}{\color[rgb]{0,0,0}$\delta$}%
}}}}
\put(11322,-437){\makebox(0,0)[lb]{\smash{{\SetFigFont{12}{14.4}{\rmdefault}{\mddefault}{\updefault}{\color[rgb]{0,0,0}$\delta$}%
}}}}
\put(11322,-1262){\makebox(0,0)[lb]{\smash{{\SetFigFont{12}{14.4}{\rmdefault}{\mddefault}{\updefault}{\color[rgb]{0,0,0}$\delta$}%
}}}}
\put(6972,1213){\makebox(0,0)[lb]{\smash{{\SetFigFont{12}{14.4}{\rmdefault}{\mddefault}{\updefault}{\color[rgb]{0,0,0}$\delta$}%
}}}}
\put(7797,1213){\makebox(0,0)[lb]{\smash{{\SetFigFont{12}{14.4}{\rmdefault}{\mddefault}{\updefault}{\color[rgb]{0,0,0}$\delta$}%
}}}}
\put(8622,1213){\makebox(0,0)[lb]{\smash{{\SetFigFont{12}{14.4}{\rmdefault}{\mddefault}{\updefault}{\color[rgb]{0,0,0}$\delta$}%
}}}}
\put(9447,1213){\makebox(0,0)[lb]{\smash{{\SetFigFont{12}{14.4}{\rmdefault}{\mddefault}{\updefault}{\color[rgb]{0,0,0}$\delta$}%
}}}}
\put(9447,-3287){\makebox(0,0)[lb]{\smash{{\SetFigFont{12}{14.4}{\rmdefault}{\mddefault}{\updefault}{\color[rgb]{0,0,0}$\delta$}%
}}}}
\put(8622,-3287){\makebox(0,0)[lb]{\smash{{\SetFigFont{12}{14.4}{\rmdefault}{\mddefault}{\updefault}{\color[rgb]{0,0,0}$\delta$}%
}}}}
\put(7797,-3287){\makebox(0,0)[lb]{\smash{{\SetFigFont{12}{14.4}{\rmdefault}{\mddefault}{\updefault}{\color[rgb]{0,0,0}$\delta$}%
}}}}
\put(6972,-3287){\makebox(0,0)[lb]{\smash{{\SetFigFont{12}{14.4}{\rmdefault}{\mddefault}{\updefault}{\color[rgb]{0,0,0}$\delta$}%
}}}}
\put(11547,-1037){\makebox(0,0)[lb]{\smash{{\SetFigFont{12}{14.4}{\rmdefault}{\mddefault}{\updefault}{\color[rgb]{0,0,0}$q$}%
}}}}
\put(10647,-1037){\makebox(0,0)[lb]{\smash{{\SetFigFont{12}{14.4}{\rmdefault}{\mddefault}{\updefault}{\color[rgb]{0,0,0}$q$}%
}}}}
\put(8472,-2687){\makebox(0,0)[lb]{\smash{{\SetFigFont{12}{14.4}{\rmdefault}{\mddefault}{\updefault}{\color[rgb]{0,0,0}$q$}%
}}}}
\put(8472,913){\makebox(0,0)[lb]{\smash{{\SetFigFont{12}{14.4}{\rmdefault}{\mddefault}{\updefault}{\color[rgb]{0,0,0}$q$}%
}}}}
\put(896,1391){\makebox(0,0)[lb]{\smash{{\SetFigFont{12}{14.4}{\rmdefault}{\mddefault}{\updefault}{\color[rgb]{0,0,0}$y_{i1}$}%
}}}}
\put(1725,1369){\makebox(0,0)[lb]{\smash{{\SetFigFont{12}{14.4}{\rmdefault}{\mddefault}{\updefault}{\color[rgb]{0,0,0}$z_{i1}$}%
}}}}
\put(890,-3134){\makebox(0,0)[lb]{\smash{{\SetFigFont{12}{14.4}{\rmdefault}{\mddefault}{\updefault}{\color[rgb]{0,0,0}$\overline{y}_{i1}'$}%
}}}}
\put(-974,389){\makebox(0,0)[lb]{\smash{{\SetFigFont{12}{14.4}{\rmdefault}{\mddefault}{\updefault}{\color[rgb]{0,0,0}$\delta$}%
}}}}
\put(-974,-2086){\makebox(0,0)[lb]{\smash{{\SetFigFont{12}{14.4}{\rmdefault}{\mddefault}{\updefault}{\color[rgb]{0,0,0}$\delta$}%
}}}}
\put(-824,1064){\makebox(0,0)[lb]{\smash{{\SetFigFont{12}{14.4}{\rmdefault}{\mddefault}{\updefault}{\color[rgb]{0,0,0}$A_{i0}$}%
}}}}
\put(-824,-2836){\makebox(0,0)[lb]{\smash{{\SetFigFont{12}{14.4}{\rmdefault}{\mddefault}{\updefault}{\color[rgb]{0,0,0}$\overline{A}_{i0}$}%
}}}}
\put(-1151,-885){\makebox(0,0)[lb]{\smash{{\SetFigFont{12}{14.4}{\rmdefault}{\mddefault}{\updefault}{\color[rgb]{0,0,0}$z_i$}%
}}}}
\put(-1151,-1708){\makebox(0,0)[lb]{\smash{{\SetFigFont{12}{14.4}{\rmdefault}{\mddefault}{\updefault}{\color[rgb]{0,0,0}$y_i$}%
}}}}
\put(-1141,-66){\makebox(0,0)[lb]{\smash{{\SetFigFont{12}{14.4}{\rmdefault}{\mddefault}{\updefault}{\color[rgb]{0,0,0}$\overline{y}_i$}%
}}}}
\put(-599,-3586){\makebox(0,0)[lb]{\smash{{\SetFigFont{12}{14.4}{\rmdefault}{\mddefault}{\updefault}{\color[rgb]{0,0,0}$d$}%
}}}}
\put(2547,1368){\makebox(0,0)[lb]{\smash{{\SetFigFont{12}{14.4}{\rmdefault}{\mddefault}{\updefault}{\color[rgb]{0,0,0}$\overline{y}_{i1}$}%
}}}}
\put(2535,-3146){\makebox(0,0)[lb]{\smash{{\SetFigFont{12}{14.4}{\rmdefault}{\mddefault}{\updefault}{\color[rgb]{0,0,0}$y_{i2}'$}%
}}}}
\put(1727,-3146){\makebox(0,0)[lb]{\smash{{\SetFigFont{12}{14.4}{\rmdefault}{\mddefault}{\updefault}{\color[rgb]{0,0,0}$\overline{z}_{i1}$}%
}}}}
\put(3291,1044){\makebox(0,0)[lb]{\smash{{\SetFigFont{12}{14.4}{\rmdefault}{\mddefault}{\updefault}{\color[rgb]{0,0,0}$A_{i1}$}%
}}}}
\put(3341,-2831){\makebox(0,0)[lb]{\smash{{\SetFigFont{12}{14.4}{\rmdefault}{\mddefault}{\updefault}{\color[rgb]{0,0,0}$\overline{A}_{i1}$}%
}}}}
\put(5799,1379){\makebox(0,0)[lb]{\smash{{\SetFigFont{12}{14.4}{\rmdefault}{\mddefault}{\updefault}{\color[rgb]{0,0,0}$\overline{y}_{i2}$}%
}}}}
\put(5836,-3134){\makebox(0,0)[lb]{\smash{{\SetFigFont{12}{14.4}{\rmdefault}{\mddefault}{\updefault}{\color[rgb]{0,0,0}$y_{i2}'$}%
}}}}
\put(4181,1391){\makebox(0,0)[lb]{\smash{{\SetFigFont{12}{14.4}{\rmdefault}{\mddefault}{\updefault}{\color[rgb]{0,0,0}$y_{i2}$}%
}}}}
\put(5010,1369){\makebox(0,0)[lb]{\smash{{\SetFigFont{12}{14.4}{\rmdefault}{\mddefault}{\updefault}{\color[rgb]{0,0,0}$z_{i2}$}%
}}}}
\put(4175,-3134){\makebox(0,0)[lb]{\smash{{\SetFigFont{12}{14.4}{\rmdefault}{\mddefault}{\updefault}{\color[rgb]{0,0,0}$\overline{y}_{i2}'$}%
}}}}
\put(4996,-3134){\makebox(0,0)[lb]{\smash{{\SetFigFont{12}{14.4}{\rmdefault}{\mddefault}{\updefault}{\color[rgb]{0,0,0}$\overline{z}_{i2}$}%
}}}}
\put(10572,-3587){\makebox(0,0)[lb]{\smash{{\SetFigFont{12}{14.4}{\rmdefault}{\mddefault}{\updefault}{\color[rgb]{0,0,0}$d$}%
}}}}
\put(11159,-56){\makebox(0,0)[lb]{\smash{{\SetFigFont{12}{14.4}{\rmdefault}{\mddefault}{\updefault}{\color[rgb]{0,0,0}$y_i'$}%
}}}}
\put(11154,-1707){\makebox(0,0)[lb]{\smash{{\SetFigFont{12}{14.4}{\rmdefault}{\mddefault}{\updefault}{\color[rgb]{0,0,0}$\overline{y}_i'$}%
}}}}
\put(7458,1388){\makebox(0,0)[lb]{\smash{{\SetFigFont{12}{14.4}{\rmdefault}{\mddefault}{\updefault}{\color[rgb]{0,0,0}$y_{i\ell}$}%
}}}}
\put(9126,1378){\makebox(0,0)[lb]{\smash{{\SetFigFont{12}{14.4}{\rmdefault}{\mddefault}{\updefault}{\color[rgb]{0,0,0}$\overline{y}_{i\ell}$}%
}}}}
\put(10722,-2837){\makebox(0,0)[lb]{\smash{{\SetFigFont{12}{14.4}{\rmdefault}{\mddefault}{\updefault}{\color[rgb]{0,0,0}$\overline{A}_{i\ell}$}%
}}}}
\put(8325,-3135){\makebox(0,0)[lb]{\smash{{\SetFigFont{12}{14.4}{\rmdefault}{\mddefault}{\updefault}{\color[rgb]{0,0,0}$\overline{z}_{i\ell}$}%
}}}}
\put(8302,1368){\makebox(0,0)[lb]{\smash{{\SetFigFont{12}{14.4}{\rmdefault}{\mddefault}{\updefault}{\color[rgb]{0,0,0}$z_{i\ell}$}%
}}}}
\put(11142,-886){\makebox(0,0)[lb]{\smash{{\SetFigFont{12}{14.4}{\rmdefault}{\mddefault}{\updefault}{\color[rgb]{0,0,0}$z_i'$}%
}}}}
\put(7472,-3136){\makebox(0,0)[lb]{\smash{{\SetFigFont{12}{14.4}{\rmdefault}{\mddefault}{\updefault}{\color[rgb]{0,0,0}$\overline{y}_{i\ell}'$}%
}}}}
\put(10722,1063){\makebox(0,0)[lb]{\smash{{\SetFigFont{12}{14.4}{\rmdefault}{\mddefault}{\updefault}{\color[rgb]{0,0,0}$A_{i\ell}$}%
}}}}
\end{picture}%

%% file: Clause-Up.pdf_t
\begin{picture}(0,0)%
\includegraphics{Clause-Up.pdf}%
\end{picture}%
\setlength{\unitlength}{3947sp}%
\begingroup\makeatletter\ifx\SetFigFont\undefined%
\gdef\SetFigFont#1#2#3#4#5{%
  \reset@font\fontsize{#1}{#2pt}%
  \fontfamily{#3}\fontseries{#4}\fontshape{#5}%
  \selectfont}%
\fi\endgroup%
\begin{picture}(15066,4611)(2218,-2689)
\put(14746,-876){\makebox(0,0)[lb]{\smash{{\SetFigFont{12}{14.4}{\rmdefault}{\mddefault}{\updefault}{\color[rgb]{0,0,0}$b_{m-1}$}%
}}}}
\put(7801,-436){\makebox(0,0)[lb]{\smash{{\SetFigFont{12}{14.4}{\rmdefault}{\mddefault}{\updefault}{\color[rgb]{0,0,0}$v_h$}%
}}}}
\put(8117,-854){\makebox(0,0)[lb]{\smash{{\SetFigFont{12}{14.4}{\rmdefault}{\mddefault}{\updefault}{\color[rgb]{0,0,0}$\overline{y}_{hh''}$}%
}}}}
\put(7824, 39){\makebox(0,0)[lb]{\smash{{\SetFigFont{12}{14.4}{\rmdefault}{\mddefault}{\updefault}{\color[rgb]{0,0,0}$y_{hh'}$}%
}}}}
\put(10674, 39){\makebox(0,0)[lb]{\smash{{\SetFigFont{12}{14.4}{\rmdefault}{\mddefault}{\updefault}{\color[rgb]{0,0,0}$y_{ii'}$}%
}}}}
\put(10992,-861){\makebox(0,0)[lb]{\smash{{\SetFigFont{12}{14.4}{\rmdefault}{\mddefault}{\updefault}{\color[rgb]{0,0,0}$\overline{y}_{ii''}$}%
}}}}
\put(10655,-436){\makebox(0,0)[lb]{\smash{{\SetFigFont{12}{14.4}{\rmdefault}{\mddefault}{\updefault}{\color[rgb]{0,0,0}$v_i$}%
}}}}
\put(13507,-423){\makebox(0,0)[lb]{\smash{{\SetFigFont{12}{14.4}{\rmdefault}{\mddefault}{\updefault}{\color[rgb]{0,0,0}$v_j$}%
}}}}
\put(13542, 39){\makebox(0,0)[lb]{\smash{{\SetFigFont{12}{14.4}{\rmdefault}{\mddefault}{\updefault}{\color[rgb]{0,0,0}$y_{jj'}$}%
}}}}
\put(10726,1739){\makebox(0,0)[lb]{\smash{{\SetFigFont{12}{14.4}{\rmdefault}{\mddefault}{\updefault}{\color[rgb]{0,0,0}$S_k$}%
}}}}
\put(10876,-2611){\makebox(0,0)[lb]{\smash{{\SetFigFont{12}{14.4}{\rmdefault}{\mddefault}{\updefault}{\color[rgb]{0,0,0}$S_{k'}$}%
}}}}
\put(8026,389){\makebox(0,0)[lb]{\smash{{\SetFigFont{12}{14.4}{\rmdefault}{\mddefault}{\updefault}{\color[rgb]{0,0,0}$\delta$}%
}}}}
\put(13726,389){\makebox(0,0)[lb]{\smash{{\SetFigFont{12}{14.4}{\rmdefault}{\mddefault}{\updefault}{\color[rgb]{0,0,0}$\delta$}%
}}}}
\put(10876,389){\makebox(0,0)[lb]{\smash{{\SetFigFont{12}{14.4}{\rmdefault}{\mddefault}{\updefault}{\color[rgb]{0,0,0}$\delta$}%
}}}}
\put(8326,-1261){\makebox(0,0)[lb]{\smash{{\SetFigFont{12}{14.4}{\rmdefault}{\mddefault}{\updefault}{\color[rgb]{0,0,0}$\delta$}%
}}}}
\put(11176,-1261){\makebox(0,0)[lb]{\smash{{\SetFigFont{12}{14.4}{\rmdefault}{\mddefault}{\updefault}{\color[rgb]{0,0,0}$\delta$}%
}}}}
\put(14026,-1261){\makebox(0,0)[lb]{\smash{{\SetFigFont{12}{14.4}{\rmdefault}{\mddefault}{\updefault}{\color[rgb]{0,0,0}$\delta$}%
}}}}
\put(5326,-436){\makebox(0,0)[lb]{\smash{{\SetFigFont{12}{14.4}{\rmdefault}{\mddefault}{\updefault}{\color[rgb]{0,0,0}$v_2$}%
}}}}
\put(3001,-436){\makebox(0,0)[lb]{\smash{{\SetFigFont{12}{14.4}{\rmdefault}{\mddefault}{\updefault}{\color[rgb]{0,0,0}$v_1$}%
}}}}
\put(15901,-436){\makebox(0,0)[lb]{\smash{{\SetFigFont{12}{14.4}{\rmdefault}{\mddefault}{\updefault}{\color[rgb]{0,0,0}$v_n$}%
}}}}
\put(4371,-856){\makebox(0,0)[lb]{\smash{{\SetFigFont{12}{14.4}{\rmdefault}{\mddefault}{\updefault}{\color[rgb]{0,0,0}$b_1$}%
}}}}
\put(13834,-868){\makebox(0,0)[lb]{\smash{{\SetFigFont{12}{14.4}{\rmdefault}{\mddefault}{\updefault}{\color[rgb]{0,0,0}$\overline{y}_{jj''}$}%
}}}}
\end{picture}%